%% file: __arxiv_main.tex
\documentclass[opre]{./opre_style/informs4-arxiv}

\OneAndAHalfSpacedXI 

\usepackage{endnotes}
\let\footnote=\endnote
\let\enotesize=\normalsize
\def\notesname{Endnotes}%
\def\makeenmark{\hbox to1.275em{\theenmark.\enskip\hss}}
\def\enoteformat{\rightskip0pt\leftskip0pt\parindent=1.275em
  \leavevmode\llap{\makeenmark}}

\usepackage{caption}
\usepackage{subcaption}
\usepackage{amsmath}
\usepackage{amssymb}
\usepackage{bm}
\usepackage{hyperref}
\usepackage{thm-restate}

\input{mycommands}

\newcommand{\capx}{C_{\textup{apx}}}
\renewcommand{\chmades}[1]{{#1}}
\newcommand{\removeforarxiv}[1]{}

\usepackage{natbib}
 \bibpunct[, ]{(}{)}{,}{a}{}{,}%
 \def\bibfont{\small}%
\TheoremsNumberedThrough     
\ECRepeatTheorems

\EquationsNumberedThrough    

\begin{document}


\RUNAUTHOR{Williams, Wang, Harchol-Balter}

\RUNTITLE{{A}n {U}pper {B}ound on the {M}/{M}/k {Q}ueue {W}ith {D}eterministic {S}etup {T}imes}

\TITLE{{A}n {U}pper {B}ound on the {M}/{M}/k {Q}ueue {W}ith {D}eterministic {S}etup {T}imes}

\ARTICLEAUTHORS{%
\AUTHOR{Jalani Williams}
\AFF{Electrical Engineering and Computer Science Department, University of Michigan, Ann Arbor, MI 48109, \EMAIL{jalaniw@umich.edu}}
\AUTHOR{Weina Wang}
\AFF{Computer Science Department, Carnegie Mellon University, Pittsburgh, PA 15213, \EMAIL{weinaw@cs.cmu.edu}}
\AUTHOR{Mor Harchol-Balter}
\AFF{Computer Science Department, Carnegie Mellon University, Pittsburgh, PA 15213, \EMAIL{harchol@cs.cmu.edu}}
} 

\ABSTRACT{%
\input{0_abstract}
}%



\maketitle

%




\section{Introduction}\label{sec:intro}

\input{1_intro_revised}

\section{Model: The M/M/k/Setup-Deterministic Queue}\label{sec:model}
\input{2_model_new}

\section{Related Work}\label{sec:relatedworks}

\input{3_relatedworks_new}

\section{Results}\label{sec:results} 
\input{4_results}

\section{Key Ideas and Techniques}\label{sec:keyideas}
\input{5_keyideas}

\section{Proof of Theorem~\ref{thm:main}'s Three Main Lemmas}\label{sec:proof}

\input{6_proof}

\chmades{
\section{Proof of Multiplicative Tightness}\label{sec:tightness}
\input{claims/tightness}

}

\section{Conclusion and Future Work}\label{sec:conclusion}
\input{7_conclusion}

\bibliographystyle{./opre_style/informs2014}
\bibliography{setup}

\input{8_appendix}

\end{document}

%% file: mycommands.tex
\newcommand{\Poisson}{\text{Poisson}}

\newcommand{\ex}[1]{\mathbb{E}\left[ #1 \right]}
\newcommand{\pr}[1]{\Pr \left( #1 \right)}

\newcommand{\maxpar}[1]{\max \left \{ #1 \right \}}

\newcommand{\dd}{\textrm{d}}

\newcommand{\setuptime}{\beta}
\newcommand{\indc}[1]{\mathbf{1}_{#1}}
\newcommand{\indbig}[1]{\mathbf{1}\left\{ #1 \right\}}

\newcommand{\minpar}[1]{\min \left \{ #1\right \} }

\newcommand{\comp}{\textrm{comp}}
\newcommand{\tup}{\tau_\text{up}}

\newcommand{\s}{\mathcal{S}}

\newcommand{\cycle}{X}

\newcommand{\mintwo}[2]{\min\left( #1, #2 \right)}
\newcommand{\prise}[1]{p_{\text{rise}}^{(#1)}}
\newcommand{\threshone}{C_3 \shs}
\newcommand{\dcouple}{\Tilde{d}}

\newcommand{\upind}{n_u}
\newcommand{\rem}[1]{\dets_{#1}}
\newcommand{\ntar}{\left[ N \left(T_A \right) - R \right]}
\newcommand{\filt}[1]{\mathcal{F}_{#1}}
\newcommand{\dgen}{d_{\textup{gen}}}

\newcommand{\taul}{\tau + \ell}
\newcommand{\abs}[1]{\left | #1 \right |}

\newcommand{\ceti}{\Tilde{\eta}_{i}}

\newcommand{\tdown}{\tau_{\textup{down}}}
\newcommand{\dnb}[1]{v^{(\textup{down})}_{#1}}
\newcommand{\unb}[1]{v^{(\textup{up})}_{#1}}

\newcommand{\torp}[2]{\texorpdfstring{#1}{#2}}

\newcommand{\on}{\text{On}}
\newcommand{\cyclem}{X}
\newcommand{\lext}{L_{\text{ex}}}
\newcommand{\poss}{\tau_{\text{poss}}}

\newcommand{\chmade}[1]{{#1}}
\newcommand{\chmades}[1]{{\color{red}#1}}

%% file: 0_abstract.tex
In many systems, servers do not turn on instantly; instead, a \textit{setup time} must pass before a server can begin work.
These ``setup times’’ can wreak havoc on a system's queueing; this is especially true in modern systems, where servers are regularly turned on and off as a way to reduce operating costs (energy, labor, $CO_2$, etc.).
To design modern systems which are both efficient \emph{and} performant, we need to understand how setup times affect queues.

Unfortunately, despite successes in understanding setup in a single-server system, setup in a multiserver system remains poorly understood.
To circumvent the main difficulty in analyzing multiserver setup, all existing results assume that setup times are memoryless, i.e. distributed {E}xponentially.
However, in most practical settings, setup times are close to Deterministic, and the widely used Exponential-setup assumption leads to unrealistic model behavior and a dramatic underestimation of the true harm caused by setup times.

This paper provides a comprehensive characterization of the average waiting time in a multiserver system with \emph{Deterministic} setup times, the {M}/{M}/k/{S}etup-{D}eterministic.
In particular, we derive upper and lower bounds on the average waiting time in this system, and show these bounds are within a multiplicative constant of each other.
These bounds are the first closed-form characterization of waiting time in any finite-server system with setup times.
Further, we demonstrate how to combine our upper and lower bounds to derive a simple and accurate approximation for the average waiting time. 
These results are all made possible via a new technique for analyzing random time integrals that we named the {M}ethod of {I}ntervening {S}topping {T}imes, or MIST.

%% file: 1_intro_revised.tex
\newcommand{\ttarget}{T^{\textup{target}}}
\newcommand{\ssres}{T^{\textup{single-server}}}

\subsection{What Are Setup Times?}

In many systems, servers do not turn on instantly; instead, a \textit{setup time} must pass before a server can begin work \citep{allahverdi2008significance}.
For example, for applications hosted in the cloud, application replicas must take time to boot up before they can begin fulfilling requests \citep{rzadca2020autopilot};
for overwhelmed hospitals, traveling nurses must wait to have their credentials confirmed before they can begin helping patients \citep{tuttas2013travel};
for many businesses, workers must go through a long and expensive recruitment/onboarding process before they can begin serving customers \citep{behroozi2020debugging}.
By thinking about this ``initial delay before service'' as an abstract \textit{setup time}, we can learn how setup time affects all of these systems simultaneously \citep{allahverdi2008significance}.
        
\subsection{Why Do Setup Times Matter?}
Setup times can have a significant impact on a system's queueing behavior, especially in modern systems.
For systems which keep their servers \textit{on} all the time, clearly setup times do not affect their performance.
However, in many modern systems, servers are regularly turned on and off.
Because servers don't turn on instantly, jobs in a system with setup times end up delayed compared to their no-setup counterparts.
If setup times are long enough, this additional delay can be significant.

Nevertheless, many systems still regularly turn their servers on and off.
Why? Because, by doing so, they can save a considerable amount on operating costs (e.g. energy, money, $CO_2$, etc.) \citep{rzadca2020autopilot}.
However, this cost-saving measure is only a viable option if the additional delay caused by setup times is not too large.
Therefore, if we want to design systems which are simultaneously efficient \emph{and} performant, we need a good understanding of how setup times affect queueing performance.



\subsection{Prior Art on Understanding the Effect of Setup Times}

\paragraph{State of the art.} Unfortunately, despite continued academic interest, we still struggle to understand the impact of setup times on customer wait times outside of a few very simple settings.
In the single server setting, \citep{welchfirstservice} completely characterized the behavior of the waiting time under extremely general conditions, i.e. in the M/G/$1$ queue with generally-distributed setup times.
In the more complex multiserver setting, little progress was made until the publication of \citep{RRR} almost half a century later.
Outside of the fact that the model of \citep{RRR} has multiple servers, their model is much simpler than that of \citep{welchfirstservice}. 
In particular, the main results of \citep{RRR} rely on the assumption that both job service times and server setup times are distributed Exponentially, i.e. the authors study the M/M/k/Setup-Exponential model.

\paragraph{Limitations of the Exponential model.} Despite the fact that \citep{RRR} represented the first breakthrough in our understanding of the setup effect in 50 years, their results are limited in two significant ways. 
First, instead of a closed-form formula for the average waiting time, the authors only derive an algorithm for computing the average waiting time. 
This algorithm is useful in the sense that it bypasses the need to simulate the system, but unfortunately fails to give intuition about how wait times scale with system parameters.
Second, one of their simplifying assumptions, that setup times are distributed Exponentially, turns out to severely limit the utility of their results.
We discuss this point in depth within Section~\ref{sec:limitExp}, but, put briefly, this lack of utility stems from the fact that setup times have relatively low variance \citep{mao2012performance}, i.e. they are closer to \textit{Deterministic} than Exponential.
Accordingly, we analyze the Deterministic version of the \citep{RRR} model.

\paragraph{Previous work on the Deterministic model.}
At the moment, only one other result characterizing the average waiting time in the M/M/k/Setup-Deterministic exists: \citep{WHW}. 
The paper \citep{WHW} proves the first-ever lower bound on the average waiting time in the M/M/k/Setup-Deterministic queue.
\citep{WHW} differs from this paper in a few ways. 
First, \citep{WHW} proves \emph{only} a {lower bound} on the average waiting time, whereas we prove the first-ever upper bound as well as a nearly matching lower bound.
Furthermore, our lower bound is significantly stronger.
To be more precise, \citep{WHW} incorporates additional assumptions to prove their result, and the final result there becomes trivial in certain scaling regimes. 
In contrast, our proofs only assume 1) that setup times are significantly larger than service times and 2) that enough servers are being regularly utilized, and we prove explicitly that our lower bound and upper bound are separated by at most a multiplicative constant.

\begin{figure}
    \centering
    \includegraphics[scale=0.6]{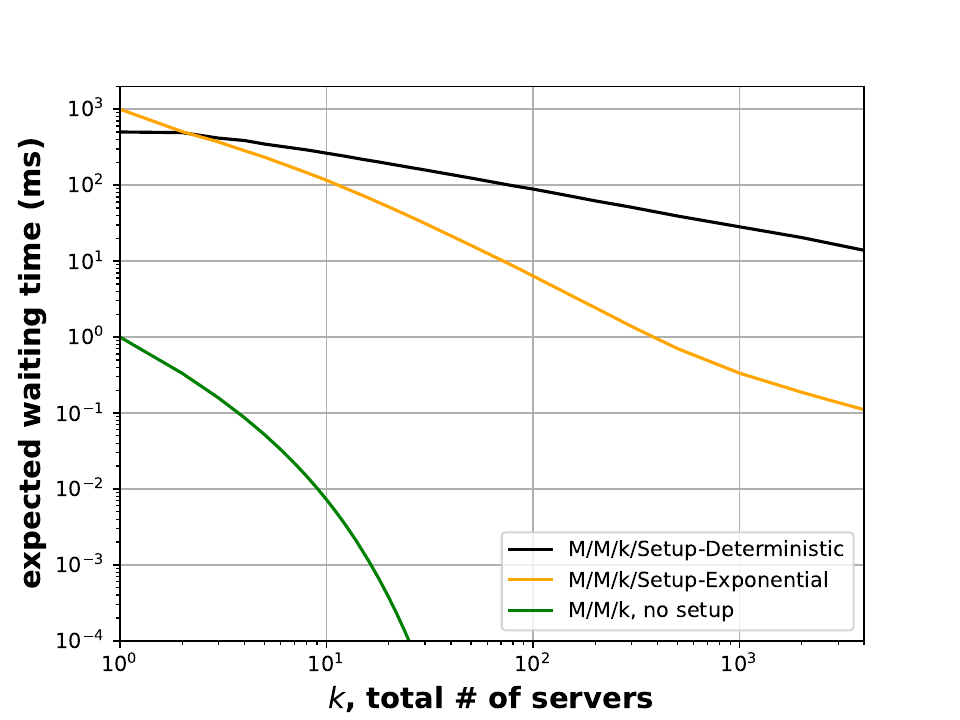}
    \caption{Simulation results for the M/M/k/Setup-Deterministic, M/M/k/Setup-Exponential, M/M/k (no setup), with mean service time $\frac{1}{\mu}=1$ ms, mean setup time $\beta = 1000$ ms, and load kept at a constant $\rho = 0.5$. Note the high separation between the Exponential and Deterministic models at large scales.}
    \label{fig:expVdet}
\end{figure}

\subsection{Limitations of the Exponential model}\label{sec:limitExp}

\paragraph{The Exponential model is unrealistic.} 
As previously mentioned, the Exponential assumption of \citep{RRR} turns out to be extremely problematic; we highlight two particular issues.
First, this assumption leads to unrealistic model behavior.
To illustrate where the breakdown in realism happens, consider a scenario where only a single server is setting up and compare it to a scenario where 100 servers begin setup at the same time.
In the Deterministic model, all 100 servers perform no work while they are setting up, then all of them turn on simultaneously.
By contrast, in the Exponential model, the 100-server system receives its first server on average 100 times faster than the single-server system receives its first server.
In other words, in the Exponential model, the \textit{longer} the system's queue is, the \textit{more rapidly} the system's servers turn on to help \textit{drain} that queue. 
In a sense, the Exponential system can rapidly ``react'' to increases in queue length.
\chmade{ Note that this ``blessing of variability'' effect only comes into play when there are multiple servers; in the single server case, as seen in Figure~\ref{fig:expVdet}, the variability in setup time hurts the waiting time.}

\paragraph{The Exponential model underestimates waiting.} This unrealistic ``reactivity'' phenomenon causes a further, more concerning, problem: in real systems, the Exponential model dramatically underestimates how much waiting actually occurs.
To be more precise, in modern systems: 1) average setup times are often larger than average job sizes by two or three orders of magnitude \citep{gandhi2012autoscale, mao2012performance,mogul2015inferring, maheshwari2018scalability, hao2021empirical};
and 2), as noted in the state-of-the-art paper \citep{RRR}, setup times are actually closer to \emph{Deterministic}. 
When these two criteria are satisfied, as observed in Figure~\ref{fig:expVdet}, the true waiting time is often orders of magnitude larger than what the Exponential model predicts.
Accordingly, in many practical studies of the setup effect \citep{gandhi2012autoscale,kara2017energy,hyytia2018dynamic}, setup times are assumed to be Deterministic, e.g. servers take a fixed time of 2 minutes to set up. 
However, despite its apparent practical limitations, the Exponential setup model remains the \textit{de facto} choice for theoretical analysis, since it allows for the application of a number of existing theoretical techniques.
\chmades{By contrast, we craft our techniques to specifically address the case where setup times are both Deterministic and long.}

\subsection{Our Results}

\chmades{
\paragraph{Model Summary.} In this paper, for the first time, we provide a non-asymptotic analysis of the average waiting time in a dynamically-provisioned multiserver system with Deterministic setup times. A detailed description of this model is given in Section~\ref{sec:model}, but we give a brief description here.
We study an M/M/k queueing system with load $\rho$ and mean service time $\frac{1}{\mu}$.
Unlike the classical M/M/k model, in our setting, a server turns off when it completes serving a job and the queue is empty.
When a new job arrives into an empty queue and there are off servers, one of them is turned on, experiencing a deterministic setup time of duration $\beta$.
This model can be viewed as a Deterministic analogue of the one studied in \citep{RRR}, and we refer to it as the M/M/k/Setup-Deterministic model.
}
\paragraph{\chmades{Main Results.}}
We derive three main results on the average waiting time (queueing time) of jobs, $\ex{T_Q}$; \chmades{we plot these results in Figure~\ref{fig:boundComparison}.}
Our first result is Approximation~\ref{apx:only}, a prediction of the average waiting time \chmades{$\ex{T_Q}$}, which appears extremely accurate across a wide variety of parameter settings (see Figure~\ref{fig:approxEval} for a detailed visualization).
When simplified, this closed-form approximation roughly states that
\begin{equation*}
    \ex{T_Q} \approx \frac{1}{2} \sqrt{\frac{\pi}{2}} \frac{\beta}{\sqrt{k \rho}} + \frac{1}{\mu k(1-\rho)},
\end{equation*}
\chmades{where $\rho k$ represents the offered load and $\frac{\beta}{1/\mu}$ represents the ``relative setup time,'' namely the setup as a proportion of the job size.}
In support of this approximation, in our second and third results, we prove upper and lower bounds on the average waiting time \chmades{which hold when both the offered load $k\rho$ and the relative setup time $\frac{\beta}{1/\mu}$ are large ($\geq 100$)}; these make up Theorems~\ref{thm:main} and~\ref{thm:improvement}, respectively.
Furthermore, we later show, in Theorem~\ref{thm:tightness}, that our results capture the correct order-wise scaling of the average waiting time---they agree, up to multiplication by an absolute constant.
Note that, at the moment, no other such characterizations exist for any similar multiserver setup system; not even for the extensively-studied M/M/k/Setup-Exponential model.

\begin{figure*}[t!]
    \centering
\begin{subfigure}[t]{0.495 \linewidth}
    \centering
    \includegraphics[width =\textwidth,clip]{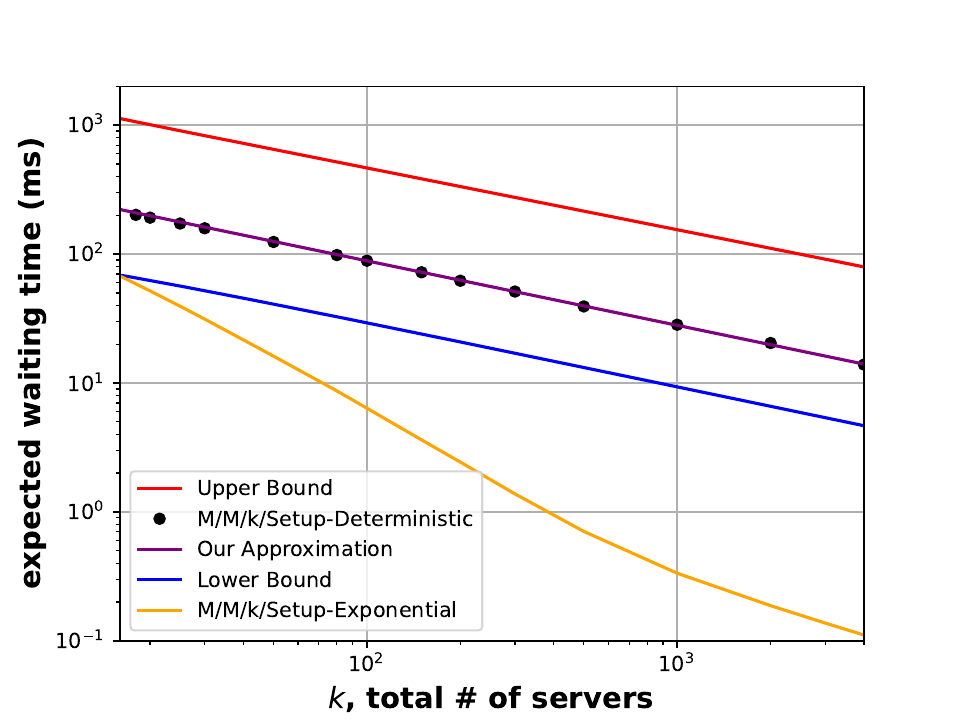}
    \caption{\footnotesize An illustration of our results.}
    \label{fig:boundComparison}
\end{subfigure}
    \hfill
    \begin{subfigure}[t]{0.495 \linewidth}
    \centering
    \includegraphics[width=\textwidth,clip]{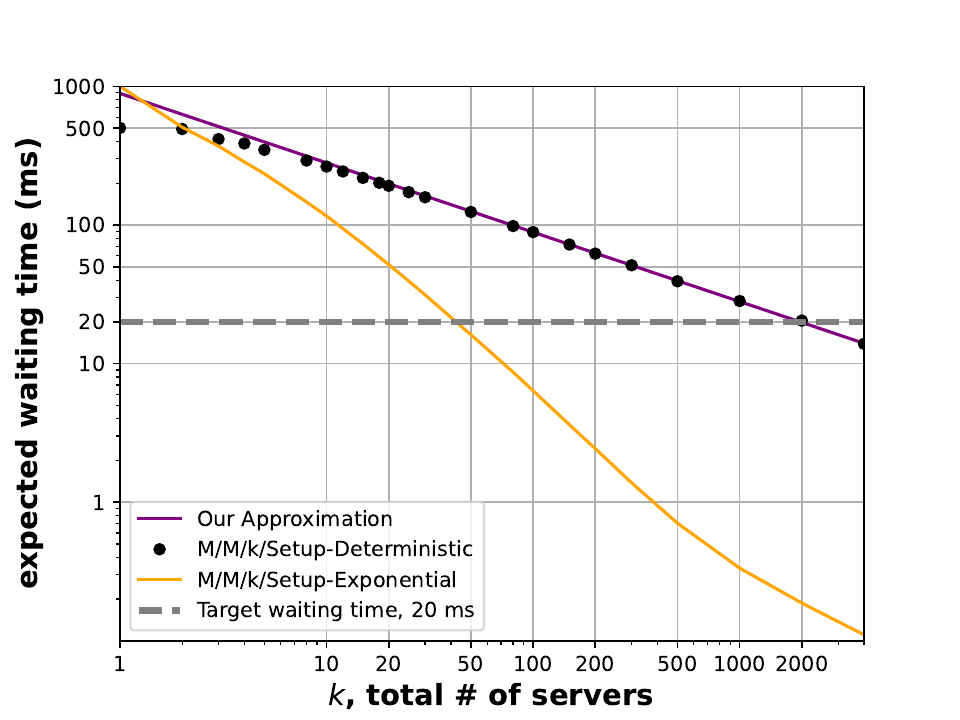}
    \caption{\footnotesize A provisioning example.}
    \label{fig:prov_ex}
    \end{subfigure} 
    \vspace{3pt}
    \caption{Our theoretical results along with simulation data for the M/M/k/Setup-Deterministic and M/M/k/Setup-Exponential, varying the number of servers $k$ while keeping the mean service time $\frac{1}{\mu} =1$ ms, the mean setup time $\beta = 1000$ ms, and the load $\rho = 0.5$ fixed. 
    (a) A comparison of our results to the true average waiting time in the M/M/k/Setup-Deterministic. Our results behave like the true average waiting time, while the Exponential model behaves differently.
    (b) A provisioning example highlighting the differences between the Deterministic and Exponential models. To achieve a target waiting time of $20$ ms, our approximation correctly predicts it will take $k\approx 2000$ servers, while the Exponential model predicts that only $k\approx50$ servers should suffice. See Figure~\ref{fig:approxEval} for a more comprehensive evaluation.}
\end{figure*}


\subsection{Impact:  How Understanding Setup Times Helps with Capacity Provisioning}\label{sec:impact}

A common but complex problem which arises in many areas is that of designing a system such that the average waiting time of a customer is below some target waiting time.
Historically, we understand this problem well for systems without setup times, e.g. there’s a straightforward formula for the average waiting time in the M/M/k without setup.
Unfortunately, our understanding of this problem is quite poor for more modern systems, since their average waiting times are affected by setup times. 
In particular, modern systems dynamically control the number of servers that they keep on, periodically turning servers off in order to save energy.
As we mentioned before, previous results on understanding the relationship between setup times and the average waiting time leave much to be desired.
Our new results expand on the state-of-the-art Exponential model in two important ways: 1) obtaining the predicted average waiting time is much easier computationally, and 2) the quality of the prediction is much better.

\paragraph{Easier predictions.}
Compared to the Exponential model, our new Deterministic approximation greatly simplifies the design process.
In particular, when predicting the average waiting time in the Exponential model using the state-of-the-art method from \citep{RRR}, one must solve a system of $O(k^2)$ quadratic equations to find the average waiting time $\ex{T_Q}$.
Two practical issues arise from this fact.
First, the equations change depending on the number of servers $k$, meaning that the computation must be repeated every time one wishes to test a new number of servers.
Second, the opacity of the process makes it difficult to get intuition about how the average waiting time changes as one alters the system parameters.
In contrast, Approximation~\ref{apx:only} is a relatively simple function of the relevant parameters.
The simplicity of our approximation has, likewise, two benefits: 1) computing the waiting time becomes easy, and 2) our approximation's form makes it clear how and why the waiting time behaves the way it does.


\paragraph{Higher quality predictions.}
Moreover, when compared to the predictions of the Exponential model, the predictions we obtain using our Deterministic approximation are of a much higher quality.
This difference in quality is perhaps best illustrated by looking at a simple example.
In Figure~\ref{fig:prov_ex}, we compare the prediction from the Exponential model to the prediction from our approximation, plotting how the predicted average waiting time changes as one increases the number of servers $k$ while fixing the load $\rho=0.5$, the average setup time $\beta = 1000$ ms, and the average service time $\frac{1}{\mu} = 1$ ms (note that these are typical relative values in many applications \citep{mao2012performance}).
Our goal is to determine how large the number of servers $k$ needs to be before we reach our target waiting time $\ttarget = 20$ ms.
In both models, the average waiting time decreases as the system gets larger.
However, the Exponential model predicts that the average waiting time will be small enough once $k=50$. 
On the other hand, as captured by our approximation, the Deterministic setup system will only reach the target waiting time once the number of servers $k \approx 2000$ ---a full $40$ times larger than what the Exponential system predicts!
At even a modest number of servers, the Exponential system underestimates the waiting time by orders of magnitude.

\subsection{Challenges and Our Approach}
\paragraph{What makes the multiserver setting difficult.}
Although the harm caused by setup times is non-trivial to understand even in single server systems, the setup effect can be especially difficult to understand when multiple servers can set up at the same time.
In particular, when servers can set up \emph{simultaneously}, their server states begin to interact; via the speed of their processing, the system's \textit{busy} servers indirectly control the setup behavior of the system's \textit{not-busy} servers.
For example, if server A is \textit{on} while server B is \textit{setting up}, then server A might finish all the work in the queue before server B even has a chance to turn \textit{on}.
As such, in the multiserver setting, it can sometimes make sense to \textit{cancel} a server's setup process; a situation which would \emph{never} occur in the single server setting.
Of course, the opposite can also happen: if the busy servers are working much more slowly than expected, then the queue might grow large enough that we begin \textit{setting up} a server that would otherwise be left \textit{off}.
This interaction between departure behavior and setup behavior is exactly what makes the setup effect so much harder to understand in the multiserver setting.

\paragraph{Why Deterministic setup is especially difficult.} As such, while using Deterministic setup times might be more realistic, it also comes with a set of unique theoretical challenges.
In the Deterministic case, there is no avoiding the complexity of setup: even in simulation, one must track the individual remaining setup time of \emph{every} server that is currently setting up.
By contrast, because the Exponential distribution is \emph{memoryless}, in the Exponential case it suffices to track only the \textit{total number} of servers setting up, greatly simplifying the system state.
Moreover, the Exponential model's simple state forms a Continuous-Time Markov Chain, a well-studied class of stochastic processes for which a number of techniques have been developed.
For the Deterministic setup model, no such techniques exist.

\paragraph{Our approach.}
To derive our first-of-their-kind results, we develop a new technique called the Method of Intervening Stopping Times, or MIST, which allows us to more naturally investigate queueing systems with high-dimensional, continuous-valued state spaces.
Generally, MIST is a method for bounding the expectation of a random time integral, 
and it works by dividing up the original time interval into smaller, more manageable pieces using stopping times.
By focusing on these smaller pieces, we avoid the issue of tracking the entire high-dimensional system state at every point in time. 
Instead, MIST allows us to narrow our focus to only the ``most important’’ aspects of the state for a particular integral, and to shift our definition of ``most important’’ based on the particular integral we are analyzing. 
Since, for a particular integral, we know the most important aspects of the system state, we can derive strong bounds on that integral that hold \textit{regardless} of the remaining state information.
Moreover, we can apply this approach recursively, chopping a still-large time interval into a sequence of even smaller time intervals.
By chopping time in the right way, eventually we can make the system so well-behaved during an interval that we can adapt powerful techniques like martingale theory and Wald’s equation to the analysis of a phase.
\chmades{To place our approach in a broader context, MIST is reminiscent of a suggestion made in \cite{kingman2009first} that
    ``it may be possible to sew these small martingales ... together to form a martingale on the whole line,''
    at which point one can make use of ``all the apparatus of martingale theory---stopping identities, inequalities, central limit theorems, and the like.''
MIST does exactly this: it sews together martingales in useful ways. For more details, see Section~\ref{sec:keyideas}.
}

    
\chmade{\subsection{Potential Mitigation}
The delay in serving jobs due to setup times comes from turning servers off and then needing to wait before those servers can turn on.
One way to potentially mitigate the impact of setup times is then to turn off servers less aggressively.
We analyze such an approach in Appendix~\ref{sec:extension}, and our analysis reveals that a natural class of policies have minimal impact on the waiting time.
This preliminary analysis highlights the need for further studies on designing policies for turning servers on and off in queueing systems with setup times to achieve a better tradeoff between energy efficiency and performance.
}

\subsection{Outline}
We proceed in the following order:
In Section~\ref{sec:model}, we give a detailed description of our model.
In Section~\ref{sec:relatedworks}, we discuss the previous work on analyzing setup times in greater detail.
In Section~\ref{sec:results}, we fully state our main results.
In Section~\ref{sec:keyideas}, we discuss the key ideas which enabled us to prove our main results, including MIST.
In Section~\ref{sec:proof}, we give the proof of Theorem~\ref{thm:main}; we relegate the proof of Theorem~\ref{thm:improvement} to Appendix~\ref{sec:improvement}.
\chmades{In Section~\ref{sec:tightness}, we give a proof of the multiplicative tightness of our bounds.}
Finally, in Section~\ref{sec:conclusion} we make some closing remarks and propose follow-up work based on our findings.

%% file: 2_model_new.tex
\newcommand{\brc}[1]{\left\{ #1 \right \}}
\newcommand{\state}{S}
\newcommand{\dets}{Y}
\newcommand{\detsvec}{\bm{\dets}}
\newcommand{\dept}{D}
\newcommand{\interv}{\mathcal{I}}
\newcommand{\depOp}[2]{\mathcal{D}\left[#1 \right]\left( #2 \right)}
\newcommand{\bp}[2]{T^{\textup{busy}}\left(#1, #2 \right)}
\newcommand{\ibp}[2]{I^{\textup{busy}}\left(#1, #2 \right)}

\paragraph{The system behavior, excluding setup.} We now describe our model of interest, the M/M/k/Setup-Deterministic.
As in the typical M/M/k queue, jobs arrive in a Poisson process of rate $k\lambda$ into a FCFS queue where jobs wait to be served at one of $k$ servers. 
The job at the head of queue enters service whenever a server frees up, either from a job completing service or from a server finishing set up. 
Once a job enters service, it remains in service for Exp$(\mu)$ time before departing.
We assume all the servers have identical service and setup \chmade{distributions}.
As such, we can assign each server an index from $1$ to $k$, and without loss of generality assume that departures always occur at the busy server with the highest index;
i.e., we re-index the servers when a job departs so the server with the newly departed job has the highest index among the busy servers.
From here, we define the quantity $Z(t)$ to be the number of busy servers (or jobs in service) at time $t$, the quantity $Q(t)$ to be the number of jobs waiting in the queue at time $t$, and the quantity $N(t)=Q(t) + Z(t)$ to be the total number of jobs in our system.
Excluding the setup dynamics, one sees that, as promised, the behavior of our model is identical to the M/M/k queue.

\begin{figure}
    \centering
    \includegraphics[scale=0.25]{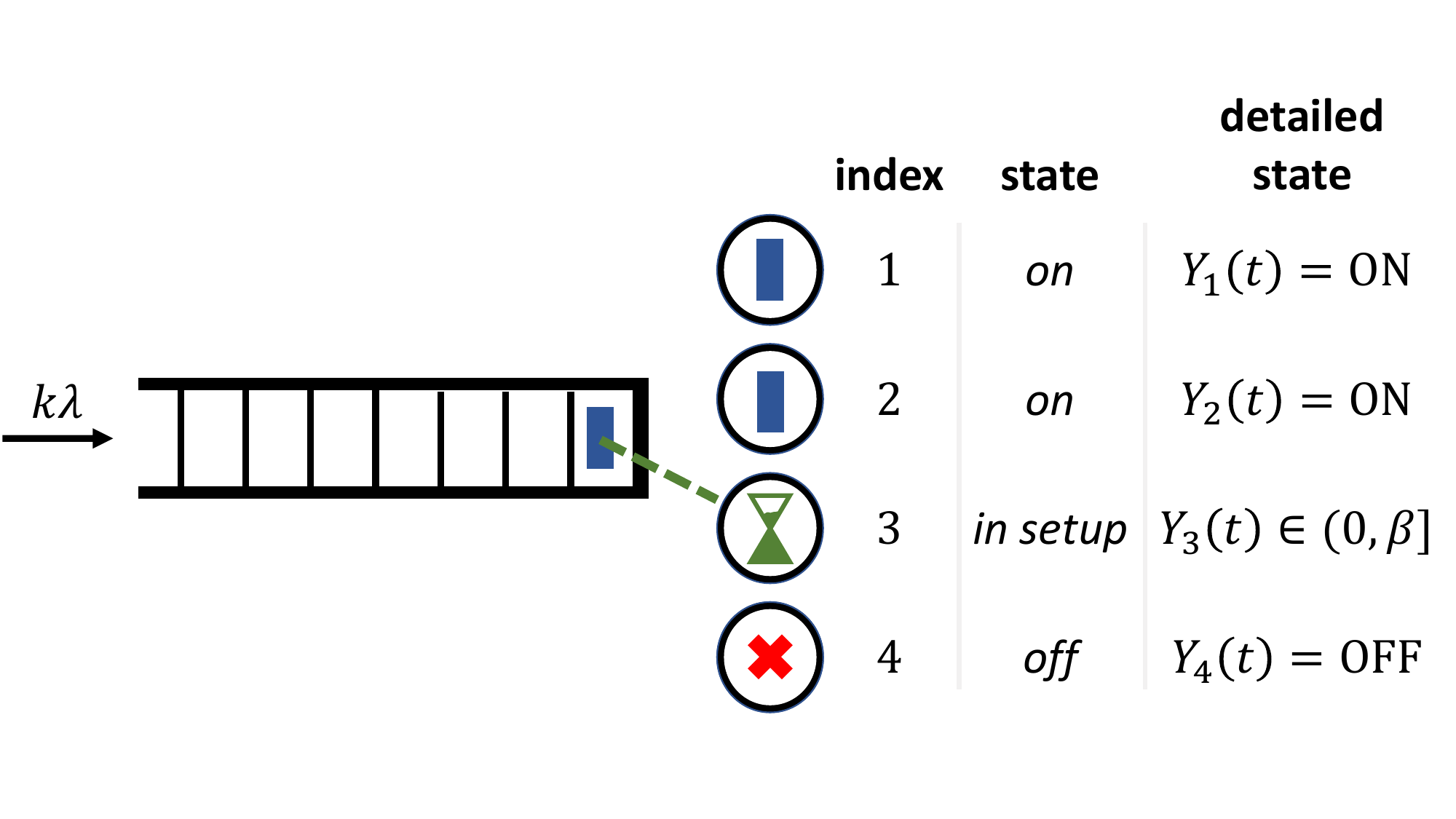}
    \vspace{-3pt}
    \caption{An example of M/M/k/Setup-Deterministic with $k=4$. The state pictured has $Z(t)=2$ busy servers, which means there are $2$ jobs in service. There is $Q(t)=1$ job in queue, and thus $N(t)=Z(t)+Q(t)=4$ jobs in system.}
    \label{fig:model}
\end{figure}

\paragraph{The setup dynamics.} From here, it suffices to describe precisely how servers will be turned \textit{on} and \textit{off}.
We assume that each server is always in one of three states: \textit{on}, \textit{off}, or \textit{in setup}. 
A given server remains \textit{on} only as long as that server remains busy. In other words, a server turns \textit{off} when it finishes its current job and the queue is empty.
Recall that we re-index the servers to ensure that the job departure is from the highest indexed server among the busy servers.
On the other hand, server $i$ begins \textit{setup} when a job arrives to the system and there are only $i-1$ jobs in the system. 
Server $i$ remains in setup until one of two events occurs: either 1) some fixed quantity $\setuptime$ time has passed, or 2) there are fewer than $i$ jobs in the system; accordingly, we refer to $\beta$ as the \textit{setup time} of a server.
In the first case, if $\beta$ time has passed without $N(t)$ dipping below $i$, then server $i$ has completed its setup and begins working on the job at the head of the queue. 
In the second case, if the number of jobs $N(t)$ dips below $i$ before server $i$ completes setup, then the setup is canceled and server $i$ turns \textit{off}.
We use $\dets_i(t)$ to denote the detailed state of server $i$ at time $t$. 
If server $i$ is \textit{off}, we set $\dets_i(t) = \text{OFF}$; if server $i$ is \textit{on}, we set $\dets_i(t) = \text{ON}$; if server $i$ is \textit{in setup}, we let $\dets_i(t)$ denote the remaining amount of time until server $i$ would finish setup, if left uninterrupted. 
To be precise, $\dets_i(t)$ is set to $\beta$ when server $i$ first initiates \textit{setup}, and this value decreases at rate $1$ until either setup completes or setup is canceled.
For convenience, we assume, without loss of generality, that $\text{ON}< s < \text{OFF}$ for every possible remaining setup time $s \in (0, \beta]$; this ensures that the detailed state $\dets_i(t)$ is non-decreasing in $i$.
As a shorthand, we use $\detsvec = \left(\dets_1(t), \dets_2(t), \dots, \dets_k(t) \right)$ to denote the vector of detailed server states.

\paragraph{ A state descriptor.} Accordingly, a Markovian state descriptor for our system at time $t$ is $\state(t) \triangleq \left(N(t), \detsvec(t)\right)$. Note that, since one can recover the number of jobs in service $Z(t)$ from the detailed server states $\detsvec(t)$, one could also choose the state to be $\left(Q(t), \detsvec(t) \right)$. Either suffices in providing a complete description of the forward dynamics of the system.  Furthermore, when discussing the steady-state distribution of, say, the number of jobs $N(t)$, we use the notation $N(\infty)$.

\paragraph{Some important constants.} We define some system parameters which are critical to system behavior. We use $\rho \triangleq  \frac{\lambda}{\mu}$ to refer to the load of our system, i.e., the time-average fraction of servers working on a job. We call the offered load $R \triangleq k\rho$; this is the time-average \emph{number} of busy servers. To enforce stability, we require that $\rho < 1$. As discussed, the symbol $\setuptime$ refers to the fixed (Deterministic) setup time of a server.

\paragraph{Busy period notation.} 
Our results can be stated more concisely with two quantities related to a busy period of an M/M/1 queue.  We give the notation below.
We use $\bp{n}{j}$ to denote the expectation of the random \emph{length} of an $M/M/1$ busy period with arrival rate $k\lambda$, service rate $k\lambda + \mu j$, and which starts with $n$ jobs in the system. Likewise, we use $\ibp{n}{j}$ to denote expectation of the random \emph{time integral of the number of jobs} within the $M/M/1$ over the same period. Explicitly, we have
\begin{equation}
    \bp{n}{j} \triangleq \frac{n}{\mu j}
\end{equation}
and
\begin{equation}
    \ibp{n}{j} \triangleq \frac{n}{\mu j}\left[\frac{n+1}{2} + \frac{1}{1 - \frac{k\lambda}{k \lambda + \mu j}} \right] = \frac{n}{\mu j}\left[\frac{n+1}{2} + \frac{R}{j} + 1 \right].
\end{equation}

%% file: 3_relatedworks_new.tex

In Section~\ref{sec:intro}, we discussed two works most directly related to ours, \citep{RRR} and \citep{welchfirstservice}. 
Here, we briefly describe a few other works studying setup times in queueing systems.

\subsection{Single Server Setup}

\paragraph{The single server setting.} Since the seminal paper of \citep{welchfirstservice} on the M/G/1/Setup, various followup work has been devoted to extending this work to service disciplines beyond First-Come-First-Served \citep{bischof2001analysis} and arrival processes beyond Poisson \citep{he1995flow,choudhury1998batch}.

\paragraph{Staggered setup.} The first investigations into multiserver setup study the M/M/k/Setup and the M/G/k/Setup with the additional assumption that \chmade{only one server can be \textit{in setup} at a time} \citep{artalejo2005analysis,gandhiharchol13}.
However, in real systems, \chmade{servers can undergo the setup process \textit{simultaneously} (e.g. ten virtual machines can be booting up at the same time),} meaning that the strong results obtained for the staggered setup model are of limited practical use.

\subsection{Multiserver Setup}

\paragraph{M/M/k/Setup-Exponential: Approximations.} In \citep{gandhi2010server}, a precursor to \citep{RRR}, the authors derive a number of intuitive approximations for the average waiting time.
Further, \citep{pender2016law} analyzes an extended version of their model which includes customer abandonment and time-varying arrival rates, providing an approximation for the average queue length at each point in time.
Unfortunately, unlike our results, none of these results are in closed-form, nor do they provide a bound on the approximation error.

\paragraph{M/M/k/Setup-Exponential: Exact analysis.} In \citep{RRR}, the authors provide a method for computing the expectation of any function of the steady-state queue length; their method was later shown to be analogous to the matrix analytic method \citep{phung2017exact}.
As usual, though, these results are not in closed-form; even now, it remains unknown how the average waiting time in the M/M/k/Setup-Exponential model varies as one varies the system parameters.

\paragraph{Dispatching/Load-balancing.} In \citep{mukherjee2017optimal}, the authors consider a dispatched, finite-buffer version of the M/M/k/Setup-Exponential, for which they design TABS, an asymptotically-optimal (in energy waste, as $k \to \infty$) load-balancing/scaling scheme.
In a followup paper, \citep{mukherjee2019join}, the authors extend their work to the infinite buffer case, confirming its asymptotic optimality.

\paragraph{Deterministic setup.} Besides the lower bound of \citep{WHW}, a few others have investigated multiserver systems with Deterministic setup times. In the control setting, \citep{hyytia2018dynamic} considers a dispatching version of the M/G/2/Setup-Deterministic model, building near-optimal policies for the joint control of setup and dispatching. In \citep{kara2017energy}, the author provides a simulation-based analysis of the M/M/k/Setup-Deterministic queue, corroborating our findings on the large gap between Exponential and Deterministic setup systems.

\paragraph{Multiserver scheduling.} In  \citep{hong2023performance}, the authors study the performance of the Gittins policy in the G/G/k/Setup queue and, instead of explicitly bounding the average wait as we do, they bound the difference between Gittins' average waiting time and the optimal average waiting time.

\paragraph{Policies for mitigating the harm caused by setup.}

Many works have considered the problem of mitigating the energy waste and additional delay caused by setup times, often formalized in terms of the Energy-Response-time-Product (ERP) \citep{gandhiguptaharcholkuzuch10}, the Normalized-Performance-Per-Watt (NPPW) \citep{gandhiharchol11Allterton}, or the energy expended given a fixed tail cutoff for response time \citep{gandhi2012autoscale}.
Many works study the \emph{DelayedOff} policy, where a server must idle for a period before it may turn off \citep{gandhiguptaharcholkuzuch10, gandhi2010server, gandhi2012autoscale, pender2016law}. 
When using DelayedOff, the choice of which \textit{idle-but-on} server receives a job matters:
\citep{gandhiguptaharcholkuzuch10} considers routing the job to the \emph{Most Recently Busy} server; \citep{gandhi2012autoscale} creates a ranking of all servers, always routing to the lowest-ranked idle server.
Other works \citep{gandhiguptaharcholkuzuch10, gandhiharcholkozuch11} consider utilizing sleep states, useful intermediates between the ``no power, long setup'' \textit{off} state and the ``full power, no setup'' \textit{on} state.
\chmades{We note that the policy considered in this paper has been studied before, and is a natural generalization of a single-server policy known to be optimal in some settings. However, determining whether that optimality continues to hold in the multiserver setting remains a challenging open problem.}


%% file: 4_results.tex
\paragraph{Summary.} We now state our three main results concerning the average queue length $\ex{Q(\infty)}$: an approximation, an upper bound, and a lower bound. 
Note that, by Little's Law \citep{Kle_75}, which states that $\ex{Q(\infty)} = k\lambda \ex{T_Q}$, our results immediately translate to results on the average waiting time $\ex{T_Q}$; see Figure~\ref{fig:boundComparison} for an illustration. 
Our first main result, Approximation~\ref{apx:only}, is a simple and accurate approximation to the average queue length $\ex{Q(\infty)}$; we demonstrate its accuracy in Figure~\ref{fig:approxEval}.
Our second main result, the upper bound of Theorem~\ref{thm:main}, is the first-ever upper bound on the average queue length $\ex{Q(\infty)}$ in a multiserver system with Deterministic setup times. 
By contrast, our third main result, the lower bound of Theorem~\ref{thm:improvement}, is only the second lower bound on the queue length $\ex{Q(\infty)}$; however, improving on the lower bound of \citep{WHW}, we show that our bounds differ by at most a multiplicative constant. 
Notably, all three results are in explicit closed-form, unlikely nearly all other results concerning multiserver systems with setup times.

\chmade{
We begin by discussing Approximation~\ref{apx:only} in Section~\ref{sec:apx_statement}.
We then present and discuss our upper and lower bounds on the average queue length as Theorems~\ref{thm:main} and~\ref{thm:improvement}, respectively, in Section~\ref{sec:upper-lower-bounds-statements}.
Finally, we provide some justification of Approximation~\ref{apx:only} using the established upper and lower bounds in Section~\ref{sec:approx}.
}

\subsection{Approximation: Statement and Evaluation}\label{sec:apx_statement}

\begin{restatable}[Approximation of the Average Queue Length]{apprx}{ourapx}\label{apx:only}
    Let $\capx \triangleq \sqrt{\frac{\pi}{2}}$. Then, if the offered load $R\triangleq k \rho > 1$, one can approximate the steady state queue length as
    \begin{equation}\label{eq:full}
        \ex{Q(\infty)} \approx Q_{\textup{apx}} \triangleq 
        \frac{
            \frac{1}{2} \mu \beta^2 \capx \sqrt{R} + 
            \frac{\mu \beta \capx \sqrt{R}} {\mu k(1-\rho)}
            \left[
                \frac{\mu \beta \capx \sqrt{R} + 1}{2} + \frac{1}{1-\rho}
            \right]
        }
        {
            \beta + \frac{\mu \beta \capx \sqrt{R} }{\mu k (1-\rho)}
        }.
    \end{equation}

\end{restatable}

\begin{figure}
    \centering
    \includegraphics[ trim={0cm 0cm 0cm 0cm},clip, width=0.9\textwidth]{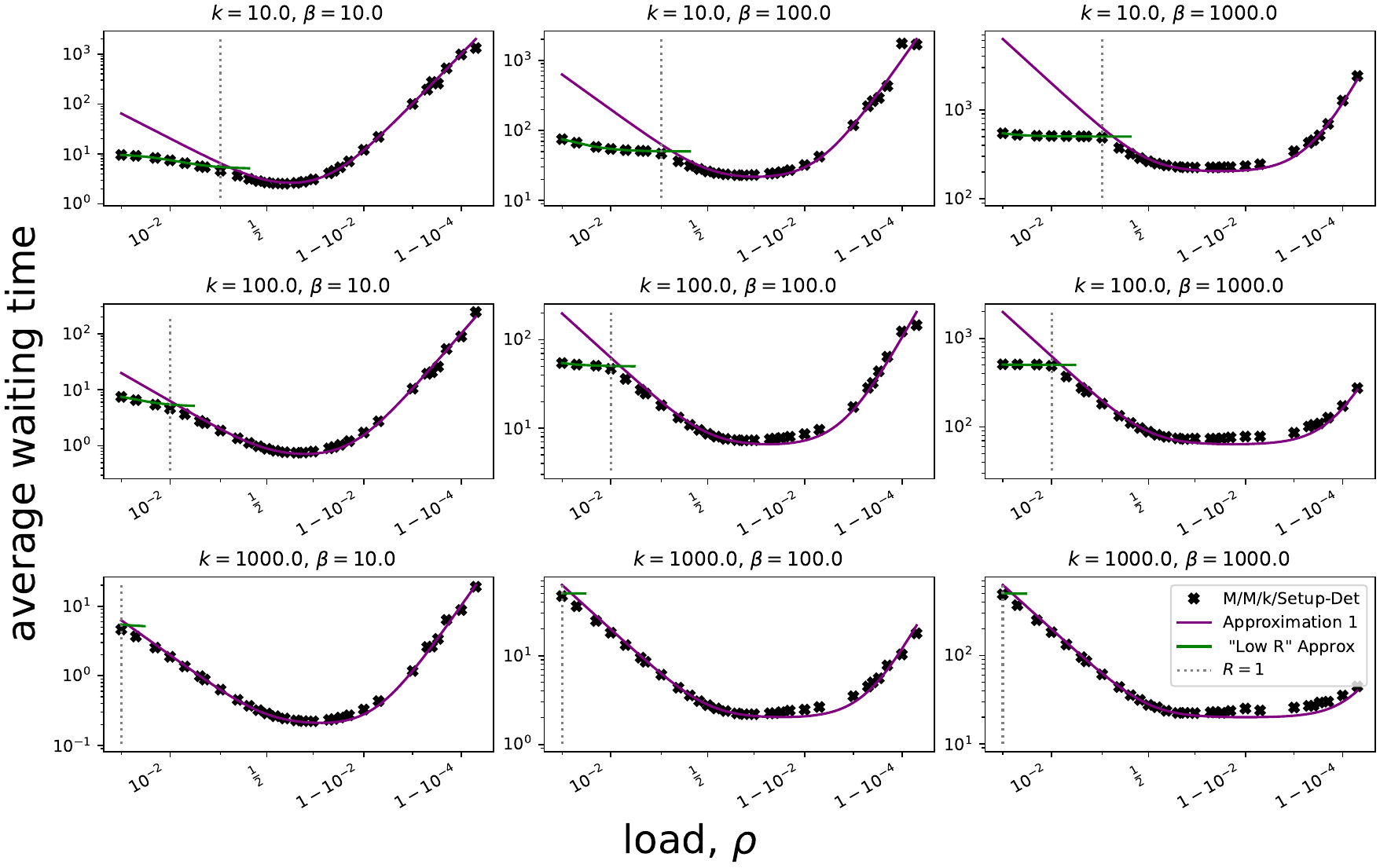}
    \caption{Simulation results demonstrating the high accuracy of Approximation~\ref{apx:only}. For each of these 9 plots, we plot the behavior of the average waiting time \chmades{(in ms)} as one varies the load $\rho$ from $0$ to $1$, holding fixed the total number of servers $k$, \chmades{the mean service time $\frac{1}{\mu} = 1$ ms,} as well as the setup time $\beta$. In each row, we hold the number of servers $k$ constant while testing increasing values of the setup times $\beta$. In each column, we hold the setup time $\beta$ constant while increasing the number of servers. We plot three quantities: 1) in black, the simulated average waiting time for the M/M/k/Setup-Deterministic; 2)in purple, the predicted average waiting time of Approximation~\ref{apx:only}; and \chmade{3) the predicted average waiting time of as given by the ``low R'' approximation of \eqref{eq:ss_approx}, a variation on the single-server setup result of \cite{welchfirstservice}.  We also include, as a reference, a dotted line illustrating the point at which the offered load $R\triangleq k\rho = 1$. Our approximation works well when the average number of busy servers $R > 1$, and the ``low R'' approximation works well when $R < 1$.}}
\label{fig:approxEval}
\end{figure}

\paragraph{Evaluation.} 
\chmade{
In Figure~\ref{fig:approxEval}, we observe that our approximation is accurate across a wide range of system parameters, so long as the offered load $R > 1$, a considerable weakening of the assumptions used in our proofs.
In that sense, it serves as a testament to the strength of our approach that our resulting approximation remains accurate all the way down to offered loads $R \approx 1$.}

\chmade{
When $R < 1$, we find that the waiting time exhibits a ``single-server bottleneck'' effect, mirroring the behavior in \citep{welchfirstservice}. 
In particular, comparing the M/M/1/Setup-Deterministic to the M/M/1 without setup times, one finds \citep{welchfirstservice} that the additional wait due to setup times is
\begin{equation*}
    \chmades{\ex{T_Q}^{\text{M/M/1/Setup-Deterministic}} - \ex{T_Q}^{\text{M/M/1}} =  \frac{\beta}{2}\left[\frac{2 + \lambda \beta}{1 + \lambda \beta}\right] =  \frac{\beta}{2}\left[\frac{2 + \mu R \beta}{1 + \mu R \beta}\right],}
\end{equation*}
\chmades{where we have noted that, in the single server setting, the arrival rate $\lambda = \mu (1 \cdot \rho) = \mu R$.
We can thus use this term to  obtain a ``low $R$'' approximation for the additional wait due to setup in the \textit{multiserver} setting:}
\begin{equation}\label{eq:ss_approx}
    \chmades{\ex{T_Q}^{\text{M/M/k/Setup-Deterministic}} - \ex{T_Q}^{\text{M/M/k}} \approx \ex{T_Q}^{\text{M/M/k/Setup-Deterministic}} \approx} \frac{\beta}{2}\left[\frac{2 + \mu R \beta}{1 + \mu R \beta}\right].
\end{equation}
We include this ``low R'' approximation in Figure~\ref{fig:approxEval} and note that, when $R < 1$, this ``low R'' approximation seems to capture the behavior of the waiting time quite well.
}

\subsection{Upper and Lower Bounds}\label{sec:upper-lower-bounds-statements}
\paragraph{Assumptions.}\label{sec:assumptions}
In our theoretical analysis, in order to more simply characterize the system's behavior, we make two assumptions.
First, we assume that setup times are large compared to service time, i.e. the average setup time $\beta \geq 100\frac{1}{\mu}$; this is often satisfied in practice \citep{gandhi2012autoscale, mao2012performance,mogul2015inferring, maheshwari2018scalability, hao2021empirical}.
Second, we assume that the system utilizes, on average, at least 100 servers, i.e. the offered load $R\triangleq k\rho\geq 100$.

\chmade{
Note that the specific values within these assumptions are not strictly necessary for our analysis to go through; these assumptions are made predominantly to simplify the expressions which arise in our analysis.
For example, in our final bound on the quantity $\ex{L}$ at the end of Appendix G.6, we are left with the expression
\begin{equation*}
    \ex{L} \geq \left(1 - \frac{b_1}{\sqrt{\mu \beta R}}\right)\left[\left(1 - \frac{2}{\sqrt{R}}\right) \left(\sqrt{\frac{\pi}{2}} - \frac{1.15}{\sqrt{\mu \beta}} - \chmades{e^{-4}}\right) - \frac{1}{2 \sqrt{R}}\right] \sqrt{R}.
\end{equation*}
\chmades{By noting that the multiplicative terms in front of the dominant $\sqrt{R}$ term are increasing in the offered load $R$ and relative setup time $\mu \beta$, substituting in $R=\mu\beta=100$, then rounding,} we can reduce the above expression to
\begin{equation*}
    \ex{L} \geq \frac{2}{3} \sqrt{\frac{\pi}{2}}\sqrt{R},
\end{equation*}
which still captures the correct order-wise scaling (and even points towards the correct asymptotic behavior) while being much easier to read.
}

\newcommand{\UBone}{C_2 \left(\setuptime\right) \frac{}}
\newcommand{\shs}{\sqrt{\mu \beta R}}
\newcommand{\virt}{M}

\paragraph{\chmades{Statement of Bounds.}} We are now ready to state the upper and lower bounds on the average queue length.
\begin{theorem}[Upper Bound on Average Queue Length]\label{thm:main}
For an M/M/k/Setup-Deterministic with an offered load $R \triangleq k \rho \geq 100$ and a setup time $\setuptime \geq 100 \frac{1}{\mu}$,
the expected number of jobs in queue in steady state is upper-bounded as
\begin{align*}
    \ex{Q(\infty)} 
    &\leq 3.6\sqrt{\mu \beta R} + 2.04 \frac{\rho}{1-\rho} + \frac{4.05\mu \beta^2 \sqrt{R} + g\left(9 (\mu  \beta)^2 R, 3 \mu \beta \sqrt{R}, k(1-\rho)\right)}{\beta + \frac{L_1 \mu \beta \sqrt{R}}{\mu k(1-\rho)}},
\end{align*}
where the function $g(x,y,z) \triangleq  x \frac{1}{2\mu z} + y\left[\frac{R}{\mu z^2} + \frac{3}{2 \mu z} \right]$ and the constant $L_1 = \frac{2}{3}\sqrt{\frac{\pi}{2}}$.
\end{theorem}


\begin{restatable}[Improved Lower Bound on Average Queue Length]{theorem}{improvement}\label{thm:improvement}
    For an M/M/k/Setup-Deterministic with an offered load $R \triangleq k \rho \geq 100$ and a setup time $\setuptime \geq 100 \frac{1}{\mu}$, the expected number of jobs in queue in steady state is lower-bounded as
    \begin{equation*}
        \ex{Q(\infty)} \geq 
        \frac{ 
            L_1 \beta^2 \sqrt{R} 
            + 
            \ibp{\left[L_1 \beta \sqrt{R} - (k-R) \right]^+}{k - R}
        }
        {
           2.08 \beta +
             \frac{1}{\mu}\frac{F_1 \beta \sqrt{R}}{k - R} + 
             \frac{1}{\mu}\frac{3}{2}\ln(\beta) + \frac{1}{\mu}\ln(F_1 D_1)
        +
        \frac{2}{\mu} 
        + \frac{1}{\mu}\left[D_2 + \frac{D_3}{\sqrt{R}} \right]\max\left(\frac{1}{D_1 \sqrt{\mu \beta}}, \frac{1}{\sqrt{R}}\right)
        },
    \end{equation*}
     where $L_1$, $F_1$, $D_1$, $D_2$, and $D_3$ are constants independent of system parameters. 
\end{restatable}


\paragraph{Discussion of bounds.}\label{sec:discussion_of_bnds}
After simplification, our bounds become
\begin{equation*}
    \ex{Q(\infty)} =_c  \mu \beta \sqrt{R} + \frac{1}{1-\rho};
\end{equation*}
where the operator $=_c$ denotes equality up to multiplicative constants; we show this explicitly in Section~\ref{sec:tightness}. 

This simplified characterization of the queue length gives us insight into how the system parameters govern the system's queueing behavior.
The first term, $\mu\beta \sqrt{R}$, scales linearly with the setup time $\beta$ and captures the effect of turning servers \emph{off} and \emph{on}.
The second term, $\frac{1}{1-\rho}$, dominates the first term only when the load $\rho$ is high enough, say in the renowned super-Halfin-Whitt regime (see, e.g., \citep{HalWhi_81,LiuYin_22,HonWan_24}) where $\rho = 1-\gamma k^{-\alpha}$ with $0<\gamma<1$ and $\alpha >0.5$.
In this case, it recovers the $\frac{1}{1-\rho}$ scaling seen in the M/M/k without setup times.

\chmades{
\subsection{Justification for the Approximation based on the Bounds}\label{sec:approx}
\input{999_approximation_new}

}

%% file: 999_approximation_new.tex
\newcommand{\qapx}{Q_{\textup{apx}}}


\paragraph{Initial Steps.} We arrive at this bound via a straightforward combination of the bounds from Theorems~\ref{thm:main} and~\ref{thm:improvement}, along with a few modifications.
In particular, we follow our previous use of the renewal-reward theorem, then separately approximate the expected time integral over our renewal cycle and the expected length of that renewal cycle, which represent the numerator and denominator of Approximation~\ref{apx:only}, respectively.

\paragraph{Justification of Numerator.}

We first approximate the 
expected time integral over our chosen renewal cycle.
We begin by recalling the lower bound on the time integral, Lemma~\ref{lem:LBint}, which states
\begin{align*}
    \ex{\int_0^{\cycle}{Q(t) \dd t}} \geq    L_1 \beta^2 \sqrt{R} 
            + 
            \ibp{\left[L_1 \beta \sqrt{R} - (k-R) \right]^+}{k - R},
\end{align*}
where $\ibp{x}{z} \triangleq \frac{x}{\mu z} \left[\frac{x+1}{2} + \frac{1}{1-\frac{k\lambda}{k\lambda + \mu z}}\right]$ represents the time integral of the queue length a certain M/M/1 queue over a busy period started by $x$ jobs. 

To obtain the appropriate constant $\capx$, we next note that, although our theorem states $L_1$ as an absolute constant, as the setup time $\beta$ and the offered load $R$ grow, the best possible constant will become $\capx = \sqrt{\frac{\pi}{2}}$. Under the hood, this convergence stems from the fact that
\begin{equation*}
    \sum_{j=1}^{R}{\prod_{i=1}^{j}{\left(1-\frac{j}{R}\right)}} \approx \int_0^{\infty}{e^{\frac{-j^2}{2R}}\dd j} = \frac{1}{2} \sqrt{2 \pi R};
\end{equation*}
see the proof of Lemma~\ref{lem:LBint} for more details. 

To complete the approximation, it suffices to remove the subtraction of $(k-R)$ in the busy period term (which we anticipate is an artifact of our analysis). Removing it, we obtain the desired approximation 
\begin{equation}\label{eq:approxNum}
    \ex{\int_0^{X}{Q(t)\dd t}} \approx  \frac{1}{2} \beta^2 \capx \sqrt{R} + \frac{\beta \capx \sqrt{R}}{\mu k(1-\rho)}\left[\frac{\beta \capx \sqrt{R} + 1}{2} + \frac{1}{1-\rho} \right].
\end{equation}

\paragraph{Justification of Denominator.}

We next approximate the denominator of our expression, the expected length of our chosen renewal cycle.
To do so, we again make use of the lower bound on the expected cycle length $\ex{X}$ from Lemma~\ref{lem:LBX}, which states
\begin{equation*}
    \ex{X} \geq \beta + \frac{L_1 \beta \sqrt{R}}{\mu k(1-\rho)}.
\end{equation*}
By making the same convergence argument for $L_1$, i.e. that $L_1 \to \capx$ for large setup times $\beta$ and large offered loads $R$, we obtain the denominator, completing both parts of our bound.

%% file: 5_keyideas.tex
\newcommand{\renewpts}{\mathcal{X}}
\newcommand{\ncouple}{\Tilde{N}}
\newcommand{\starttime}{T_0}
\newcommand{\finalpt}{P}
\newcommand{\zplus}{\mathbb{Z}^+}
\newcommand{\mint}[2]{\min\left(#1, #2\right)}
\newcommand{\epochind}{n_e}

We now describe our approach to analyzing the average waiting time in the M/M/k/Setup-Deterministic, using the upper bound, Theorem~\ref{thm:main}, as a case study.
To begin, we go through the first few steps in our proof, arriving at three main lemmas. 
We then give a detailed explanation of the technique we have developed for proving those lemmas, which we call the \emph{Method of Intervening Stopping Times} (MIST). 
Throughout, we highlight the technical challenges that arise in our approach and how we address those challenges.

\subsection{Initial Steps: Applying the Renewal-Reward Theorem}\label{sub:init}
\subsubsection{\texorpdfstring{Reduction to Analyzing $\ex{N(\infty) - R}$.}{Reduction to analyzing E[N(infty) - R].}} We begin by applying the Renewal Reward theorem. Although it is tempting to apply the theorem directly to the queue length $\ex{Q(\infty)}$, it simplifies the analysis considerably if one analyzes the number of jobs $\ex{N(\infty) - R}$ instead. To justify this, note that, in steady-state, the number of busy servers $\ex{Z(\infty)} = R$, and that the total number of jobs in system $N(t)$ satisfies $N(t) = Q(t) + Z(t)$.
It follows that the average queue length
\begin{equation*}
	\ex{Q(\infty)} = \ex{N(\infty) - Z(\infty)} = \ex{N(\infty)} - R = \ex{N(\infty) - R}.
\end{equation*}

\subsubsection{Applying Renewal-Reward.}
Applying the Renewal Reward Theorem, for any renewal cycle,
\begin{equation}\label{eq:RRT}
	\ex{N(\infty) - R} = \frac{
		\ex{\int_{\text{cycle}}{\left[N(t) - R\right] \dd t}}
	}
	{
		\ex{\text{cycle length}}
	}
	=
	\frac{
		\ex{\int_{0}^{X}{[N(t) - R] \dd t}}
	}
	{
		\ex{\int_0^X{1 \dd t}}
	},
\end{equation}
where we have set time $0$ to be an arbitrary renewal point and time $X$ to be the renewal point which immediately follows it. 
To upper bound \eqref{eq:RRT}, it suffices to upper-bound the right side's numerator and lower-bound the right side's denominator. 
These bounds constitute our three main lemmas; two lemmas for the numerator (which is harder to bound) and one for the denominator. 
But before we can state these main lemmas, we must first address our first technical challenge: defining the collection of renewal points $\renewpts$.

\subsubsection{Technical Challenge: Defining the Renewal Points.}

To define a ``good'' collection of renewal points $\renewpts$ is a non-trivial task. There are two important criteria one must consider when selecting the collection of renewal points. First, we need the length of a renewal cycle to have finite expectation. While this is trivially satisfied in many queueing models, in our model, the presence of remaining setup times in our system state makes this criterion non-trivial.  
\chmade{Second, we want the renewal period to be defined in such way that one can straightforwardly reason about the system's behavior in between renewals.}
For example, we could define our renewals as occurring when the system completely empties; this would avoid the continuous-state-space concerns, but then one would need to reason about what happens to the system until the next time it empties ---by analogy to the M/M/k without setup, a typical renewal cycle would last exponentially long (in $k$). Thus, while the ``system empties'' renewal definition does indeed give a simple analysis for $k=1$, for even $k=2$ the analysis becomes much more difficult. We must find another way.


\subsubsection{\texorpdfstring{Key Idea: Renewals When the $(R+1)$-th Server Turns Off.}{Key Idea: Renewals when the (R+1)-th server turns off.}}
By defining our collection of renewal points $\renewpts = \left\{t: Z(t^-)=R+1, Z(t) = R \right\}$ as those moments when the $(R+1)$-th server is turned off, we satisfy both criteria.
By defining our renewals to occur at moments when a server is turned \textit{off}, we avoid the continuous state-space issue: any time a server shuts \textit{off}, we know that there must be no servers in \textit{setup} (since we would cancel any servers in \textit{setup} before canceling an idling server).
Furthermore, by specifically analyzing around the moment the $(R+1)$-th server turns off, the length of time between renewals becomes well-controlled, in the sense that it is lower-bounded almost surely by a setup time and upper-bounded in expectation by a quantity which we can show is ``not too large.''

\begin{figure}
    \centering
    \includegraphics[trim={0 3cm 0 0 }, scale=0.25,clip]{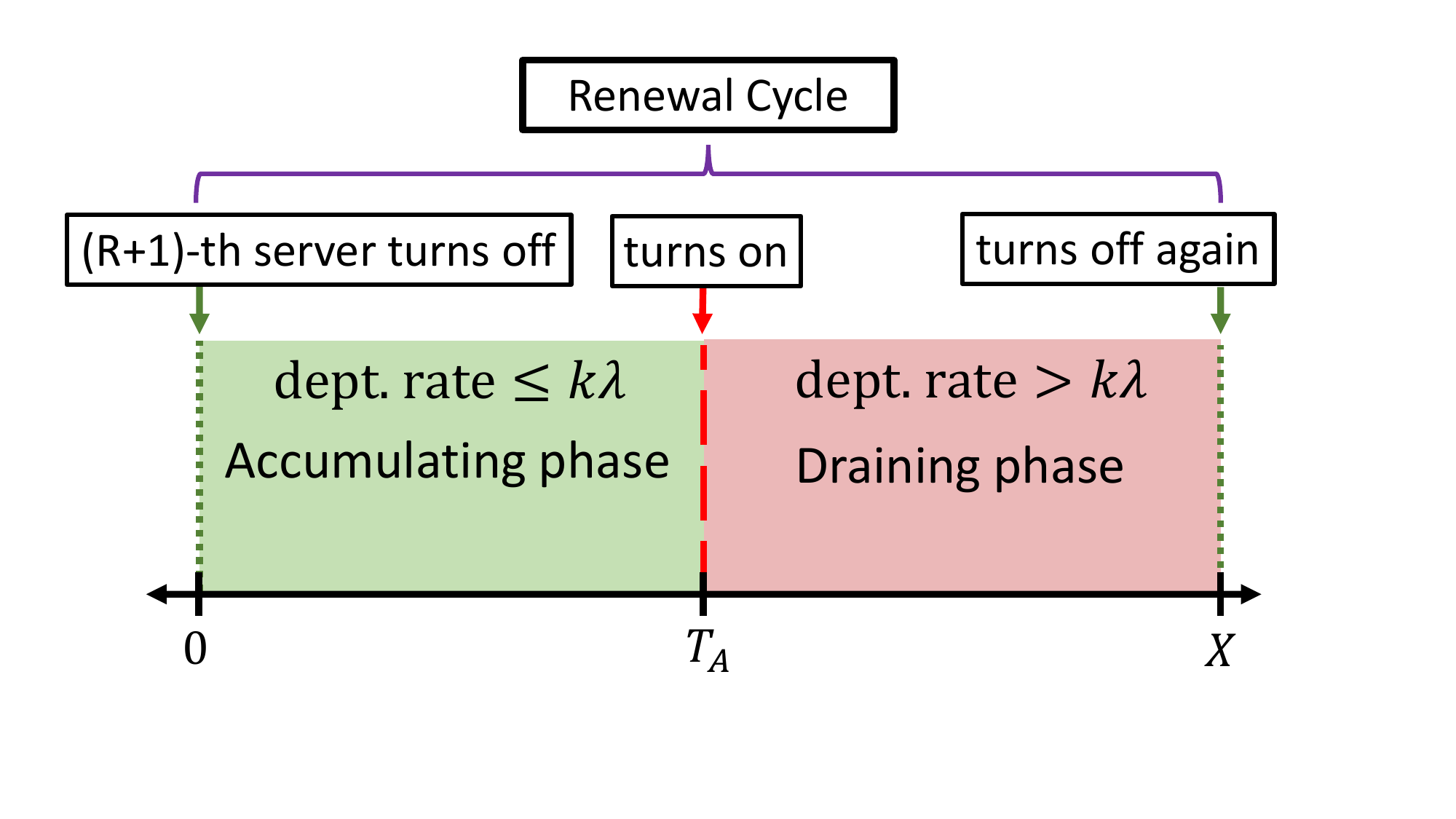}
    \caption{A depiction of our decomposition of a renewal cycle into an accumulating phase and draining phase, described in Section~\ref{sub:init}. During the accumulating phase, the departure rate $\mu Z(t) \leq \mu R = k \lambda$, so that the system is transiently unstable and a queue \emph{accumulates}. During the draining phase, the departure rate $\mu Z(t) > k \lambda$, so that the queue \emph{drains}.}
    \label{fig:accdec}
\end{figure}

\subsubsection{Further Benefit: A Natural Decomposition.}Another benefit of defining a renewal cycle in this way is that it naturally splits a renewal cycle into two parts: the first part of the cycle \emph{before} the $(R+1)$-th server turns on, and the second part of the cycle \textit{after} it turns on; we depict this decomposition in Figure~\ref{fig:accdec}.
During the beginning of a cycle (when $Z(t) \leq R$), the departure rate $\mu Z(t)$ is smaller than the arrival rate, making the system behave, in a transient sense, like a critically- or over-loaded queue.
On the other hand, after the $(R+1)$-th server turns on and until the renewal cycle is over, the departure rate $\mu Z(t)$ is guaranteed to be strictly greater than the arrival rate, making the queue, on the whole, drain over time.
This observation turns out to be hugely useful in our analysis. As such, we have special names for each of these special times: we call the time before the $(R+1)$-th server turns on the \emph{accumulation period}, we call the moment when the $(R+1)$-th server turns on the \emph{accumulation time} $T_A$, and we call the period from time $T_A$ until the cycle ends the \emph{draining period}.
At times, we will also refer to these periods as \textit{phases}.

\subsection{Statement of Three Main Lemmas}
We may now state our three main lemmas. Note that the preconditions of Theorem~\ref{thm:main} are implicit here.

\begin{restatable}[Upper Bound on Integral Over Accumulation Period]{lemma}{mainone}\label{lem:UBintTA}
Suppose the system begins at time $0$ with $R$ jobs in service and no jobs in the queue (and thus no servers in setup), and define the accumulation time $T_A \triangleq \minpar{t \geq 0: Z(t) = R+1}$ to be the moment the $(R+1)$-th server turns on.
Then, taking $B_1 = 3.6$ and $B_2 = 1.04$, one can bound the integral of the reduced number of jobs $\left[N(t) - R\right]$ by 
	\begin{equation*}
		\ex{\int_0^{T_A}{[N(t) - R] \dd t}} \leq B_1 \shs \cdot \ex{T_A} + B_2 \beta^2 \mu \sqrt{R}.
	\end{equation*}
\end{restatable}


\begin{restatable}[Upper Bound on Integral Over Draining Period]{lemma}{maintwo}\label{lem:UBintTB}
Recall that the accumulation time $T_A$ is the first (and only) time the $(R+1)$-th server turns on during a renewal cycle, and that the next renewal point $X = \minpar{t > T_A: Z(t) = R}$ is simply the next time the $(R+1)$-th server turns off. 
Then, 
	\begin{align*}
 \ex{\int_{T_A}^{X}{\left[N(t) - R\right] \dd t}} \leq 3.01\mu\beta^2 \sqrt{R} + 2.04 \beta \frac{\rho}{1-\rho} + g\left(9(\mu\beta)^2 R, 3 \mu \beta \sqrt{R}, k(1-\rho)\right)
	\end{align*}
 where $g(x,y,z) \triangleq  x \frac{1}{2\mu z} + y\left[\frac{R}{\mu z^2} + \frac{3}{2 \mu z} \right].$
\end{restatable}

\begin{restatable}[Lower Bound on Cycle Length]{lemma}{mainthree}\label{lem:LBX}
Taking $L_1 = \frac{2}{3}\sqrt{\frac{\pi}{2}}$, the renewal cycle length $\ex{X}$ is
	\begin{equation*}
		\ex{X} \geq \setuptime + \frac{L_1 \mu \beta \sqrt{R}}{\mu k(1-\rho)}.
	\end{equation*}

\end{restatable}

\subsection{Proof of Theorem~\ref{thm:main}, Assuming the Three Main Lemmas}\label{sec:mainComp}
After proving these lemmas, the result easily follows. 
First note that, by summing Lemmas~\ref{lem:UBintTA} and~\ref{lem:UBintTB},
\begin{align*}
    \ex{\int_0^{X}{[N(t) - R] \dd t}} \leq B_1 \sqrt{\mu \beta R}\cdot \ex{T_A} + (B_2 + 3.01)\mu \beta^2 \sqrt{R} + 2.04 \beta \frac{\rho}{1-\rho} + g\left(9 (\mu  \beta)^2 R, 3 \mu \beta \sqrt{R}, k(1-\rho)\right),
\end{align*}
where $g(x,y,z) \triangleq  x \frac{1}{2\mu z} + y\left[\frac{R}{\mu z^2} + \frac{3}{2 \mu z} \right].$

Upon dividing by Lemma~\ref{lem:LBX} and noting that a setup time $\beta <\ex{T_A}< \ex{X}$, we obtain 
\begin{align*}
    \ex{Q(\infty)} 
    &\leq B_1 \sqrt{\mu \beta R} + 2.04 \frac{\rho}{1-\rho} + \frac{(B_2 + 3
    .01)\mu \beta^2 \sqrt{R} + g\left(9 (\mu  \beta)^2 R, 3 \mu \beta \sqrt{R}, k(1-\rho)\right)}{\beta + \frac{L_1 \mu \beta \sqrt{R}}{\mu k(1-\rho)}}. \tag*{\Halmos}
\end{align*}

\subsection{Key Technique: The Method of Intervening Stopping Times (MIST)}

\subsubsection{Technical Challenge: Applying Simple Bounds to a Complex System.}

We begin by describing the technical challenge which led to the development of the MIST method: attempting to apply simple coupling bounds to a complex, dynamic system. 
In general, one can easily obtain simple coupling bounds on the behavior of, e.g., the number of jobs in the system $N(t)$ or the number of busy servers $Z(t)$; unfortunately, these simple bounds may only be accurate and/or applicable for a \emph{random} amount of time. 
For example: Suppose that at time $\tau$ the number of busy servers $Z(\tau) = 6$ and the smallest remaining setup time in the system is $10$ seconds. 
Then, until either a server is turned off or $10$ seconds have passed, we understand the behavior of the number of jobs $N(t)$ very well  --- in a sample path sense, it is \emph{precisely} the behavior of an M/M/1 queue, with arrival rate $k\lambda$ and departure rate $6 \mu$.
Suppose, though, that by time $\tau + 2$, the system is empty. 
If we continued to estimate the system behavior via the same M/M/1 queue as before, then we would be wildly over-estimating the true system's departure rate.
If we want to keep our couplings accurate, we need to track when the system's behavior changes significantly, and alter our coupling accordingly.

\begin{figure}
    \centering
    \includegraphics[scale=0.25,trim={6cm 4cm 7.5cm 4.5cm},clip]{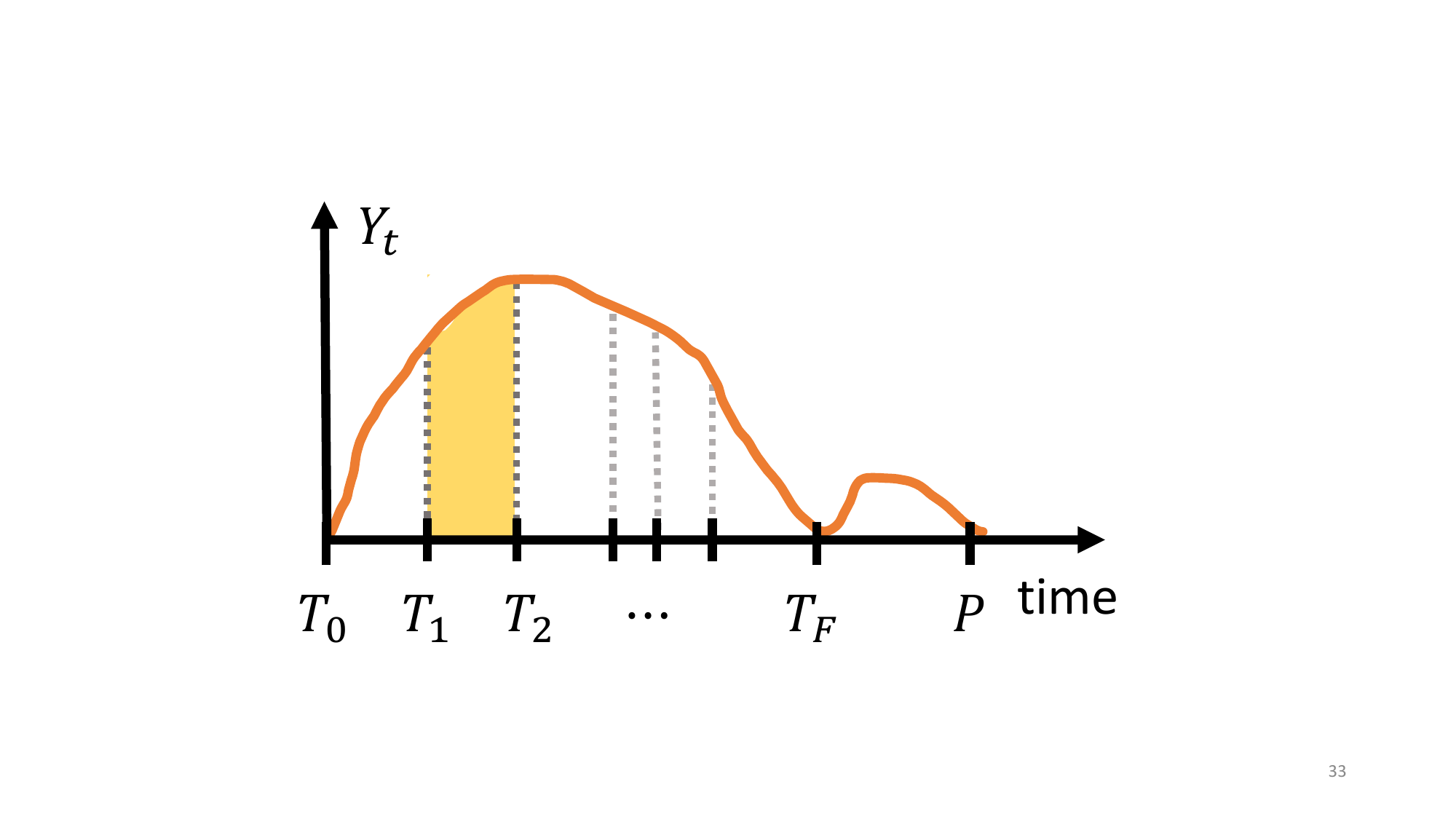}
    \caption{A depiction of the stopping time decomposition of Lemma~\ref{lem:waldstop}. In order to bound the integral $\ex{\int_{T_0}^{P}{Y_t \dd t}}$, we bound the expected contribution of each random interval $[T_i, T_{i+1})$, where every $T_i$ is a cleverly defined stopping time. This divide-and-conquer approach features heavily in our proofs.}
    \label{fig:waldstop}
\end{figure}

\subsubsection{Key Idea: Decompose Using State-Based Stopping Times.}
The Method of Intervening Stopping Times (MIST) allows us to easily and dynamically apply these coupling arguments.
Our key idea is to break up our original long time interval into a \textit{random number} of smaller, more manageable pieces.
We do this by defining \emph{intervening events}, moments where the system state changes in a way that allows to easily characterize the system's behavior.
From there, we can define a ``small piece'' of time as the time in between our intervening events.
For example, in this work, it can often be useful to analyze the system around time points where the number of jobs $N(t)$ gets large.
This is useful because, in a system with setup times, if we have \textit{enough} jobs for \textit{long enough}, we can afterwards guarantee that many servers are turned on.

\subsubsection{Reduction to Three Claims.} By performing this decomposition of the integral, we reduce our initial bounding problem to showing two or three bounds. 
First, we must bound the integral of each of the chunks. 
Often, it is useful to analyze the ``initial'' chunk separately from the ``successive'' chunks afterward.
Typically, we use martingale arguments combined with worst-case coupling arguments to prove these two integral bounds. 
Next, we must show that not too many of these ``successive chunks'' actually occur.
For this ``not too many'' condition, it's often helpful to show some kind of regularity condition, e.g. if the $i$-th intervening event has occurred, then the $(i+1)$-th event occurs with at most constant probability.
By formalizing our notion of \textit{intervening events} using stopping times and applying some ideas from Wald's equation, we obtain the Intervening Stopping Time Lemma, Lemma~\ref{lem:waldstop}.



\subsection{The Intervening Stopping Time Lemma: Statement}

We now state the Intervening Stopping Time Lemma, Lemma~\ref{lem:waldstop}, deferring its proof to Section~\ref{sec:waldstop}. Afterwards, we sketch how we use MIST to prove Lemma~\ref{lem:UBintTA}, the bound on the time integral $\ex{\int_0^{T_A}{[N(t) - R]}}$.

\begin{restatable}[Intervening Stopping Time Lemma]{lemma}{waldstop}\label{lem:waldstop}
Given a starting stopping time $T_0$, an ending stopping time $\finalpt$, and a collection of intervening stopping times $\left(T_i: i\in \zplus \right)$, define the random variable $F$ to be such that $T_F \leq \finalpt < T_{F+1}$. 
Now, given some time-varying random variable $Y_t \geq 0$ which is a function of the underlying Markov state of the system $\s(t)$, suppose that:
\begin{enumerate}
	\item $\ex{\int_{\starttime}^{\min(T_1, \finalpt)}{Y_t \dd t}\middle | \filt{\starttime}} \leq G_0(\s(T_0)),$
	\item $\ex{\int_{T_{i}}^{\min(T_{i+1}, \finalpt)}{Y_t \dd t} \middle | \filt{T_i}, F \geq i} 
	\leq  G_i + B \cdot \ex{\min(T_{i+1}, \finalpt) - T_i\middle | \filt{T_i}, F \geq i} + C\left[T_i - T_{i-1}\right]$,
	\item and $\pr{F \geq i \middle | \filt{T_i}, F \geq i-1 } \leq 1 - p_i,$
\end{enumerate}
where $G_0$ is also some function of the system state, and the $G_i$'s, the $p_i$'s, $C$, and $B$ are all constants (possibly depending on system parameters).
Then one can bound the time integral of $Y_t$ from $T_0$ to $P$ by
\begin{equation*}
	\ex{\int_{\starttime}^{\finalpt}{Y_t \dd t}} \leq \ex{G_0\left(\s(T_0)\right)} + \pr{F > 0}\sum_{j=1}^{\infty}{G_j \prod_{i=2}^{j}{\left(1 - p_i\right)}}  + (B+C) \cdot \ex{\finalpt - T_0}.
\end{equation*}
\end{restatable}

\subsection{Proof of Lemma~\ref{lem:waldstop}, the Intervening Stopping Time (IST) Lemma}\label{sec:waldstop}

\input{claims/waldstop}

\subsection{Example: Application to Accumulation Period Upper Bound (Lemma~\ref{lem:UBintTA})}

\begin{figure}\label{fig:epoch}
    \centering
    \includegraphics[scale=0.23,clip]{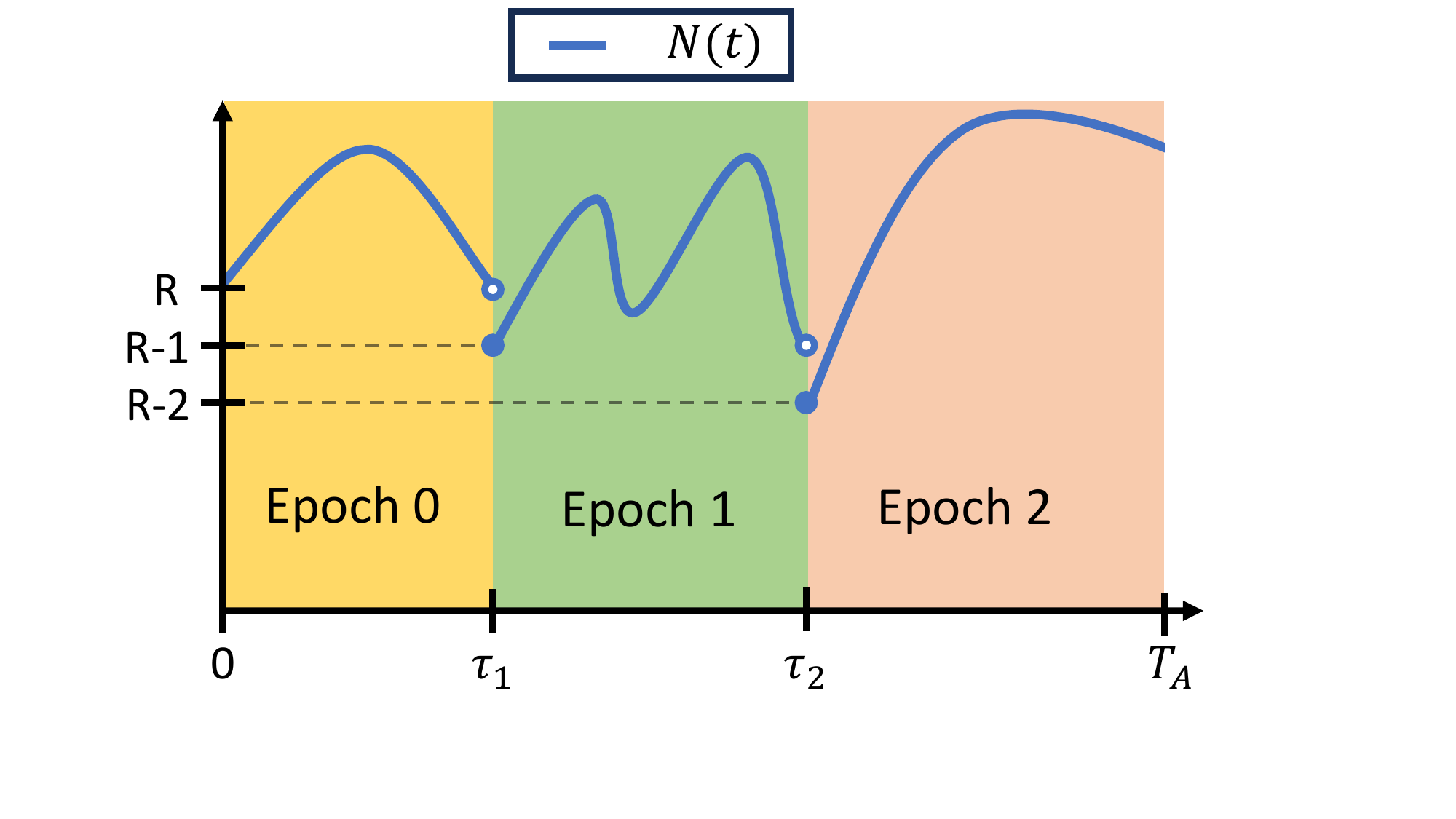}
    \vspace{-12pt}
    \caption{A depiction of the sequence of ``new lows'' $\tau_j$ and the epochs that they define. In this case, the number of epochs which occur is $n_e = 3$.}
\end{figure}


\paragraph{Specification step.}We begin our bound on the integral over the accumulation period with an application of the Intervening Stopping Time Lemma, Lemma~\ref{lem:waldstop}. 
To apply our new lemma, we must first specify all the necessary stopping times: the initial time $\starttime$, the final time $\finalpt$, and the intervening times $T_i$. 
In this case, we of course take $\starttime = 0$ and $\finalpt = T_A$. 
For the intervening times $T_i$, we choose a sequence which allows us to lower bound the departure rate of our system. 
In particular, we take $T_j = \tau_j$ where the $j$-th epoch start $\tau_j$ is
\begin{equation*}
    \tau_j \triangleq \minpar{t\geq0: N(t) \leq R-j};
\end{equation*}
in other words, the moments when the number of jobs $N(t)$ reaches a new minimum value or ``new low.''
We call the period $[\tau_j, \mint{\tau_{j+1}}{T_A})$ the \emph{$j$-th epoch}, and say epoch $j$ \textit{occurs} whenever $\tau_j < T_A$. We then let the random variable $\epochind$ denote the number of epochs that occur, i.e. we take $F=\epochind$.

\paragraph{The reduced problem.} With this decomposition, understanding the behavior of the accumulation phase reduces to understanding the behavior of each epoch. The beauty in our construction is that the behavior of $N(t)$ within an epoch is much more constrained: within epoch $j$, the departure rate $\mu Z(t)$ of our system is both lower-bounded by $\mu (R-j)$ and upper-bounded by $\mu R$. We use this observation together with various coupling arguments to show the three necessary preconditions of Lemma~\ref{lem:waldstop}. 
To prove these preconditions, we actually need to apply Lemma~\ref{lem:waldstop} \emph{again}; see Section~\ref{sec:UBintTA} for details.


%% file: claims/waldstop.tex
We begin with a manipulation of the integral, finding
\begin{align*}
    \int_{T_0}^{P}{Y_t \dd t}
    &\;\;= \int_{T_0}^{\mintwo{T_1}{P}}{Y_t \dd t} + \sum_{i=1}^{\infty}{\int_{\mintwo{T_i}{P}}^{\mintwo{T_{i+1}}{P}}{Y_t \dd t}}\\
    &\;\; = \int_{T_0}^{\mintwo{T_1}{P}}{Y_t \dd t} + \sum_{i=1}^{\infty}{\indc{T_i < P} \int_{T_i}^{\mintwo{T_{i+1}}{P}}{Y_t \dd t}}.
\end{align*}
Applying linearity of expectation and the tower property, we find that
\begin{align*}
    \ex{\int_{T_0}^{P}{Y_t \dd t}}
    &\;\;= \ex{\ex{\int_{T_0}^{\mintwo{T_1}{P}}{Y_t \dd t}\middle | \filt{T_0}}} + \sum_{i=1}^{\infty}{\ex{\ex{\indc{T_i < P} \int_{T_i}^{\mintwo{T_{i+1}}{P}}{Y_t \dd t}\middle | \filt{T_i}}}}\\
    &\;\;= \ex{\ex{\int_{T_0}^{\mintwo{T_1}{P}}{Y_t \dd t}\middle | \filt{T_0}}} + \sum_{i=1}^{\infty}{\ex{\indc{T_i < P}\ex{\int_{T_i}^{\mintwo{T_{i+1}}{P}}{Y_t \dd t}\middle | \filt{T_i}}}}.\intertext{Noting that the event $\brc{T_i < P}= \brc{F \geq i}$, we have}
    &\;\;= \ex{\ex{\int_{T_0}^{\mintwo{T_1}{P}}{Y_t \dd t}\middle | \filt{T_0}}} + \sum_{i=1}^{\infty}{\ex{\indc{F \geq i}\ex{\int_{T_i}^{\mintwo{T_{i+1}}{P}}{Y_t \dd t}\middle | \filt{T_i}}}}\\
    &\;\;\leq \ex{G_0\left(S(T_0)\right)} + \sum_{i=1}^{\infty}{\ex{\indc{F \geq i}\left( G_i + B \cdot \ex{\min(T_{i+1}, \finalpt) - T_i\middle | \s(T_i), F \geq i} + C\left[ T_i - T_{i-1}\right]\right)}}\\
    &\;\;= \ex{G_0\left(S(T_0)\right)} + B \cdot {\ex{\finalpt - T_0}} + C \ex{T_F - T_0} + \sum_{i=1}^{\infty}{G_i \pr{F \geq i}}\\
    &\;\;\leq \ex{G_0\left(S(T_0)\right)} + B \cdot {\ex{\finalpt - T_0}} + C \ex{\finalpt - T_0} + \sum_{i=1}^{\infty}{G_i \pr{F \geq i}}.
\end{align*}
Applying our final assumption to bound $\pr{F \geq i}$, 
\begin{align*}
    \pr{F \geq i} &= \pr{F > 0} \prod_{j=2}^{i}{\pr{F \geq j\middle | F \geq j - 1}}\\
    &= \pr{F > 0} \prod_{j=2}^{i}{\ex{\pr{F \geq j\middle | F \geq j - 1, \filt{T_{j-1}}}}}\\
    &\leq \pr{F > 0} \prod_{j=2}^{i}{\ex{1- p_j}}= \pr{F > 0} \prod_{j=2}^{i}{(1- p_j)}.
\end{align*}
Applying this final result, we find
\begin{equation}
     \ex{\int_{T_0}^{P}{Y_t \dd t}} \leq \ex{G_0\left(\s(T_0)\right)} + \pr{F > 0}\sum_{j=1}^{\infty}{G_j \prod_{i=2}^{j}{\left(1 - p_i\right)}}  + (B+C) \cdot \ex{\finalpt - T_0}.\tag*{\Halmos}
\end{equation}

%% file: 6_proof.tex
As shown in Section~\ref{sec:mainComp}, to prove Theorem~\ref{thm:main} it suffices to prove the following lemmas: 
\begin{enumerate}
    \item Lemma~\ref{lem:UBintTA}, an upper bound on the reward integral over the accumulation phase,
    \item Lemma~\ref{lem:UBintTB}, an upper bound on the reward integral over the draining phase, and
    \item Lemma~\ref{lem:LBX}, a lower bound on the length of our renewal cycle.
\end{enumerate}
Before we continue, we first give a formal system construction, then prove some useful coupling claims.


\input{claims/coupling_construction}

\input{ThesisFiles/UB/lemma_UB_int_TA}

\input{ThesisFiles/UB/lemma_UB_int_TB}

\input{ThesisFiles/UB/lemma_LB_X}

%% file: claims/coupling_construction.tex
\subsection{Construction}
We now discuss how we formally construct this system using Poisson processes; being explicit here will prove useful when we make coupling arguments in the future.

\paragraph{The arrival and departure processes.} We take the number of jobs that have arrived at time $t$ to be $\Pi_A(t)$, where $\Pi_A$ is a Poisson process of rate $k \lambda$. In a slight abuse of notation, we let $\Pi_A( [a,b] )$ denote the number of arrivals that occur in the interval $[a,b]$; we apply the same extension to all other counting processes mentioned here.
We set the potential departure process of, say, server $i$ to be $\Pi_i( t)$, where $\Pi_i$ is a Poisson process of rate $\mu$.
A potential departure from server $i$ only ``counts'' if server $i$ is busy when that potential departure occurs, i.e., if the number of busy servers $Z(t) \geq i$ at the time.
Thus, the total number of departures from our system by time $t$ is, taking integrals with respect to the Poisson processes $\Pi_i$ as counting processes,
$$\dept(t) \triangleq \sum_{i=1}^{k}{\int_0^{t}{\indbig{Z(s) \geq i} \dd \Pi_i(s)} }.$$

\paragraph{The number of busy servers $Z(t)$.} To find the number of busy servers $Z(t)$, one could count the number of setup completion events that have occurred so far and the number of server shutoffs that have occurred so far; this description is a bit difficult to work with. Alternatively, one can see from the initial description of setup dynamics that server $i$ is \textit{on} at time $t$ if and only if the total number of jobs $N(s) \geq i$ for all $s \in [t-\beta, t]$, where one should recall that $\beta$ is the setup time.
An easier description of $Z(t)$ follows:
\begin{equation*}
    Z(t) = \min\left(k, \min_{s \in [t-\beta, t]}{N(s)}\right).
\end{equation*}

\paragraph{A departure operator.} We can extend our departure process $\dept(t)$ to a departure operator $\depOp{f(s)}{\interv}$ which takes a function $f(s) \in \left\{0,1,\dots, k\right\}$ defined on some interval $\interv$ and computes the number of departures that would occur in that interval provided that the number of busy servers $Z(s)= f(s)$, i.e.
\begin{equation*}
    \depOp{f(s)}{(a,b]} \triangleq \sum_{i=1}^{k}{\int_a^{b}{\indbig{f(s) \geq i} \dd \Pi_i(s)} }.
\end{equation*}
Note that  the total number of departures can now be written as $\dept(t) = \depOp{Z(s)}{[0,t]}$.

\subsection{Three Coupling Claims}
We now describe three useful claims applied throughout the proof. The first, we will state and prove immediately. The latter two, we prove later, in Section~\ref{sec:appendix}.

\subsubsection{Basic Coupling Claim: Maintaining an Initial Relation}\label{sec:coupling}
\begin{claim}[Basic Coupling]\label{clm:coupling}
    Suppose that we have two processes $N_1$ and $N_2$ with an initial relation $N_1(a) \leq N_2(a),$
    where the behavior of each process is governed, for all times $s$ from $a$ up to some stopping time $\tau$, by the equation
    \begin{equation*}
        N_j(s) \triangleq N_j(a) + \Pi_A\left( (a, s] \right) - \depOp{Z_j(x)}{(a,s]}, \text{ for $j \in \left\{1,2\right\}$}.
    \end{equation*}
    Furthermore, suppose that the first system's number of busy servers $Z_1(s)\geq Z_2(s)$ for all times $s \in [a, \tau].$
    Then, for all $s \in [a, \tau]$, the relation is maintained, i.e. $N_1(s) \leq N_2(s).$
\end{claim}
\begin{proof}{Proof.}
We show equivalently that $N_2(s) - N_1(s) \geq 0$. Applying the definitions of $N_1$ and $N_2$, 
\begin{align*}
    N_2(s) - N_1(s) &= N_2(a) - N_1(a) + \left[\depOp{Z_1(x)}{(a,s]} - \depOp{Z_2(x)}{(a,s]}\right]\\
    &\geq \left[\depOp{Z_1(x)}{(a,s]} - \depOp{Z_2(x)}{(a,s]}\right]\\
    &= \sum_{i=1}^{k}{\int_a^{s}{\indbig{Z_1(x) \geq i} \dd \Pi_i(x)} } - \sum_{i=1}^{k}{\int_a^{s}{\indbig{Z_2(x) \geq i} \dd \Pi_i(x)} }\\
    &= \sum_{i=1}^{k}{\int_a^{s}{\bigg[ \indbig{Z_1(x) \geq i} - \indbig{Z_2(x) \geq i}\bigg] \dd \Pi_i(x)} }.
\end{align*}
Since $Z_1(x) \geq Z_2(x)$, the integrand $\bigg[ \indbig{Z_1(x) \geq i} - \indbig{Z_2(x) \geq i}\bigg] \geq 0;$ the claim follows. \hfill\Halmos
\end{proof}


\subsubsection{Description of Remaining Coupling Claims}\label{sec:remCoupling}

\paragraph{High-level explanation.}
This claim leads nicely into a couple more claims.
Both are concerned with bounding a quantity involving a general ``down-crossing'' time.
In particular, our analysis will begin at a stopping time $\tau$ and will ``end'' at the down-crossing time $\dgen$, where $\dgen \triangleq \minpar{t\geq 0: N(t + \tau) \leq h}$ is the length of time it takes for the number of jobs $N(t)$ to become lower than some given threshold $h$.
The first claim, Claim~\ref{clm:dcbnd}, uses a coupling argument to bound the expected integral of $N(t)$ from some arbitrary time $\tau$ until $N(t)$ drops below some pre-defined threshold $h$, provided that one has a lower bound on the number of busy servers $Z(t)$ over that period. 
The second claim, Claim~\ref{clm:dcprob}, uses a related argument to bound the probability that $N(t)$ drops below some threshold $h$ within some amount of time $\ell$, given that one has bounds on $Z(t)$ over the relevant period. 
\chmades{
We state and prove both of these claims now.
}

\subsubsection{Statement and Proof of Claim~\ref{clm:dcbnd}, the Coupling Integral Bound.}\label{sec:dcbnd}
\begin{restatable}[Coupling Integral Bound]{claim}{dcbnd}\label{clm:dcbnd}
    Let $\tau$ be some stopping time and $\dgen$ be the next down-crossing as described in Section~\ref{sec:remCoupling}. Suppose that, at time $\tau$, we have a lower bound on the number of busy servers over a period, i.e. we know that the number of busy servers $Z(t) \geq R - j,$
    for all $t \in \left[\tau, \tau + \mintwo{\ell}{\dgen} \right]$ and for some non-negative $j$.
    Then we have the following bound on the integral over this time period:
    \begin{equation*}
        \ex{\int_{\tau}^{\tau + \mintwo{\dgen}{\ell}}{\left[N(t) - h\right] \dd t}\middle | \filt{\tau}} \leq \ell \cdot \left[N(\tau) - h \right]^+ + \frac{1}{2}\mu j \ell^2.
    \end{equation*}
\end{restatable}

\chmades{
\input{claims/dcbnd}
}
\subsubsection{Proof of Claim~\ref{clm:dcprob}, the Coupling Probability Bound}\label{sec:dcprob}

\begin{restatable}[Coupling Probability Bound]{claim}{dcprob}\label{clm:dcprob}
    Let $\tau$ be some stopping time and $\dgen$ be the next down-crossing as described in Section~\ref{sec:remCoupling}.
    We consider two cases.
    
    In the first case, suppose that we have a \textbf{lower} bound on the number of busy servers $Z(t)$ over some length $\ell$ interval starting at time $\tau$, i.e. the busy servers $Z(t) \geq R - j,$
    for all $t \in \left[\tau, \tau + \mintwo{\ell}{\dgen} \right]$ and for some non-negative $j$. 
    Then, we can bound the threshold-crossing probability by
    \begin{equation*}
        \pr{\dgen < \ell\middle | \filt{\tau}} \geq 2\Phi\left(-\left[\frac{N(\tau) - h + \mu j \ell}{\sqrt{\ell(2 k\lambda - \mu j) }}\right]\right) - \frac{2}{3\sqrt{\ell(2k\lambda - \mu j)}}.
    \end{equation*}
    In particular, if $N(\tau) - h = c_1 \sqrt{\mu \beta R}$, then the probability $ \pr{\dgen < \ell\middle | \filt{\tau}} \geq 2 \Phi\left( - \frac{c_1}{\sqrt{2}} \right) - \frac{1}{100}.$

    In the second case, suppose that we instead have the \textbf{upper} bound on $Z(t) \leq R$ during this interval instead. Then,
    \begin{equation*}
        \pr{\dgen < \ell \middle | \filt{\tau}} \leq 2 \Phi\left(-\left[ \frac{N(\tau) - h}{\sqrt{2 \ell k \lambda}}\right] \right) - \frac{2}{3\sqrt{2 k\lambda \ell}}.
    \end{equation*}
    As before, if $N(\tau) - h =  c \sqrt{\mu \beta R}$, then the probability $\pr{\dgen < \ell\middle | \filt{\tau}} \leq 2 \Phi\left( - \frac{c}{\sqrt{2}} \right) + \frac{1}{100}.$
\end{restatable}

\chmades{
\input{claims/dcprob}
}

%% file: claims/dcbnd.tex
\newcommand{\startpt}{\tau}
\newcommand{\intend}{\tau + \min\left(\dgen, \ell\right)}
\newcommand{\dgcouple}{\Tilde{\dgen}}

\begin{proof}
    We prove this claim in two parts.  First, we construct a coupled process $\ncouple(t) \geq N(t)$ on the interval of interest. Then, we give an upper bound on $\ex{\int_{\startpt}^{\intend}{\ncouple(t) \dd t}\middle | \filt{\tau}}$.
Define $\ncouple(t)$ as 
\begin{equation*}
\ncouple(t) \triangleq N(\tau) + A(\tau, t) - \depOp{R-j}{(\tau, t)}.
\end{equation*}
Then, by Claim~\ref{clm:coupling}, we have that
$\ncouple(t) \geq N(t)$
on the interval of interest.
To develop the integral, we first move the minimum from the bounds of integration into the integrand. In particular, we note that the quantity $N(\dgen) - h = 0,$ and thus, for any $t > \tau + \dgen$, the quantity $N(\mintwo{\tau + \dgen}{t}) - h = 0$. 
On the other hand, for any $t < \tau + \dgen$, the quantity $N\left(\mintwo{\tau + \dgen}{t}\right) = N\left(t\right)$. 
It follows that                                   
\begin{align*}
&\int_{\startpt}^{\intend}{\left[N(t) - h \right] \dd t} \\
&\;\;= 
\int_{\startpt}^{\intend}{\left[N(\mintwo{t}{\tau+\dgen}) - h \right] \dd t}\\
&\;\;= 
\int_{\startpt}^{\intend}{\left[N(\mintwo{t}{\tau+\dgen}) - h \right] \dd t}+ \int_{\intend}^{\taul}{\left[N(\mintwo{t}{\tau+\dgen}) - h \right] \dd t}\\
&\;\;= \int_{\startpt}^{\taul}{\left[N(\mintwo{t}{\tau+\dgen}) - h \right] \dd t}\\
&\;\;\leq  \int_{\startpt}^{\taul}{\left[\ncouple(\mintwo{t}{\tau+\dgen}) - h \right] \dd t}.
\end{align*}
Defining $\dgcouple \triangleq \minpar{t > 0: \ncouple(\tau + t) \leq h}$, since $\ncouple(t) \geq N(t)$, we know both that $\dgcouple \geq \dgen$ and that, for any $t \in [\tau + \dgen, \tau + \dgcouple]$,
$
\ncouple(t) - h \geq 0.
$
Moreover, the process $V(t)$ defined as
$V(t) \triangleq \ncouple(t) - \mu j t$ is a martingale.
Thus, we have
\begin{align*}
\int_{\startpt}^{\taul}{\left[\ncouple(\mintwo{t}{\tau+\dgen}) - h \right] \dd t} 
&\leq \int_{\startpt}^{\taul}{\left[\ncouple(\mintwo{t}{\tau+\dgcouple}) - h \right] \dd t}.
\end{align*}
Taking the expectation, we find that
\begin{align}
&\ex{ \int_{\startpt}^{\taul}{\left[\ncouple(\mintwo{t}{\tau+\dgcouple}) - h \right] \dd t} \middle | \filt{\tau}}\nonumber\\
&\;\;= \int_{\startpt}^{\taul}{ \ex{\ncouple(\mintwo{t}{\tau+\dgcouple}) - h  \middle | \filt{\tau}} \dd t}\nonumber\\
 &\;\;= \int_{\startpt}^{\taul}{ \ex{V(\mintwo{t}{\tau+\dgcouple}) + \mu j \left(\mintwo{\tau + \dgcouple}{t}\right) - h  \middle | \filt{\tau}} \dd t}\nonumber\\
 &\;\;= \int_{\startpt}^{\taul}{ \ex{V(\tau) + \mu j \left(\mintwo{\tau + \dgcouple}{t}\right) - h \middle | \filt{\tau}} \dd t}\label{eq:dcb1}\\
&\;\;\leq \int_{\startpt}^{\taul}{ \ex{V(\tau) + \mu j t - h  \middle | \filt{\tau}} \dd t}\nonumber\\
&\;\;= \int_{\startpt}^{\taul}{ \ex{\ncouple(\tau)  - \mu j \tau + \mu j t - h  \middle | \filt{\tau}} \dd t}\nonumber\\
&\;\;= \left[ \ncouple(\tau) - h \right] \ell + \frac{1}{2} \mu j \ell^2,\nonumber
\end{align}
where \eqref{eq:dcb1} is an application of Doob's Optimal Stopping Theorem.
\end{proof}

%% file: claims/dcprob.tex
\begin{proof}
    We prove this result in three parts. First, we use Claim~\ref{clm:coupling} to construct a process $\ncouple(t) \geq N(t)$ on the interval of interest. Afterwards, we analyze the down-crossing probability of this coupled process. In particular, we use a reflection argument to show that
    \begin{equation*}
        \pr{\dgen < \ell} \geq 2 \pr{\ncouple(\tau + \ell) \leq h },
    \end{equation*}
    then use a Berry-Esseen bound to bound this final probability. In what follows, we focus on the lower-bound; the upper bound follows in precisely the same way. 

    To construct our coupled process, we note that, by assumption, the number of busy servers $Z(t) \geq R - j$ for any $t \in [\tau, \tau + \mintwo{\ell}{\dgen}]$. Thus, by Claim~\ref{clm:coupling}, the process $\ncouple(t)$ defined as 
    \begin{equation*}
        \ncouple(t) \triangleq N(\tau) + A(\tau, \tau + t) + \depOp{R-j}{[\tau, \tau + t]} 
    \end{equation*}
    is an upper bound for $N(t + \tau)$, i.e., $
        \ncouple(t) \geq N(\tau + t)$
    for any $t \in  [0, \mintwo{\ell}{\dgen}]$.
    By definition, we have
    \begin{align*}
        \pr{\dgen < \ell} &= \pr{\inf_{t \in [0,\ell)}{N(\tau + t)} \leq h}
        \geq \pr{\inf_{t \in [0,\ell)}{\ncouple(t)} \leq h}.
    \end{align*}
    From a reflection argument, since $\ncouple$ is upwards-biased,
    \begin{align*}
        \pr{\inf_{t \in [0,\ell)}{\ncouple(t)} \leq h} &= \pr{\inf_{t \in [0,\ell)}{\ncouple(t)} \leq h, \ncouple(\ell) < h} + \pr{\inf_{t \in [0,\ell)}{\ncouple(t)} \leq h, \ncouple(\ell) \geq h}\\
        &\geq 2\pr{\inf_{t \in [0,\ell)}{\ncouple(t)} \leq h, \ncouple(\ell) < h}\\
        &= 2\pr{ \ncouple(\ell) < h}.
    \end{align*}
    
Let $\sigma \triangleq \sqrt{\ell(2k\lambda - \mu j)}$. We now apply Now, assume that, for any $x$,
\begin{equation}\label{bebnd}
    \abs{\pr{\ncouple(\ell) < \ncouple(0) + \mu j \ell + x\sigma} - \Phi(x)} \leq \frac{0.3328}{\sigma},
\end{equation}
we have
\begin{align*}
    \pr{\ncouple(\ell) < h} &= \pr{\ncouple(\ell) < \ncouple(0) + \mu j \ell + \frac{h-\mu j \ell - \ncouple(0)}{\sigma} \cdot \sigma}\\
    &\geq \Phi\left(\frac{h-\mu j \ell - \ncouple(0)}{\sigma}\right) - \frac{1}{3\sigma}\\
    &= \Phi\left(-\frac{\left[ N(\tau) - h + \mu j \ell\right]}{\sigma}\right) - \frac{1}{3\sigma}.
\end{align*}
Putting this all together, we find
\begin{equation*}
    \pr{\dgen < \ell \middle | \filt{\tau}} \geq 2 \Phi\left(-\frac{\left[ N(\tau) - h + \mu j \ell\right]}{\sigma}\right) - \frac{2}{3\sigma}.
\end{equation*}

From here, then, it suffices to show \eqref{bebnd}. To begin, if we choose some arbitrarily large $n$ and define $$X_i \triangleq  \Pi_i' \left(\frac{k\lambda \ell}{n}\right) - \Pi_i''\left( \frac{\mu (R - j) \ell}{n}\right) - \frac{\mu j \ell}{n},$$ where each $\Pi(y)$ is an independent Poisson random variable with mean $y$,  then
$$\ncouple(\ell) =_d \sum_{i=1}^{n}{X_i} + \mu j  \ell + \ncouple(0).$$
To compute the moments of $X_i$, note that one can define centered Poisson random variables $A_i = \Pi\left(\frac{k\lambda \ell}{n}\right) - \frac{k\lambda \ell}{n}$ and $B_i = \Pi\left ( \frac{\mu (R-j) \ell}{n}\right) - \frac{\mu (R - j) \ell}{n}$, and then take $X_i = A_i - B_i$.
Doing this, one finds that
$$ \ex{X_i^2} = \ex{(A_i - B_i)^2} = \frac{k \lambda \ell}{n} + \frac{\mu (R-j) \ell}{n} = \frac{\mu(2R - j) \ell}{n}$$
and, using the triangle inequality, that
$$ \ex{\abs{X_i}^3} = \ex{\abs{A_i - B_i}^3} \leq \ex{\abs{A_i}^3} + \ex{\abs{B_i^3}} = \frac{\mu(2R - j) \ell}{n} + o\left(\frac{1}{n^2}\right).$$

We now apply the main result of \citep{shevtsova2011absolute}. Let $\sigma_n \triangleq \sqrt{\ex{X_i^2}} = \sqrt{\frac{\mu(2R - j) \ell}{n}} = \frac{\sigma}{\sqrt{n}}$ and note that $\rho_n = \ex{\abs{X_i}^3} < \sigma_n + o\left(\frac{1}{n^2}\right)$ (from \citep{normalBnd}). Then, since $\rho_n \geq 1.286 \sigma_n^3$ for sufficiently large $n$, we have 
\begin{equation*}
    \max_{x}{\abs{\pr{ \frac{\sum{X_i}}{\sqrt{n} \sigma_n} < x} - \Phi(x)}} \leq \frac{0.3328 \rho_n + 0.429 \sigma_n^3}{\sigma_n^3 \sqrt{n}} = \frac{0.3328}{\sqrt{\mu(2R - j) \ell}} + o\left(\frac{1}{n}\right).
\end{equation*}
Now noting that
\begin{equation*}
    \frac{\sum_{i=1}^{n}{X_i} }{\sqrt{n} \sigma} = \frac{\ncouple(\ell)-\ncouple(0) - \mu j \ell}{\sigma}
\end{equation*}
and taking $n \to \infty$, we have our result. \hfill \Halmos    
\end{proof}

%% file: ThesisFiles/UB/lemma_UB_int_TA.tex
\newcommand{\systag}{\textup{sys}}
\newcommand{\cfoursys}{C_4^{\systag}}

\subsection{Proof of Lemma~\ref{lem:UBintTA}, Upper Bound on Integral Over Accumulation Period}\label{sec:UBintTA}
We prove this result via two applications of the Intervening Stopping Time Lemma, Lemma~\ref{lem:waldstop}. To apply this decomposition lemma, there are two broad steps. First, we must specify a starting time ($T_0$), an ending time ($P$), a series of intervening stopping times $(T_i)$, the process ($Y_t$), and a counting variable ($F$). Second, we must prove that the three preconditions of the lemma hold, given these specifications. 

\subsubsection{First Application of Lemma~\ref{lem:waldstop}, at the Epoch Level}

\paragraph{Definition of $(\tau_j)$.} We define the sequence of stopping times $(\tau_j: j=0,1,\dots, R)$ as $\tau_j \triangleq \minpar{t>0: N(t) \leq R - j},$
i.e., $\tau_j$ is the first time there are only $R-j$ jobs within the system. Note that, by definition, $\tau_0 = 0$. We call the period $\bigg [\tau_j, \mint{\tau_{j+1}}{T_A} \bigg )$ the \emph{$j$-th epoch}, and say epoch $j$ occurs whenever $\tau_j < T_A$. We then let $\epochind$ denote the number of epochs which occur in a given renewal cycle.

\paragraph{Specification step.} 
Since we are interested in bounding $\ex{\int_0^{T_A}{[N(t) - R] \dd t}}$, we let our starting stopping time be $T_0 = 0$, our ending stopping time be $P=T_A$, our intervening stopping times be $T_j = \tau_j$, the process of interest $Y_t = N(t) - R$ and our counting variable be $F=\epochind$. 
Let the quantity
\begin{equation}\label{eq:prise}
    \prise{j} \triangleq \pr{\max_{t \in [\tau_j, \mintwo{\tau_{j+1}}{T_A}}{N(t)} \geq R + C_3\shs \middle | \epochind \geq j}
\end{equation}
be the probability that the total number of jobs $N(t)$ exceeds $R+ \threshone$ during epoch $j$. 

\paragraph{Required claims.} From here, we can apply Lemma~\ref{lem:waldstop} after showing the following claims: 

\begin{restatable}[Bound on the Probability of an Up-crossing $\prise{j}$]{claim}{prisebnd}\label{clm:prisebnd}
    Let $\prise{j}$ be the probability that the total number of jobs $N(t)$ exceeds $R+ \threshone$ during epoch $j$ defined in \eqref{eq:prise}. 
    Then, for any epoch $j \geq A_5 \sqrt{R}$,  we have $\prise{j} \geq 0.99\frac{A_5}{\sqrt{R}}.$
\end{restatable}

\begin{restatable}[Upper Bound on the Probability of Another Epoch]{claim}{atwo}\label{clm:a2}
 Recall that the total number of epochs $\epochind \triangleq \maxpar{j\in \zplus: \tau_j < T_A}$. Then, taking $C_4= 0.98$, we have $\pr{ \epochind \geq j+1 \middle | \epochind \geq j} \leq 1 - C_4\prise{j}.$
\end{restatable}

\begin{restatable}[Upper Bound on the Integral Over an Epoch]{claim}{aone}\label{clm:a1}
Let $\tau_j \triangleq \minpar{t\geq0: N(t) \leq R - j}$, $T_A \triangleq \minpar{t \geq0: Z(t) = R+1}$, and let $\epochind \triangleq \maxpar{i\in \zplus: \tau_i < T_A}$. Then,
    \begin{equation*}
        \ex{\int_{\tau_j}^{\mintwo{\tau_{j+1}}{T_A}}{[N(t) - R] \dd t}\middle | \epochind \geq j} \leq B_1 \shs \cdot \ex{\mintwo{\tau_{j+1}}{T_A} - \tau_j \middle | \epochind \geq j} + C_2 \beta^2 \mu j \prise{j} ,
    \end{equation*}
    where $B_1=3.6$ and $C_2= \frac{1}{2\cdot 0.98} > 0.511$.
\end{restatable}

\subsubsection{Proof of Lemma~\ref{lem:UBintTA} Assuming Claims~\ref{clm:prisebnd}, ~\ref{clm:a2} and~\ref{clm:a1}.}
Before going further, we show how to complete the proof of Lemma~\ref{lem:UBintTA}, assuming the three prior claims. Applying Lemma~\ref{lem:waldstop}, we find that
\begin{align*}
    \ex{\int_0^{T_A}{[N(t) - R] \dd t}} &\leq B_1 \shs \cdot \ex{T_A} + C_2 \beta^2 \mu \sum_{j=1}^{R}{j \prise{j} \prod_{i=1}^{j-1}{\left(1 - C_4 \prise{j} \right)}}\\
    &\leq B_1 \shs \cdot \ex{T_A} + \frac{C_2}{C_4} \beta^2 \mu \left[\sum_{j=1}^{R}{j C_4 \prise{j} \prod_{i=1}^{j-1}{\left(1 - C_4 \prise{j} \right)}} \right]\\
    &= B_1 \shs \cdot \ex{T_A} + \frac{C_2}{C_4} \beta^2 \mu \left[\sum_{j=1}^{R}{\prod_{i=1}^{j}{\left(1 - C_4 \prise{j} \right)}} \right],
\end{align*}
where we have used the ``expectation as a sum of tails'' trick. 

Applying Claim~\ref{clm:prisebnd}, we find that
\begin{equation}\label{eq:UBintTA_fin1}
    \sum_{j=1}^{R}{\prod_{i=1}^{j}{\left(1 - C_4 \prise{j} \right)}} \leq  \sum_{j=1}^{R}{\left(1 - \frac{0.99 C_4 A_5}{\sqrt{R}}\right)^{[j - A_5 \sqrt{R}]^+}} \leq \sum_{j=1}^{\infty}{\left(1 - \frac{0.99 C_4 A_5}{\sqrt{R}}\right)^{[j - A_5 \sqrt{R}]^+}}.
\end{equation}
Bounding this as a Geometric sum, we obtain
\begin{align*}
    \eqref{eq:UBintTA_fin1} = A_5 \sqrt{R} + \sum_{j=0}^{\infty}{\left(1 - \frac{0.99 C_4 A_5}{\sqrt{R}}\right)^j}= A_5 \sqrt{R} + \frac{1}{0.99 C_4 A_5} \sqrt{R}.
\end{align*}
Returning to our original inequality, we obtain that
\begin{align*}
    \ex{\int_0^{T_A}{[N(t) - R] \dd t}} \leq B_1 \shs \cdot \ex{T_A} + \frac{C_2}{C_4}\left( A_5 + \frac{1}{0.99 C_4 A_5}\right) \beta^2 \mu \sqrt{R}.
\end{align*}
Noting that $A_5 =  1$ and taking $B_2\triangleq 1.04 > \frac{C_2}{C_4}\left( A_5 + \frac{1}{0.99 C_4 A_5}\right)$, we finish the proof of Lemma~\ref{lem:UBintTA}. \hfill \Halmos

\subsubsection{Proof of Claim~\ref{clm:prisebnd}}\label{sec:prisebnd}
From the above, assuming the preconditions of Lemma~\ref{lem:waldstop} (Claims~\ref{clm:a2} and~\ref{clm:a1}) as well as the helper claim Claim~\ref{clm:prisebnd}, we have proven Lemma~\ref{lem:UBintTA}. It thus suffices to prove these three claims. We begin with Claim~\ref{clm:prisebnd}.


\input{claims/prisebnd}


\subsubsection{Proof of Claim~\ref{clm:a2}, Upper Bound on Probability of Another Epoch} 

\paragraph{Rewriting the claim.}
We begin by rewriting the probability of another epoch occurring as
\begin{align*}
    \pr{ \epochind \geq j+1 \middle | \epochind \geq j} &= 1 - \pr{\epochind = j \middle | \epochind \geq j} = 1 - \pr{T_A < \tau_{j+1} \middle | \epochind \geq j}.
\end{align*}
It thus suffices to show a bound on the probability that the accumulation phase ends in epoch $j$:
\begin{equation}\label{eq:UBintTA_prob1}
    \pr{T_A < \tau_{j+1} \middle | \epochind \geq j} \geq C_4 \prise{j}.
\end{equation}
\paragraph{Lower bound based on up-crossing and down-crossing times.} To show \eqref{eq:UBintTA_prob1}, we analyze a particular sequence of events which results in the accumulation phase ending in the current epoch, i.e. $T_A < \tau_{j+1}$. 
Specifically, we define the up-crossing time $u =\minpar{t >\tau_j: N(t) \geq R+\threshone}$ and the down-crossing time $d = \minpar{t > u: N(t) \leq R}$. 
We consider the event where (1) the up-crossing occurs during the accumulation phase ($u <T_A$) and (2) the accumulation phase ends before the next down-crossing occurs ($d > T_A$).
Symbolically, we have (at the end, recalling that $\prise{j}$ is the probability of an up-crossing occurs)
\begin{align*}
    \eqref{eq:UBintTA_prob1} \geq \pr{u < T_A < d \middle | \epochind \geq j} = \pr{d > T_A \middle | u < T_A}\pr{u <T_A} = \pr{d > T_A \middle | u < T_A} \prise{j}.
\end{align*}

\paragraph{Development of conditional probability.} To bound the conditional probability $\pr{d > T_A \middle | u < T_A}$, we condition on the filtration at time $u$, then make a coupling argument.
To begin, note that, if the number of jobs $N(t)$ does not fall to $R$ before the $(R+1)$-th server finishes setting up, then the accumulation time $T_A$ occurs exactly when the $(R+1)$-th server finishes, i.e. the accumulation time $T_A = u + \rem{R+1}(u)$. 
Furthermore, the number of busy servers $Z(t) \leq R$ at any time during the accumulation phase $t < T_A$.
Applying a basic coupling argument (Claim~\ref{clm:coupling}), we have a lower bound on $N(t)$ in the coupled process
\begin{equation*}
    \ncouple(t) \triangleq N(u) + \Pi_A\left( (u, t] \right) - \depOp{R}{(u, t]},
\end{equation*}
for any time $t \in [u, T_A]$.
Let $\dcouple \triangleq \minpar{t>u:\ncouple(t) \geq R}$ be the analogous down-crossing time in the coupled system.
Since the coupled $\ncouple(t)$ is a lower bound, the coupled down-crossing time $\dcouple \leq d$. Thus,
\begin{equation}\label{eq:UBintTA_prob2}
    \pr{d > T_A \middle | \filt{u}, u < T_A}
 = \pr{ d > u + \rem{R+1}(u) \middle | \filt{u}, u < T_A }
    \geq \pr{\dcouple > u + \rem{R+1}(u) \middle | \filt{u}, u < T_A}.
\end{equation}

\paragraph{Analyzing the coupled probability.}
Continuing, the probability that $\left\{\dcouple \geq \ell\right\}$ is decreasing in $\ell$. Thus,
\begin{align*}
    \eqref{eq:UBintTA_prob2} \geq \pr{\dcouple > u + \beta \middle | \filt{u}, u < T_A} 
    = \pr{\dcouple - u > \beta}  \geq 1 - 2\Phi\left(-\frac{C_3}{\sqrt{2}}\right) - \frac{2}{3\sqrt{ \mu \beta R}}\geq 0.98,
\end{align*}
where in the final inequalities we have applied both the down-crossing probability bound of Claim~\ref{clm:dcprob} and our assumptions.
Taking $C_4\triangleq 0.98$, we have the inter-epoch probability bound of Claim~\ref{clm:a2}. \hfill \Halmos



\subsubsection{Proof of Claim~\ref{clm:a1}, Upper Bound on the Integral Over an Epoch.}\label{sub:updown}
We now prove Claim~\ref{clm:a1}, the upper bound on the time integral over an epoch.
We do this via another application of Lemma~\ref{lem:waldstop} ---first specifying the intervening times, then completing the proof, then proving that the preconditions hold.

\begin{figure}
    \centering
    \includegraphics[scale=0.25,clip]{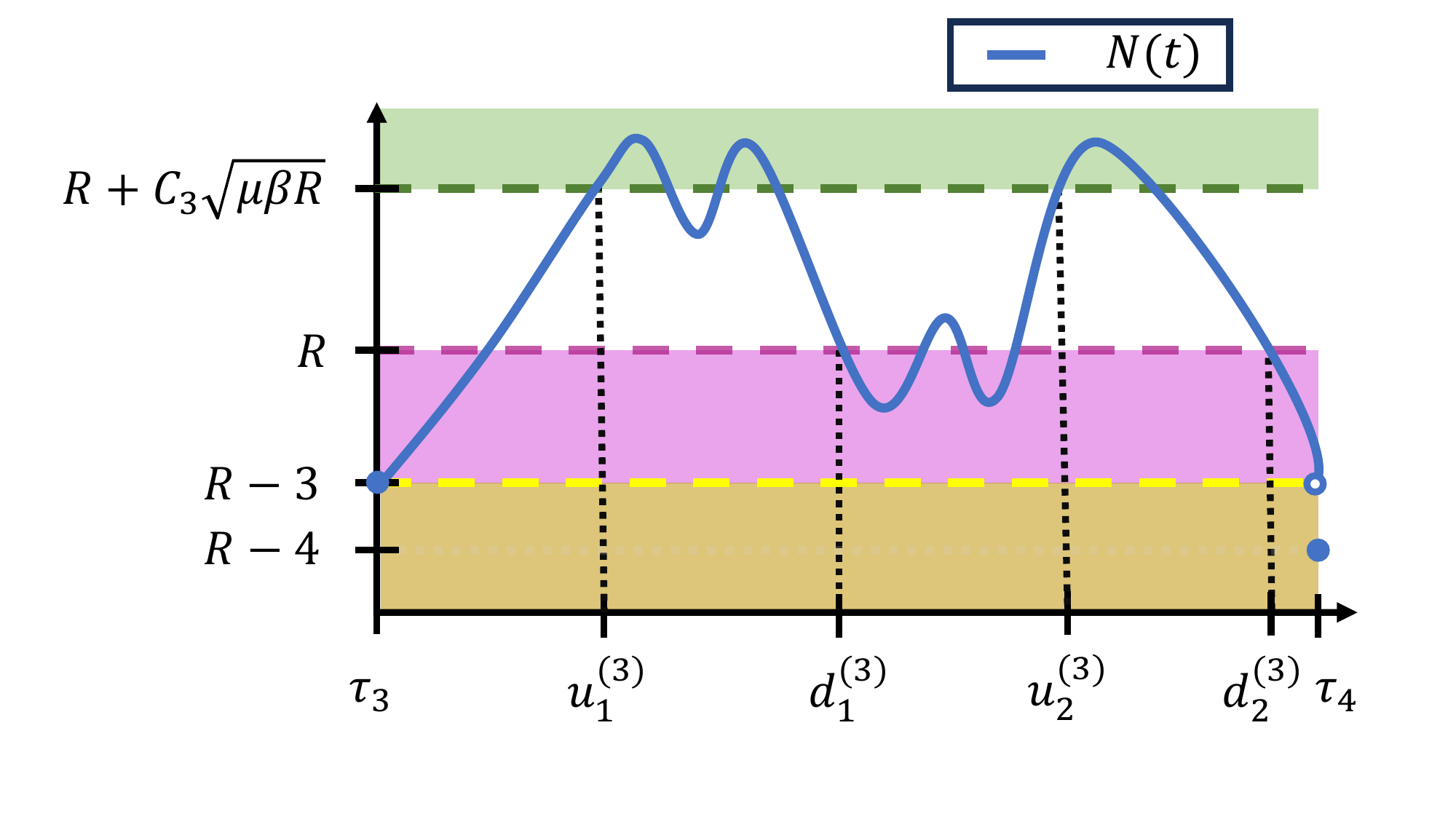}
    \vspace{-3pt}
    \caption{A depiction of the up-crossings and down-crossings defined in Section~\ref{sub:updown}. In this example, we see that the number of up-crossings in epoch $3$ is $\epochind^{(3)} = 2$ and that, in this case, epoch $3$ ends when epoch $4$ begins (i.e. at time $\tau_4$).}
    \label{fig:updown}
\end{figure}

\paragraph{Definition of up-crossings and down-crossings.} Let the $0$-th down-crossing time in epoch $j$ occur at time $\tau_j$, i.e. let $d_0^{(j)} \triangleq \tau_j.$
Next, define the first up-crossing in epoch $j$ as the first time during epoch $j$ that the total number of jobs $N(t)$ exceeds $ R+ \threshone$, i.e.
\begin{equation*}
    u_1^{(j)} \triangleq \minpar{t > \tau_j : N(t) \geq R + \threshone}.
\end{equation*}
From here, define $i$-th down-crossing in epoch $j$ and the $i+1$-th up-crossing in epoch $j$ as
\begin{gather*}
    d_i^{(j)} \triangleq \minpar{t \geq u_i^{(j)} : N(t) \leq R},\\
    u_{i+1}^{(j)} \triangleq \minpar{t \geq d_i^{(j)} : N(t) \geq R+ \threshone},
\end{gather*}
respectively; we visualize these definitions in Figure~\ref{fig:updown}. We say the $i$-th up-crossing occurs if $u_i^{(j)} < \mintwo{T_A}{\tau_{j+1}}$ and let 
$\upind \triangleq \maxpar{i \geq 0: u_i^{(j)} < \mintwo{\tau_{j+1}}{T_A}}$
be the random number of up-crossings which occur in epoch $j$. 
We call the interval $\bigg [ d_i^{(j)}, \mintwo{u_i^{(j)}}{\mintwo{T_A}{\tau_{j+1}}} \bigg)$ the $i$-th rise, and the interval $\bigg [u_i^{(j)}, \mintwo{d_i^{(j)}}{\mintwo{T_A}{\tau_{j+1}}}\bigg)$ the $i$-th fall.  
Note that, if the $i$-th up-crossing occurs, then, by definition, $d_i < \tau_{j+1}$; this means that the $i$-th fall can always be written as $\bigg[u_i, \mintwo{T_A}{d_i} \bigg )$. 
For readability, we fix our epoch of interest and freely omit the superscript $j$ on our up-crossings and down-crossings.

\paragraph{Specification step.} With up-crossings and down-crossings defined, we are now ready to specify our application of the IST Lemma,  Lemma~\ref{lem:waldstop}. We define our starting time as $T_0 = \tau_j = d_0$, our ending time as $P = \mintwo{T_A}{ \tau_{j+1}}$, our intervening sequence as $\left(u_i\right)_{i=1}^{\infty}$, and our counting variable as $F = \upind$. 

\paragraph{Required claims.} From here, in order to apply Lemma~\ref{lem:waldstop}, we must show the following three claims:

\begin{restatable}[Upper Bound on Integral Until First Up-crossing]{claim}{aonei}\label{clm:updowninit}
    Taking $B_1 = 3.6$, the integral until $u_1$ is
\begin{equation*}
        \ex{\int_{d_0}^{\mintwo{u_1}{\mintwo{T_A}{\tau_{j+1}}}}{[N(t) - R] \dd t} \middle | \upind \geq j} \leq B_1  \shs \cdot \ex{\mintwo{u_1}{\mintwo{T_A}{\tau_{j+1}}} - \tau_j \middle | \upind \geq j}.
\end{equation*}

\end{restatable}

\begin{restatable}[Upper Bound on Integral Between Up-crossings]{claim}{aoneii}\label{clm:updowncont}
    The integral between up-crossings $u_i$ is
\begin{align*}
      \ex{\int_{u_i}^{\mintwo{u_{i+1}}{\mintwo{T_A}{\tau_{j+1}}}}{[N(t) - R] \dd t} \middle | \upind \geq i} &\leq B_1 \shs \cdot  \ex{\mintwo{u_{i+1}}{\mintwo{T_A}{\tau_{j+1}}} - u_i \middle | \upind \geq i} \\
      &\;\;\;\; + \frac{1}{2} \beta^2 \mu j.
\end{align*}
\end{restatable}

\begin{restatable}[Upper Bound on Probability of Another Up-crossing]{claim}{aoneiii}\label{clm:updownprob}
    Recall that $\prise{j}$ is the probability that the number of jobs $(N(t)\geq \threshone$ at some point during epoch $j$, given that epoch $j$ occurs. 
    Then, $\pr{\upind > 0} = \prise{j}$, and, for all counts $i \geq 1$ and $p_2 = 0.98$, we have $\pr{\upind \geq i+1 \middle | \upind \geq i} \leq 0.02 = 1 - p_2.$
\end{restatable}

\subsubsection{Proof of Claim~\ref{clm:a1}, Assuming Claims~\ref{clm:updowninit},~\ref{clm:updowncont}, and~\ref{clm:updownprob}.} 
Once again, before we move on to proving these claims, we show that they indeed suffice to prove Claim~\ref{clm:a1}. 
By Lemma~\ref{lem:waldstop}, taking $C_2 \triangleq \frac{0.5}{p_2},$
\begin{align*}
    &\ex{\int_{\tau_j}^{\mintwo{T_A}{\tau_{j+1}}}{[N(t) - R] \dd t}\middle | \epochind \geq j}\\ 
    &\;\;\leq B_1  \shs \cdot \ex{{\mintwo{T_A}{\tau_{j+1}}} - \tau_j \middle | \epochind \geq j}+ \prise{j} 0.5 \beta^2 \mu j \sum_{i=1}^{\infty}{   \left(1 - p_2\right)^{i-1}}\\
    &\;\;= B_1  \shs \cdot \ex{{\mintwo{T_A}{\tau_{j+1}}} - \tau_j \middle | \epochind \geq j} + \prise{j} \frac{0.5}{p_2} \beta^2 \mu j. \tag*{\Halmos}
\end{align*}

\subsubsection*{Proofs of Claims~\ref{clm:updowninit},~\ref{clm:updowncont}, and~\ref{clm:updownprob}.}
All that remains to be proven are our three aforementioned claims. 
\paragraph{Proof of Claim~\ref{clm:updowninit}: Upper Bound on Integral until First Up-crossing.}
Proving this claim is quite simple. 
In fact, we now prove a far more general claim, that the integral from a down-crossing to the next up-crossing
\begin{equation}\label{eq:downbnd}
    \int_{d_i}^{\mintwo{u_i}{\mintwo{T_A}{\tau_{j+1}}}}{[N(t) - R] \dd t}\leq \threshone \cdot \left[\mintwo{T_A}{\tau_{j+1}} - d_i \right].
\end{equation}
To see this, note that, at any point between a down-crossing and up-crossing, the total number of jobs $N(t)$ must be strictly less than $R+\threshone$. Apply this to the $0$-th down-crossing and we have the claim. \hfill \Halmos

\paragraph{Proof of Claim~\ref{clm:updowncont}: Upper Bound on Integral Between Up-crossings.}
This proof is a bit more involved. 
We separate the interval $\bigg[u_i, \mintwo{u_{i+1}}{\mintwo{T_A}{\tau_{j+1}}} \bigg )$ into the $i$-th fall and the $i$-th rise, as discussed previously.  For the rising portion, we can simply apply the simple bound from $\eqref{eq:downbnd}$. For the falling portion, we apply the integral coupling claim, Claim~\ref{clm:dcbnd}. In particular, note that $Z(t)\geq R - j$ until time $\tau_{j+1}$ and that the interval $[u_i, \mintwo{d_i}{T_A})$ is equivalent to the interval $[u_i, \mintwo{d_i}{u_i + Y_{R+1}(u_i)})$. Applying Claim~\ref{clm:dcbnd},
\begin{align*}
    \ex{ \int_{u_i}^{\mintwo{d_i}{T_A}}{[N(t) - R] \dd t}\middle | \state(u_i)}
    &\leq  \frac{1}{2}\beta^2 \mu j + \left[N(u_i) - R\right] \cdot \rem{R+1}(u_i)\\
    &= \frac{1}{2}\beta^2 \mu j + \threshone \cdot \rem{R+1}(u_i).
\end{align*}
By our analysis in the proof of Claim~\ref{clm:a2}, we note that the remaining setup time $\rem{R+1}(u_i) \leq  \mintwo{d_i}{T_A}$  with probability at least $p_2 \leq \min_{\filt{u_i}}{\pr{d_i < T_A \middle |\filt{u_i}, \upind \geq i}}.$ By Markov's inequality,
\begin{equation*}
    \rem{R+1}(u_i) \leq \frac{1}{p_2} \ex{\mintwo{d_{i}}{T_A} - u_i \middle | \state(u_i)}.
\end{equation*}
Combining our bounds on the rises and falls and taking $B_1 = \frac{C_3}{p_2}$, we have Claim~\ref{clm:updowncont}. \hfill \Halmos

\paragraph{Proof of Claim~\ref{clm:updownprob}: Upper Bound on Probability of Another Up-Crossing.}
We now proceed to our final claim, concerning the up-crossing probabilities $\pr{\upind > 0}$ and $\pr{\upind \geq i+1 \middle | \upind \geq i}$.
To begin, we first note that, since the first up-crossing occurs precisely at the moment that $N(t)$ exceeds $R+\threshone$ during epoch $j$, one has 
$\prise{j} = \pr{u_1 < \mintwo{T_A}{\tau_{j+1}}} = \pr{\upind > 0}$.

To prove the second part of the claim, first observe that an $(i+1)$-th up-crossing can only occur if an $i$-th down-crossing occurs, i.e.
$\pr{\upind \geq i+1 \middle | \upind \geq i} \leq \pr{d_i < T_A \middle | \upind \geq i}.$

To bound this conditional probability, we can apply a previous result.
Recall the proof of the inter-epoch probability bound (Claim~\ref{clm:a2}).
In \eqref{eq:UBintTA_prob1}, we have already argued a bound on the conditional probability that the \textit{first} down-crossing occurs, \textit{in a state-independent manner.}
The bound derived there thus \textit{also applies here}:
\begin{equation*}
    \pr{d_i < T_A \middle | \upind \geq i}= 1- \pr{d_i > T_A \middle | \upind \geq i} \leq 1- C_4.
\end{equation*}
Taking $p_2 \triangleq C_4 = 0.98$, we have bounded the probability of another up-crossing (Claim~\ref{clm:updowncont}). \hfill \Halmos



%% file: claims/prisebnd.tex
To prove Claim~\ref{clm:prisebnd}, we show a more general claim: that, for $j \geq A_5 \sqrt{R}$,
\begin{equation}\label{eq:prisept1}
    \prise{j} \geq 0.99\frac{j}{R}.
\end{equation}

\subsubsection{\torp{Proof of \eqref{eq:prisept1}: Lower Bound on $\prise{j}$.}{Lower Bound on p\_rise\^j}}
We begin with a simple probability manipulation: 
\begin{align}
    \prise{j} &\triangleq \pr{\text{$N(t) \geq C_3 \shs$ at some point during epoch $j$}\middle | \filt{\tau_j}}\nonumber\\
    &\geq \pr{\text{$N(t) \geq C_3 \shs$ during the interval $\left[\tau_j, \mintwo{\tau_j + \beta}{\tau_{j+1}} \right]$}\middle | \filt{\tau_j}}.\nonumber
\end{align}
From here, we make with a useful observation: since there are no server \emph{in setup} at the beginning of an epoch (as we have just turned off a server), no servers can complete setup in the first $\setuptime$ time of an epoch.
Thus, the number of busy servers $Z(t) \leq R - j$ during this time, and, by Claim~\ref{clm:coupling}, the coupled process
\begin{equation*}
    \ncouple(t) \triangleq N(\tau_j) + A(\tau_j, t) - \depOp{R-j}{(\tau_j, t)} 
\end{equation*}
must be a lower bound on $N(t)$, during the interval $[\tau_j, \tau_j + \beta]$. Moreover, the number of busy servers $Z(t)$ can not be smaller than $R-j$ until the beginning of epoch $j+1$ either. Thus, we find that the behavior of $N(t)$ corresponds \emph{exactly} with the behavior of $\ncouple(t)$ during the interval $[\tau_j, \mintwo{\tau_{j+1}}{\tau_j + \setuptime}]$.

We now use this coupled process to analyze our original probability. Define the up-crossing time $\tup$ as
\begin{equation*}
    \tup \triangleq \minpar{t > 0: \ncouple(\tau_j + t) \geq R + C_3 \shs}.
\end{equation*}
Likewise, define the down-crossing time $\tdown$ as 
\begin{equation*}
    \tdown \triangleq \minpar{t > 0: \ncouple(\tau_j + t) \leq R - (j+1)}.
\end{equation*}
It follows that
\begin{align*}
    &\pr{\text{reach $N(t) \leq R - (j+1)$ during the interval $[\tau_j, \mintwo{\tau_j + \beta}{\tau_{j+1}}]$}\middle | \filt{\tau_j}}\\
    &\;\; = \pr{\text{reach $\ncouple(t - \tau_j) \leq R - (j+1)$ during the interval $[\tau_j, \mintwo{\tau_j + \beta}{\tau_{j+1}}]$}\middle | \filt{\tau_j}}\\
    &\;\;= \pr{\tup \leq \beta, \tup < \tdown}\\
    &\;\;= \pr{\tup \leq \beta} - \pr{\tup \leq \beta, \tup \geq \tdown}\\
    &\;\; = \pr{\tup \leq \beta} - \pr{\tup \leq \beta \middle | \tup \geq \tdown}\pr{\tup \geq \tdown}.
\end{align*}
We now observe that 
\begin{equation}\label{eq:transinv}
    \pr{\tup \leq \beta \middle | \tup \geq \tdown} \leq \pr{\tup \leq \beta},
\end{equation}
since the process has farther to go, less time to do so, and the process's behavior is translation-invariant (this last point is why we needed to analyze the coupled process instead).

Continuing from where we left off, we find that
\begin{align*}
    \prise{j} &= \pr{\tup \leq \beta} - \pr{\tup \leq \beta \middle | \tup \geq \tdown} \pr{\tup \geq \tdown}\\
    &\geq \pr{\tup \leq \beta} - \pr{\tup \leq \beta} \pr{\tup \geq \tdown}\\
    &= \pr{\tdown > \tup} \pr{\tup \leq \beta}\\
    &\geq \pr{\tdown > \infty} \pr{\tup \leq \beta}\\
    &= \frac{j}{R} \pr{\tup \leq \beta},
\end{align*}
where the last equality is a classical result on upwards-biased discrete random walks (one can think of $\ncouple$ as a discrete random walk driven by a Poisson process of rate $(k \lambda +\mu(R-j)$, where the probability that $\ncouple$ increases at a Poisson event is $\frac{k\lambda}{k\lambda + \mu(R-j)} = \frac{R}{2R - j}$).

From here, it suffices to lower bound $\pr{\tup \leq \beta}$. To begin, note
\begin{align*}
    \pr{\tup \leq \beta} & = \pr{\sup_{t \in [0,\beta)}{\ncouple(t)} \geq R + C_3 \shs}\\
    &\geq \pr{ \ncouple (\beta) \geq R + C_3 \shs}\\
    &= \pr{A(\tau_j, \tau_j + \beta) - \depOp{R-j}{\tau_j, \tau_j + \beta} \geq j + C_3 \shs}.\intertext{Noting that the number of arrivals $A(\tau_j, \tau_j + \beta)$ and the number of departures $\depOp{R-j}{[\tau_j, \tau_j + \beta]}$ are independent Poisson r.v.'s, we can apply the Berry-Esseen bound of Claim~\ref{clm:berryess} to find}
    &= 1 - \Phi\left(\frac{ \mu \beta j - j - C_3 \shs}{\sqrt{\mu \beta (2R - j)}}\right) - \frac{1}{3 \sqrt{\mu \beta (2R - j)}}\\
    &\geq 1 - \Phi\left(\frac{ 0.99 \mu \beta j- C_3 \shs}{\sqrt{2 \mu \beta R}}\right) - \frac{1}{3 \sqrt{\mu \beta R}}\\
    &= 1 - \Phi\left(-0.99\frac{j}{\sqrt{R}}\sqrt{\mu \beta} + \frac{C_3}{\sqrt{2}}\right) - \frac{1}{3 \shs}\\
    &\geq 1 - \Phi\left(-9.9 A_5 + \frac{C_3}{\sqrt{2}}\right) - \frac{1}{300}.
\end{align*}

To complete the proof, we set $A_5$ such that the final probability is $\geq 0.99$. In particular, we need
\begin{equation*}
    \Phi\left(-9.9 A_5 + \frac{C_3}{\sqrt{2}}\right) \leq \frac{2}{300},     
\end{equation*}
which is achieved when $A_5 > \frac{C_3}{9.9 \sqrt{2}} + 0.25$; choosing $A_5 = 1$ gives the result. \hfill \Halmos

%% file: ThesisFiles/UB/lemma_UB_int_TB.tex
\newcommand{\visit}[1]{v_{#1}}
\newcommand{\threshtwo}{C_4 \shs}
\newcommand{\vind}{n_V}
\newcommand{\upb}[1]{\upsilon^{(\text{up})}_{#1}}
\newcommand{\downb}[1]{\upsilon^{(\text{down})}_{#1}}
\newcommand{\mb}{M_B}
\newcommand{\tbind}{n_{b}}
\newcommand{\ntr}{\left[N(t) - R\right]}
\newcommand{\ibpfn}[3]{g\left(#1,#2,#3\right)}
\newcommand{\donesys}{D_1 }
\newcommand{\upbnda}{\frac{b_2}{\mu\sqrt{R}}}
\newcommand{\downbnd}{\frac{ b_1 \sqrt{\mu \beta}}{\mu \sqrt{R}} + \frac{6}{\mu R}}
\newcommand{\tbbnd}{\frac{1}{\mu}\left[\upbnda + \downbnd + \frac{1}{\mb}\right]  }
\newcommand{\ntk}{\left[ N\left(T_A\right) - k\right]}
\newcommand{\ntauk}{\left[ N\left(\tau\right) - k\right]}
\newcommand{\dgencouple}{\Tilde{\dgen}}
\newcommand{\downmod}{ b_1 \sqrt{ \mu \beta R}}
\newcommand{\nta}{N\left(T_A\right)}
\newcommand{\numserv}[1]{\textup{ns}\left(#1\right)}

\subsection{Proof of Lemma~\ref{lem:UBintTB}, Upper Bound on Integral Over Draining Period}\label{sec:UBintTB}


To prove this lemma, we again make use of Lemma~\ref{lem:waldstop}. 
We proceed through the usual two-step process, first defining the stopping time sequence we will analyze over, 
then proving the preconditions of the lemma.


\paragraph{Definition of the upward visits $\upb{i}$ and downward visits $\downb{i}$.} 
Recall that the draining phase begins at time $T_A$. 
Let $\mb \triangleq \min\left(k-R, \sqrt{R})\right)$ be a specially-set analysis threshold. 
Let the stopping time $\dnb{1} \triangleq \minpar{t \geq T_A: N(t) < R + \mb}$ be the first time the number of jobs $N(t)$ drops below $R +\mb$, and recursively define
\begin{align*}
    \unb{i} &\triangleq \minpar{ t\geq \dnb{i}: N(t) \geq R + \mb},\\
    \dnb{i+1} &\triangleq \minpar{ t \geq \unb{i}: N(t) < R+ \mb}.
\end{align*}

\subsubsection{Application of Lemma~\ref{lem:waldstop}.}
Applying Lemma~\ref{lem:waldstop}, we take our initial stopping time to be the accumulation time $T_A$, our final stopping time to be the end of the renewal cycle $X$, our intervening stopping times to be the downward visits $\dnb{i}$, and our counting index to be $\tbind$. 

\paragraph{Required claims.} To apply Lemma~\ref{lem:waldstop}, we need to show is the usual three claims: a bound on the initial integral, a bound on the continuing integral, and a bound on the probability. 
\begin{restatable}[Upper Bound on Integral Until First Downward Visit]{claim}{UBintTBinitial}\label{clm:UBintTBinit}
    Let $g(x,y,z) \triangleq  x \frac{1}{2\mu z} + y\left[\frac{R}{\mu z^2} + \frac{3}{2 \mu z} \right].$ Then, one can bound the integral immediately after time $T_A$ with
    \begin{align*}
        \ex{\int_{T_A}^{\dnb{1}}{\left[N(t) - R \right] \dd t}} 
    &\leq \left[\beta + \frac{1}{\mu}\right] \left[3 \mu \beta \sqrt{R} + \max\left(\sqrt{R}, \frac{\rho}{1-\rho}\right)\right] + \frac{2}{\mu}\ln\left(3 \frac{\mu \beta}{\sqrt{R}} \max\left(\sqrt{R}, \frac{\rho}{1-\rho}\right) \right)\\
    &\;\;\;+  g\left((9(\mu \beta)^2 R, 3 \mu \beta \sqrt{R}, k(1-\rho)\right) + \beta \frac{\rho}{1-\rho}.
    \end{align*}
\end{restatable}


\begin{restatable}[Upper Bound on Integral Between Downward Visits]{claim}{UBintTBbetween}\label{clm:UBintTBbtwn}
    The integral between consecutive downward visits is bounded as
    \begin{equation*}
        \ex{\int_{\dnb{i}}^{\dnb{i+1}}{\ntr \dd t}\middle | \filt{\dnb{i}}}\leq \frac{1}{\mu \mb} \left[ \max\left(\sqrt{R}, \frac{\rho}{1-\rho}\right) + 14 + \mu \beta \sqrt{R}  + 2 b_1\sqrt{\mu \beta R} + b_2\sqrt{R}\right].
    \end{equation*}
\end{restatable}

\begin{restatable}[Upper Bound on Probability of Another Downward Visit]{claim}{UBintTBprob}\label{clm:UBintTBprob}
    The probability of another downward visit occurring is bounded by
    \begin{equation}\label{eq:succ_TB_prob}
        \pr{\tbind \geq i+1 \middle | \tbind \geq i} \leq \frac{1}{\mb}.
    \end{equation}
\end{restatable}


\subsubsection{Proof of Lemma~\ref{lem:UBintTB} Assuming Claims~\ref{clm:UBintTBinit},~\ref{clm:UBintTBbtwn}, and~\ref{clm:UBintTBprob}.}
Simplifying the first bound further,
\begin{align*}
    \ex{\int_{T_A}^{\dnb{1}}{[N(t) - R] \dd t}} &\leq \left[\beta + \frac{3}{\mu} \right]\left[2.9\mu \beta \sqrt{R} + \max\left(\sqrt{R}, \frac{\rho}{1-\rho}\right) \right] \\
    &\mspace{21mu}+ g\left(9(\mu\beta)^2 R, 3 \mu \beta \sqrt{R}, k(1-\rho)\right) + \beta \frac{\rho}{1-\rho}\\
    &\leq 1.03 \beta \left[2.91\mu \beta \sqrt{R} + \frac{\rho}{1-\rho} \right] + g\left(9(\mu\beta)^2 R, 3 \mu \beta \sqrt{R}, k(1-\rho)\right) + \beta \frac{\rho}{1-\rho}\\
    &\leq 3 \mu \beta^2 \sqrt{R} + 2.03 \beta \frac{\rho}{1-\rho} + g\left(9(\mu\beta)^2 R, 3 \mu \beta \sqrt{R}, k(1-\rho)\right),
\end{align*}
where we have bounded $\ln(x)/x <0.1$ for values $x > 100$, noted that $\max(x,y) \leq x + y,$ and done some upwards rounding.
Simplifying the ``continuing integral'' term, we have
\begin{align*}
    \ex{\int_{\dnb{1}}^{X}{\ntr \dd t}}  &\leq \frac{1}{\textup{Claim~\ref{clm:UBintTBprob}}} \cdot \left[\textup{Claim~\ref{clm:UBintTBbtwn}}\right]\\
    &= \frac{1}{\mu}\left[ \max\left(\sqrt{R}, \frac{\rho}{1-\rho}\right) + 14 + \mu \beta \sqrt{R}  + 2 b_1\sqrt{\mu \beta R} + b_2\sqrt{R}\right]\\
    &\leq \frac{1}{\mu}\left[\mu \beta\sqrt{R}\left(1 + \frac{2b_1}{\sqrt{\mu\beta}} + \frac{b_2}{\beta} + \frac{14}{\mu \beta \sqrt{R}} \right] + \frac{\rho}{1-\rho}\right]\\
    &\leq \frac{1}{\mu}\left[ 1.6\mu \beta \sqrt{R} + \frac{\rho}{1-\rho}\right].
\end{align*}
Combining these two pieces, we obtain, as desired,
\begin{equation*}
    \ex{\int_{T_A}^{X}{\left[N(t) - R\right] \dd t}} \leq 3.01\mu\beta^2 \sqrt{R} + 2.04 \beta \frac{\rho}{1-\rho} + g\left(9(\mu\beta)^2 R, 3 \mu \beta \sqrt{R}, k(1-\rho)\right). \tag*{\Halmos}
\end{equation*}

\subsubsection*{Precursor: The ``Wait-Busy'' Idea.}\label{sec:waitbusy}
As such, to complete our proof it suffices to show Claims~\ref{clm:UBintTBinit},~\ref{clm:UBintTBbtwn}, and~\ref{clm:UBintTBprob}.
To prove these claims, we make heavy use of the following idea.

\begin{restatable}[Wait Busy Claim]{claim}{waitbusy}\label{clm:waitbusy2}
Let $\tau$ be some stopping time, let the number of jobs $N(\tau) = R+h$, and define $\numserv{h} \triangleq \min\left\{h, k(1-\rho)\right\}$. Let the down-crossing $\dgen\triangleq \minpar{t>0: N(\tau + t) = R + h - 1}.$ If $Z(\tau) \geq R$, then
\begin{equation}\label{eq:waitbusy1}
    \ex{\int_\tau^{\tau + \dgen}{[N(t) - (R+h - 1)] \dd t} \middle | \filt{\tau}} \leq \rem{R+ \numserv{h}}(\tau) + g\left(1 + 2\mu R \ex{\mintwo{\rem{R+\numserv{h}(\tau)}}{\dgen}}, 1, \mu \numserv{h}\right),
\end{equation}
where the function $g(x,y,z) \triangleq  x \frac{1}{2\mu z} + y\left[\frac{R}{\mu z^2} + \frac{3}{2 \mu z} \right].$

Furthermore, 
\begin{align}
    \ex{\int_{T_A}^{\dnb{1}}{[N(t) - R] \dd t}\middle | \filt{T_A}} &\leq \left[\beta + \frac{1}{\mu}\right] \left[\ex{\nta - R} + \frac{R}{\mb}\right] + \frac{2}{\mu}\ln\left(\frac{\ex{\nta - R}}{\mb}\right)\label{eq:waitbusy2}\\
    &\;\;\;+  g\left(\ex{\left[\nta - R\right]^2} + 2\ex{\nta - R}, \ex{\nta - R}, k(1-\rho)\right).\nonumber
\end{align}
\end{restatable}
\paragraph{Intuition.} We defer the proof of Claim~\ref{clm:waitbusy2} until Section~\ref{sec:waitbusy}. For now, we give some brief intuition for how the bound is derived and how we use it in our proof. Essentially, we can consider performing the following procedure at time $\tau$: First, watch the system for $\setuptime$ time. If the number of jobs ever dips below $R + h$ during this watching period, we can end our integral immediately. If the number of jobs $N(t)$ never dips below $R+h$ during this watching period, then we know for sure that we have at least $\mintwo{R+h}{k}$ servers on at time $\tau + \rem{\mintwo{R+h}{k}}(\tau)$, since we have continually had at least $R+h$ servers either busy or setting up during that period.
Moreover, since we only turn off servers when there isn't work for them to do, those servers will \emph{stay} on until the number of jobs $N(t)$ dips below $R + h;$ in other words, they will stay on until time $\dgen$. The proof of the claim follows along essentially the same lines, formalizing things and performing computations using coupling and martingales.

\subsubsection{\texorpdfstring{Proof of Claim~\ref{clm:UBintTBinit}: Bound on Integral Until First Visit.}{Upper Bound on Integral until First Downward Visit.}}
We return to proving our claims. The proof of this claim is simple; it is essentially rolled into the proof of Claim~\ref{clm:waitbusy2}. From here, it suffices to apply the following claim, substituting in and simplifying constants:
\begin{restatable}[Upper Bound on $\ex{N(T_A)}$]{claim}{taexpect}\label{clm:taexpect}
    Recall that $T_A \triangleq \minpar{t>0: Z(t) = R+1}$. Then,
\begin{align*}
    &\ex{ N(T_A) - R} \leq F_1 \mu \beta \sqrt{R}\left(1 + \frac{F_2}{\sqrt{\mu\beta}}\right) \leq 2.9 \mu \beta \sqrt{R}
\end{align*}
and
\begin{equation*}
    \ex{\left(N(T_A) - R\right)^2} \leq F_1^2 (\mu \beta)^2 R \left(1 + \frac{F_2}{\sqrt{\mu\beta}}\right)^2 + 2\mu \beta R \leq 8.4(\mu\beta)^2 R + 2 \mu \beta R
\end{equation*}
where $F_1 = 2.12$ and $F_2 = 3.645$.
\end{restatable}

\subsubsection{\torp{Proof of Claim~\ref{clm:UBintTBbtwn}: Bound on Integral Between Visits.}{Bound on Integral Between Visits.}}
To prove Claim~\ref{clm:UBintTBbtwn}, we break the integral into two parts: from the down-crossing $\dnb{i}$ to the up-crossing $\unb{i}$, and vice-versa.

\paragraph{First part: from $\dnb{i}$ to $\unb{i}$.} 
To bound the integral from the down-crossing to the next up-crossing, we first make the simple observation that
\begin{equation}
    \int_{\dnb{i}}^{\mintwo{\unb{i}}{X}}{\ntr \dd t} 
    \leq \left[\mintwo{\unb{i}}{X} - \dnb{i} \right]\cdot \mb,
\end{equation}
since $\unb{i}$ is the next time $N(t) \geq R + \mb$.
To bound $\ex{\mintwo{\unb{i}}{X} - \dnb{i}\middle | \filt{\dnb{i}}}$, we couple the system to an infinite-server $M$/$M$/$\infty$ queue at time $\dnb{i}$, and note that the coupled up-crossing time $$\Tilde{T}_{(R+\mb-1)\to (R+\mb)} + \dnb{i} \geq \unb{i} \geq \mintwo{\unb{i}}{X}.$$ Since $\mb \leq \sqrt{R},$ from standard results on the $M$/$M$/$\infty$ (reproduced in Section~\ref{sec:mminf}), , we have that
\begin{equation*}
    \mb\ex{\mintwo{\unb{i}}{X} - \dnb{i}\middle | \filt{\dnb{i}}} \leq \mb\upbnda = \frac{1}{\mu\mb}b_2\frac{\mb^2}{\sqrt{R}} \leq \frac{1}{\mu \mb}b_2\sqrt{R}.
\end{equation*}

\paragraph{Second part: from $\unb{i}$ to $\dnb{i+1}$.} From $\unb{i}$ onwards, we use the ``wait-busy'' bound. Applying the ``Wait Busy'' Claim (Claim~\ref{clm:waitbusy2}) with $h = \mb$, we obtain
\begin{align*}
    \ex{\int_{\unb{i}}^{\dnb{i+1}}{\ntr \dd t} \middle | \filt{\unb{i}}} &\leq \beta + g\left(1 + 2\mu R \ex{\mintwo{\beta}{\dnb{i+1} -\unb{i} \middle | \filt{\unb{i}}}} , 1, \mb \right)\\
    &\;\;\; + \mb \ex{\mintwo{\beta}{\dnb{i+1} -\unb{i} \middle | \filt{\unb{i}}}}.
\end{align*}
To bound the conditional expectation of $\ex{\mintwo{\beta}{\dnb{i+1} - \unb{i}}\middle | \filt{\unb{i}}}$, we make our usual coupling argument.
Define a coupled system M/M/1 queue with departure rate $\mu R$, and let $\dgencouple$ be the length of its busy period.
It suffices to bound $\ex{\mintwo{\dgencouple}{\beta}}$, for the coupled relative down-crossing time $\dgencouple$.  From standard results on simple random walks (Claim~\ref{clm:contint}), we have $\ex{\mintwo{\dgencouple}{\beta}}\leq \downbnd$, giving
\begin{align*}
    \ex{\int_{\unb{i}}^{\dnb{i+1}}{\ntr \dd t} \middle | \filt{\unb{i}}} &\leq \frac{1}{\mu \mb} \left[ \frac{R}{\mb} + 2\right]  + \beta + \left[\downbnd\right]\left(\mb + \frac{R}{\mb}\right)\\
    &= \frac{1}{\mu \mb} \left[ \frac{R}{\mb} + 2 + \mu \beta \mb  + \left(\mb^2 + R\right)\left[\frac{ b_1 \sqrt{\mu \beta}}{\sqrt{R}} + \frac{6}{R}\right]\right] \\
    &\leq \frac{1}{\mu \mb} \left[ \max\left(\sqrt{R}, \frac{\rho}{1-\rho}\right) + 14 + \mu \beta \sqrt{R}  + 2 b_1\sqrt{\mu \beta R}\right],
\end{align*}
where in the last line we used $\mb \leq \sqrt{R}$; combining these two parts gives Claim~\ref{clm:UBintTBbtwn}. \hfill \Halmos

\subsubsection{\torp{Proof of Claim~\ref{clm:UBintTBprob}, Upper Bound on the Probability of Another Visit.}{Bound on the Probability of Another Visit.}}
To see \eqref{eq:succ_TB_prob}, we first note that, if there is another upcrossing, then there must be another downcrossing. As such, it suffices to upper bound $\pr{\unb{i} < X \middle | \dnb{i}}$. To do this, we note that the number of busy servers $Z(t) \geq R$. From Claim~\ref{clm:coupling}, it thus suffices to bound the corresponding probability in the coupled system with exactly $R$ busy servers. But this is simply the probability that a simple random walk started at $W(0) = \mb - 1$ hits $W(t) = \mb$ before it hits $W(t) = 0$. Classically, this probability is $\frac{1}{\mb}$; this proves the claim, and thus Lemma~\ref{lem:UBintTB}. \hfill \Halmos

%% file: ThesisFiles/UB/lemma_LB_X.tex
\newcommand{\longep}{\tau_L + \setuptime}
\newcommand{\xcouple}{\Tilde{X}}
\newcommand{\tlba}{\tau_L + \setuptime}
\subsection{Proof of Lemma~\ref{lem:LBX}: Lower Bound on the Cycle Length}\label{sec:LBX}
\paragraph{Preliminaries.}

The proof of this lemma is much simpler than the others. Before describing our strategy, we first state some preliminaries. Recall the definition of the start of the $j$-th epoch $\tau_j \triangleq \minpar{t\geq 0: N(t) \leq R - j},$ that we call the period $\bigg[\tau_j, \mintwo{\tau_{j+1}}{T_A} \bigg )$ the $j$-th epoch, and the we say epoch $j$ occurs if $\tau_j < T_A$. Now, say epoch $j$ is \textit{long} if it lasts longer than a setup time $\setuptime$; note that such an epoch must exist, since servers can only turn on during long epochs, and a server must turn on before the accumulation phase ends at time $T_A$. Let $L \triangleq \minpar{ j \in  \left \{ 0,1,2,\dots, R\right\}: \mintwo{\tau_{j+1}}{T_A} - \tau_j > \beta}$ be the index of the \emph{first} long epoch.
Note that, although the random time $\tau_L$ is \emph{not} a stopping time (we do not know how long an epoch will last when the epoch starts), the first moment we can identify epoch $L$, the random time $\tau_L + \setuptime$, \emph{is} a stopping time. Moreover, we know that $\tau_L +\setuptime < T_A$. From here, one sees that $\ex{X} = \ex{\longep } + \ex{X - \left(\longep\right)} \geq \setuptime + \ex{X - \left(\longep\right)}.$
To complete the proof, it suffices to show
\begin{equation}\label{eq:LBX1}
    \ex{X - \left(\tlba\right)} \geq \frac{L_1 \mu \beta \sqrt{R}}{\mu k (1-\rho)}.
\end{equation}

\subsubsection{\texorpdfstring{Proof of \eqref{eq:LBX1}: Lower Bound on the Remaining Cycle Length.}{Proof of Eq.~\ref{eq:LBX1}: Lower Bound on the Remaining Cycle Length.}}
To show \eqref{eq:LBX1}, we first show we can bound an analogous quantity in a coupled process, then appeal to standard results on the M/M/1 queue.

\paragraph{Defining the coupled process $\ncouple(t)$.} 
Note that the number of busy servers $Z(t) \leq k$.
It follows from Claim~\ref{clm:coupling} that, for any time $t \geq \longep$, the coupled process below satisfies $\ncouple(t) \leq N(t)$:
$$    \ncouple(t) \triangleq N \left(\longep\right) + A\left( (\longep, t]\right) - \depOp{k}{(\longep, t]}$$

\paragraph{Using the coupled process to bound $\ex{X-\longep}$.}
We now use this process to bound $\ex{X - \longep}$.
Recall that the end of the renewal cycle $X\triangleq \minpar{t>0: Z(t^-) = R+1, Z(t) = R}$ occurs when the $(R+1)$-th server turns off.
It is useful to view $X$ in a different way: since the accumulation time $T_A$ is the moment when the $(R+1)$-th server turns on, we also know that the end of the renewal cycle $X=\minpar{t>T_A: N(t) \leq R}$ is the first moment after time $T_A$ that the number of jobs $N(t) \leq R$.
Furthermore, since the time $\longep$ happens \textit{before} any server could possibly turn on, the time $\longep < T_A$. 
Denoting the end of the coupled renewal as $\xcouple$ as the first moment the coupled process $\ncouple(t) \leq R$, we have 
$$ \xcouple \triangleq \minpar{t>\tlba: \ncouple(t) \leq R} \leq \minpar{t > \tlba: N(t) \leq R} \leq \minpar{t > T_A: N(t) \leq R} = X.$$
\paragraph{Bounding the end of coupled renewal $\ex{X-\longep}$.}
To bound the quantity $\ex{X-\longep}$, we condition on the filtration at time $\longep$ and use standard results on the M/M/1 busy period.
Note that, since the departure rate of the coupled system is fixed at $\mu k$, the period $\left[\xcouple - \tlba\right]$ is precisely the length of an M/M/1 busy period with 1) arrival rate $k\lambda$, 2) departure rate $k\mu$, and 3) started by $\left[N\left(\tlba\right) - R\right]^+$ jobs.
It follows that $\ex{X-\left(\tlba\right)\middle | \filt{\tlba}} \geq \frac{\left[N\left(\tlba\right)-R\right]^+}{\mu k(1-\rho)}$.
Taking expectations, applying Jensen's and results from the lower bound (\eqref{eq:ntlba} and \eqref{eq:lbound}), we obtain, proving \eqref{eq:LBX1}, Lemma~\ref{lem:LBX}, and Theorem~\ref{thm:main} simultaneously,
\begin{equation*}
    \ex{X-\left(\tlba\right)} \geq \frac{\left[\ex{N\left(\tlba\right) - R}\right]^+}{\mu k (1-\rho)} \geq \frac{\mu \beta \ex{L}}{\mu k (1-\rho)} \geq \frac{L_1 \mu \beta \sqrt{R}}{\mu k(1-\rho)}. \tag*{\Halmos}
\end{equation*}

%% file: claims/tightness.tex
We now show that the upper and lower bounds of Theorems~\ref{thm:main} and~\ref{thm:improvement}, respectively, differ by at most a multiplicative factor. 

\begin{theorem}\label{thm:tightness}
    The bounds of Theorems~\ref{thm:main} and~\ref{thm:improvement} lie within a constant multiplicative factor of each other. In particular, using $=_c$ to denote equivalence modulo a multiplicative constant,
    \begin{equation}\label{eq:constAim}
    \ex{Q(\infty)} =_c  \beta \sqrt{R} + \frac{1}{1-\rho}.
    \end{equation}
\end{theorem}

\subsection{Proof for Lower Bound}
We prove Theorem~\ref{thm:tightness} in two parts, showing equivalence for the lower bound, then for the upper bound. For the lower bound, we first discard all constants and a number of terms in the denominator, since $\beta > \frac{1}{\mu}$ by assumption. Doing so, we obtain
\begin{equation}\label{eq:lbstripped}
    \ex{Q(\infty)} \geq_c 
    \frac{\mu \beta^2 \sqrt{R} + \frac{\left[L_1\beta\sqrt{R} - k(1-\rho)\right]^+}{\mu k (1-\rho)}\left[\left[L_1\beta\sqrt{R} - k(1-\rho\right]^+ + \frac{1}{1-\rho}  \right]}
    {\beta +  \frac{\beta \sqrt{R}}{\mu k (1 - \rho)}},
\end{equation}
where $\geq_c$ denotes that the inequality holds up to an (unspecified) constant factor.

\subsubsection{\torp{Replacing the $\left[\cdot\right]^+$ Term.}{Replacing the ``positive part'' term.}}
Now, we show that the $\left[\beta\sqrt{R} - k(1-\rho)\right]^+$ term can be replaced by the term $\mu \beta\sqrt{R}$, while losing only a constant factor; this turns out to be the difficult part.
We approach this by casing on whether the positive term $\frac{1}{2}L_1 \mu \beta \sqrt{R} \geq k(1-\rho)$.

\paragraph{First case.} 
If we have $\frac{1}{2}L_1\beta \sqrt{R} \geq k(1-\rho)$, then $\left[\mu L_1\beta\sqrt{R} - k(1-\rho)\right]^+ \geq \frac{1}{2}L_1 \beta \sqrt{R} =_c \mu\beta \sqrt{R}$.

\paragraph{Second case.} In the second case, assume that $\frac{1}{2}L_1\beta \sqrt{R} < k(1-\rho)$. 
In this case, even if we increase the value of the numerator by replacing the $[\cdot]^+$ term, the relevant term in the numerator becomes
\begin{align*}
     \frac{\beta\sqrt{R}}{ k (1-\rho)}\left[\mu \beta\sqrt{R} + \frac{1}{1-\rho}  \right] &= \mu \beta^2 \sqrt{R}\cdot \frac{\sqrt{R}}{k(1-\rho)} \cdot \left[1  + \frac{1}{\mu \beta \sqrt{R}} \right] =_c \mu \beta^2 \sqrt{R}\cdot \frac{\sqrt{R}}{k(1-\rho)} \leq_c \mu \beta^2 \sqrt{R}\cdot \frac{1}{\beta},
\end{align*}
where in the second equality we have used that $\sqrt{R}\geq 1$ and $\mu\beta \geq 1$, and in the final inequality we have used our case assumption.
From here, it's clear that one can replace the term $[L_1 \mu \beta \sqrt{R} - k(1-\rho)]^+$ with the term $\mu \beta \sqrt{R}$ without altering the scaling behavior of numerator. In other words,
\begin{equation}\label{eq:tightRes}
    \ex{Q(\infty)}  \geq_c  \frac{\mu \beta^2 \sqrt{R} + \frac{\mu \beta\sqrt{R}}{\mu k (1-\rho)}\left[\beta\sqrt{R} + \frac{1}{1-\rho}  \right]}
    {\beta +  \frac{\beta \sqrt{R}}{\mu k (1 - \rho)}}
    =_c \mu\beta\sqrt{R} + \frac{\sqrt{R}}{k(1-\rho) + \sqrt{R}}\frac{1}{1-\rho}.
\end{equation}

\subsubsection{Bounding the Final Term.} We now show equivalence for this final term, i.e. that
\begin{equation}\label{eq:tightEquiv}
     \mu\beta\sqrt{R} + \frac{\sqrt{R}}{k(1-\rho) + \sqrt{R}}\frac{1}{1-\rho} =_c \mu\beta \sqrt{R} + \frac{1}{1-\rho}.
\end{equation}
To do so, we bound the rightmost term in \eqref{eq:tightRes}. 
Note that, since $\frac{\sqrt{R}}{k - R + \sqrt{R}}  \leq \frac{R}{k - R + R} = \rho$, in order for this term to have an appreciable effect on the scaling, we must have that $\frac{\rho}{1-\rho} \geq_c
\mu \beta \sqrt{R}$, or, phrased more usefully, we must have $ \sqrt{R} \geq_c \mu \beta k(1-\rho)$.
But even in this case, we can bound the factor in the rightmost term of \eqref{eq:tightRes} with 
$\frac{\sqrt{R}}{k(1-\rho)  + \sqrt{R}} \geq_c \frac{\beta k(1-\rho)}{k(1-\rho) + \mu \beta k(1-\rho)} = \frac{\mu \beta}{1 + \mu \beta} =_c 1$; the multiplicative equivalence \eqref{eq:tightEquiv} follows.

\subsection{Proof for Upper Bound}

The proof for the upper bound follows along the same lines.
First, note that the terms outside of the fraction are $\frac{R}{M} \leq \frac{R}{k(1-\rho)} = \frac{\rho}{1-\rho}$ and $\sqrt{\mu \beta R} \ll \mu \beta \sqrt{R}$. 
Discarding the lower order terms and constants, we obtain
\begin{equation}
    \ex{Q(\infty)} \leq_c \mu \beta \sqrt{R} + \frac{\rho}{1-\rho} + \frac{\mu \beta^2 \sqrt{R} + g\left(9 (\mu  \beta)^2 R, 3 \mu \beta \sqrt{R}, k(1-\rho)\right)}{\beta + \frac{\mu \beta \sqrt{R}}{\mu k(1-\rho)}},
\end{equation}
where one should recall that $M =_c \min \left( k(1-\rho), \sqrt{\mu \beta R}\right)$.
For the terms in the fraction, the denominator of the upper bound is already up-to-constants-equivalent to the denominator of \eqref{eq:lbstripped}. It thus suffices to show that the numerator of the upper bound aligns with the numerator of \eqref{eq:tightRes}. 
However, note that, by definition, the function $$g\left(9(\mu\beta)^2 R, 3(\mu\beta)\sqrt{R}, k(1-\rho)\right) \triangleq \frac{3 \mu \beta \sqrt{R}}{k(1-\rho)}\left[ \frac{3 \mu \beta \sqrt{R} + 1}{2} + \frac{1}{1-\rho}\right];$$ thus, the terms are clearly equivalent up to scaling. \hfill \Halmos

%% file: 7_conclusion.tex
\paragraph{Summary.} In this work, we provided a comprehensive characterization of the average waiting time within the M/M/k/Setup-Deterministic.
In particular, we proved the first multiplicatively-tight upper and lower bounds, respectively, on the average waiting time.
By combining our upper and lower bounds, we also obtained an approximation which we demonstrate to be highly accurate.
These results significantly advance our understanding of how setup times affect queueing systems beyond the prevalent Exponential setup model, providing new insights into capacity provisioning in systems with setup times.

Our work leaves open many possible directions for future research, discussed below.
\paragraph{Tail of Response Time.} 
One such possibility is studying other aspects of system performance in the M/M/k/Setup-Deterministic.
For context, in this work, we only analyze the average waiting time of a job in the system.
However, in many applications, a customer's quality of service is determined by more advanced metrics like the \emph{tail} of their waiting time, e.g. the fraction of jobs which are served within 100 ms of their arrival.
While our analysis lends itself most easily to studying the average waiting time, it’s conceivable that, with the right selection of stopping time and martingales, our results could be extended to higher moments and possibly even to a direct analysis of the tail.

\chmade{
\paragraph{Other natural policies.}
Another possible expansion on this work is studying other natural policies.
For example, consider the following variation on the policy studied in this paper: instead of turning a server off instantly when there is no work for that server to do, we allow the server to sit idle for, say, $5$ seconds before turning off.
One would imagine that such a policy would have a lower average waiting time and a higher power consumption than this paper’s setup policy.
Moreover, one would anticipate that the tradeoff between the waiting time and the power consumption could be tuned by changing the amount of allowed idle time $\gamma$; at $\gamma = 0$, we recover the policy studied in this paper, and as $\gamma \to \infty$, we recover the ``Always On’’ policy.
This class of policies, often referred to as the DelayedOff class of policies, have been studied extensively in the Exponential setup model due to their tunability, ease of implementation, and excellent empirical performance \cite{gandhi2012autoscale,pender2016law,gandhiguptaharcholkuzuch10, gandhiharchol11Allterton}; it would be a great contribution to extend our knowledge of these policies to settings with Deterministic setup times.
\paragraph{Tradeoff between waiting time and power.}
Another important avenue of study lies in characterizing the fundamental tradeoff between average waiting time and average power consumption.
One point on the waiting-v.s.-power Pareto curve is clear: to obtain the minimum possible waiting time, one should keep all $k$ servers on all the time.
However, all the other points on the optimal tradeoff curve are unknown; furthermore, even if we knew a given (waiting time, power consumption) pair was feasible, it would be entirely unclear how we should go about achieving it.
For example, it’s conceivable that the DelayedOff class of policies mentioned previously could perform near-optimally as one varies the ``idle time’’ parameter $\gamma$.
A more principled understanding of the optimal tradeoff between waiting time and power consumption could allow us to directly derive a policy which achieves this optimal tradeoff, and could grant us greater insight into how to design more performant setup systems in the practical setting.
\paragraph{More general setup distributions.}
Lastly, we note that our results might be generalized to setup distributions beyond Deterministic.
A promising step would be the establishment of closed form bounds for the \textit{Exponential} setup time model; note that no closed-form result exists for the Exponential model, excluding the infinite server result of \cite{gandhi2010server}.
With the establishment of an Exponential result, it seems likely that our arguments could be used to obtain a result for a Hyperexponential setup system, or even bootstrapped into a full phase-type result.
Another way to generalize our setup model is to study a system where setup times take on two possible values (a short setup or a long setup); whether a setup is long or short (and how the setup time of a server changes over time) could be used to model various real-world scenarios, e.g. where variable data loading times or intermittent network congestion.
By generalizing our setup model in these ways, we could greatly expand the practical utility of our results and gain even deeper insight into what fundamentally governs the performance of systems with setup times.
}

%% file: 8_appendix.tex
\begin{APPENDICES}



\newpage
\removeforarxiv{
\begin{center}
{\Large \textbf{Electronic Companion}}
\end{center}
}
\section{Proof of the Improved Lower Bound, Theorem~\ref{thm:improvement}}\label{sec:improvement}
\input{ThesisFiles/LB_improv/LB_imp}





\input{8_claims}

\section{\texorpdfstring{Analysis of $m$-Policies}{Analysis of m-Policies}}\label{sec:extension}

\input{9_extension}

\section{\texorpdfstring{Miscellaneous Claims: Analysis of $m$-Policies}{Miscellaneous Claims: Analysis of m-Policies}}
\subsection{\texorpdfstring{Proof of \eqref{eq:ext_nta}: Upper Bound on $\ex{N(T_A) - R}$ for $m$-Policies}{Bound on E[N(TA) - R] for m-Policies}}\label{sec:ext_nta}
\input{ExtensionFiles/ext_nta}
\subsection{\torpdf{Bounds on hitting probabilities for an M/M/$\infty$ queue.}{Bounds on hitting probabilities for an M/M/inf queue.}}\label{sec:ext_prob_bnds}
\input{ExtensionFiles/ext_probabilities}

\subsection{\texorpdfstring{Proof of Claim~\ref{clm:ext_pe_clm}: Bound on $p_e$, an M/M/$R$ hitting probability.}{Bound on p\_e, an M/M/R hitting probability.}}\label{sec:ext_pe_clm}
\input{ExtensionFiles/ext_pe_clm}

\subsection{Proof of Claim~\ref{clm:ext_pth_clm}.}\label{sec:ext_pth_clm}
\input{ExtensionFiles/ext_pth_clm}

\end{APPENDICES}

%% file: ThesisFiles/LB_improv/LB_imp.tex
\newcommand{\erbterm}{e^{\frac{1}{24 R \setuptime}}}
\newcommand{\qcouple}{\Tilde{Q}}
\newcommand{\mtb}{\mintwo{\tau_i + \beta}{\tau_{i+1}}}
\newcommand{\upc}{\gamma}
\newcommand{\dnc}{\alpha}
\newcommand{\ncac}{\ncouple}
\newcommand{\cind}{n_\alpha}
\newcommand{\tac}{\Tilde{T}_A}


\subsection{The New Lower Bound: Proof Outline.}
\paragraph{Basic Structure.} We prove Theorem~\ref{thm:improvement} via the MIST method. 
As noted in Section~\ref{sec:keyideas}, we begin by applying the Renewal-Reward theorem to the queue length $Q(t)$, defining our renewal points as those points in time where the $(R+1)$-th server turns off. Defining time $0$ to be one of these points, and defining the cycle time $X\triangleq \minpar{t>0: Z(t^-) = R+1, Z(t) = R}$ as the next point, this gives
\begin{equation*}
    \ex{Q(\infty)} = 
    \frac
    {
        \ex{ \int_0^X{Q(t) \dd t}}
    }
    {
        \ex{X}
    }.
\end{equation*}
To obtain our lower bound, it suffices to lower bound the numerator and upper bound the denominator of this fraction, i.e. lower bound  $\ex{ \int_0^X{Q(t) \dd t}}$ and upper bound $\ex{X}$. 
The time integral lower bound is handled by Lemma~\ref{lem:LBint}, which we state at the end of this section. 
The cycle length upper bound is split into two separate lemmas: Lemma~\ref{lem:UBTA} upper bounds the length of the cycle's ``first part'' and Lemma~\ref{lem:UBTB} bounds the length of its ``second part''.

\paragraph{Decomposition into phases.} However, before we state or prove these lemmas, we first discuss the decomposition of the renewal cycle $[0,X)$ into two parts; one might think of this as a ``miniature'' application of the MIST method.
We begin by noting that the end of the renewal cycle is moment when the $(R+1)$-th server turns off. 
Since the $(R+1)$-th server is off at the start of a renewal period, we can break the renewal cycle into two phases based on whether the $(R+1)$-th server has turned on yet. 
Formally, we define the accumulation time $T_A \triangleq \minpar{t>0: Z(t) = R+1}$ as the first moment that the $(R+1)$-th server turns on.
From here, we can focus separately on the accumulation phase, from time $0$ to time $T_A$, and the draining phase, from time $T_A$ to time $X$. 

With this decomposition, we can now state our main lemmas. Their proofs follow in sequence afterwards.

\begin{restatable}[Lower bound on Cycle Integral]{lemma}{lbint}\label{lem:LBint}
    Define busy period integral $\ibp{x}{z}$ as 
    $$\ibp{x}{z} \triangleq \frac{x}{\mu z} \left[ \frac{x+1}{2} + \frac{1}{1 - \frac{k\lambda}{k\lambda + \mu z}}\right] = \frac{x}{\mu z} \left[ \frac{x+1}{2} + \frac{R}{z}\right].$$ For the time integral of the queue length $Q(t)$ over a renewal cycle, we have
    \begin{equation*}
        \ex{\int_0^{X}{Q(t) \dd t}} \geq \frac{1}{2} \mu \beta^2 L_1\sqrt{R} + \ibp{\left[\mu \beta L_1 \sqrt{R} - (k-R) \right]^+}{k - R}.
    \end{equation*}
\end{restatable}

\begin{restatable}[Upper bound on Accumulation Phase Length]{lemma}{ubacc_cycle}\label{lem:UBTA}
    Recall that $$T_A \triangleq \minpar{t > 0: Z(t) \geq R + 1}$$ is the amount of time until the $(R+1)$-th server turns on. Then we can bound the expectation $\ex{T_A}$ by
    \begin{equation*}
        \ex{T_A} \leq \erbterm \left( \sqrt{1 + \frac{1}{2 R \setuptime} } \right) 
         \left[ 
            1 
            + \frac{e^{\frac{1}{12 R}}}{\sqrt{2 \mu \beta}}       
        \right] \beta \leq 1.08 * \beta.
    \end{equation*}
\end{restatable}

\begin{restatable}[Upper bound on Draining Phase Length]{lemma}{ubdr_cycle}\label{lem:UBTB}
     Recall that the accumlation time $$T_A \triangleq \minpar{t > 0: Z(t) \geq R + 1}$$ is the amount of time until the $(R+1)$-th server turns on and the cycle time $X$ is the moment when it turns off. Then, one can bound $\ex{X-T_A}$ by
    \begin{equation*}
        \ex{X-T_A} \leq \beta +  \frac{1}{\mu}\frac{F_1 \beta \sqrt{R}}{k - R} + \frac{1}{\mu}\frac{3}{2}\ln(\beta) + \frac{1}{\mu}\ln(F_1 D_1)
        +
        \frac{2}{\mu} 
        + \left[D_2 + \frac{D_3}{\sqrt{R}} \right]\max\left(\frac{1}{D_1 \sqrt{\mu \beta}}, \frac{1}{\sqrt{R}}\right),
    \end{equation*}
    where $F_1$, $D_1$, $D_2$, and $D_3$ are constants not depending on system parameters. 
\end{restatable}

\subsection{Proof of Lemma~\ref{lem:LBint}: Lower Bound on Cycle Integral.}
\input{ThesisFiles/LB_improv/lemma_LB_int}

\subsection{\torp{Proof of Lemma~\ref{lem:UBTA}: Upper Bound on the Accumulation Time $\ex{T_A}$.}{Upper Bound on the Accumulation Time E[TA]}}\label{sec:UBTA}

\input{ThesisFiles/LB_improv/lemma_UB_TA}

\subsection{\torp{Proof of Lemma~\ref{lem:UBTB}: Upper Bound on the Remaining Cycle Time $\ex{X-T_A}$.}{Upper Bound on the Cycle Time E[X - TA]}}\label{sec:UBTB}

\input{ThesisFiles/LB_improv/lemma_UB_TB}

%% file: ThesisFiles/LB_improv/lemma_LB_int.tex
\newcommand{\epochdone}{\gamma}
\newcommand{\torpdf}[2]{\texorpdfstring{#1}{#2}}

\subsubsection{Lemma~\ref{lem:LBint} Proof Outline}

\paragraph{Basic Strategy.} First, we split the first phase $[0,T_A)$ into epochs, where epoch $i$ begins when the number of busy servers $Z(t)$ first drops to $R-i$, and an epoch ends either when the next epoch starts or when the first phase ends. Our goal will be to analyze a specific ``significant'' epoch. In particular, we say that an epoch is {\em long} if it lasts for longer than a setup time $\beta$. Because the accumulation phase ends when the $(R+1)$-th server turns on, at least one epoch must be long. We use $L$ to denote the index of the \emph{first} long epoch. From here, we argue via a martingale/coupling argument that the expected time integral over the first $\beta$ time in epoch $L$ is at least $\frac{1}{2} \beta^2 \ex{L}$. To bound the integral afterwards, we couple the behavior of the total number of jobs $N(t)$ to the queue length in an M/M/1 queue with arrival rate $k\lambda$ and departure rate $k \mu$.

\paragraph{Formalization.} Define the stopping time $\tau_i \triangleq \minpar{t\geq 0: N(t) \leq R - u}$ as the beginning of epoch $i$. 
We say that the epoch \textit{occurs} if $\tau_i < T_A$, and define the end of epoch $i$ as $\epochdone_i \triangleq \mintwo{\tau_{i+1}}{T_A}$ the moment when either epoch $i+1$ begins or the first phase ends. 
If epoch $i$ occurs, we say it is long if $\epochdone_i - \tau_i \geq \beta$. Let $L \triangleq \minpar{i \in \mathbb{N}: \epochdone_i - \tau_i \geq \beta}$ be the index of the first long epoch. 
We now proceed to the proofs.

\subsubsection{\torpdf{Lower Bound on Integral until $\tau_L + \beta$.}{Lower Bound on Integral until tau\_L + beta}}
\begin{claim}\label{clm:LBintINIT}
    Let $L$ be the index of the first long epoch.
    Then,
    \begin{equation}
        \ex{\int_0^{\tlba}{Q(t)\dd t}} \geq \frac{1}{2} \mu \beta^2 L_1 \sqrt{R},
    \end{equation}
    where $L_1$ is some absolute constant.
\end{claim}

\paragraph{Claim~\ref{clm:LBintINIT} Proof Strategy.} First, we show that the initial integral is bounded by
\begin{equation}\label{eq:int_taul}
    \ex{\int_0^{\tau_L + \setuptime}{Q(t) \dd t}} \geq \frac{1}{2} \mu \beta^2  \ex{L} .
\end{equation}
Afterwards, we give a bound on $\ex{L}$, showing that
\begin{equation}\label{eq:lbound}
    \ex{L} \geq L_1 \sqrt{R}.
\end{equation}

\paragraph{Proof of \eqref{eq:int_taul}, Bound in terms of $\ex{L}$.}
To show \eqref{eq:int_taul}, we first condition on whether $L \geq i$, giving
\begin{align*}
    \ex{\int_0^{\tlba}{Q(t) \dd t}} &= \sum_{i=0}^{\infty}{\ex{\int_{\tau_i}^{\mintwo{\tau_i + \beta}{\tau_{i+1}}}
    {Q(t) \dd t}\indc{L \geq i}}}\\
    &=  \sum_{i=0}^{\infty}{\ex{\int_{\tau_i}^{\mintwo{\tau_i + \beta}{\tau_{i+1}}}
    {Q(t) \dd t}\middle | \filt{\tau_i}}\pr{L \geq i}}.
\end{align*}

To further develop this conditional expectation, we note that during the interval $[\tau_i, \mintwo{\tau_i + \beta}{\tau_{i+1}})$, the system must have exactly $Z(t) = R -i$ busy servers running, meaning that $Q(t) = N(t) - (R - i)$. 
Defining a coupled process $\qcouple(t)$ as 
$
    \qcouple(t) =  A(\tau_i, t) - \depOp{R-i}{\tau_i, t},
$
we see that $Q(t)$ and $\qcouple(t)$ coincide during the interval in question. Moreover, one can redefine the stopping time $\gamma = \tau_{i+1}$ as $\minpar{t>\tau_i: \qcouple(t) = -1}$. Noting that $Q\left( \mintwo{\gamma}{t}\right) = -1$ for any time $t > \gamma$, we find that
\begin{align*}
    \int_{\tau_i}^{\mtb}{Q(t) \dd t} &= \int_{\tau_i}^{\mtb}{\qcouple(t) \dd t}\\
    &= \int_{\tau_i}^{\mtb}{\qcouple\left(\mintwo{t}{\tau_{i+1}}\right) \dd t} +  \int_{\mtb}^{\tau_i + \beta}{\left(\qcouple\left(\mintwo{t}{\tau_{i+1}}\right) + 1\right) \dd t} \\
    &=\int_{\tau_i}^{\tau_i + \beta}{\qcouple\left(\mintwo{t}{\tau_{i+1}}\right) \dd t} + \left[ \beta - \mintwo{\beta}{\tau_{i+1} - \tau_i}\right]\\
    &\geq \int_{\tau_i}^{\tau_i + \beta}{\qcouple\left(\mintwo{t}{\tau_{i+1}}\right) \dd t}.
\end{align*}
Taking the conditional expectation at time $\tau_i$, we find
\begin{align*}
    \ex{\int_{\tau_i}^{\tau_i + \beta}{\qcouple\left(\mintwo{t}{\tau_{i+1}}\right) \dd t}\middle | \filt{\tau_i}} = \int_{\tau_i}^{\tau_i + \beta}{\ex{\qcouple\left(\mintwo{t}{\tau_{i+1}}\right)\middle | \filt{\tau_i}} \dd t}.
\end{align*}
Noting that $V_L(t) = \qcouple(t) - \mu i \left[t - \tau_i\right]$ is a martingale, and that $\mintwo{t}{\tau_{i+1}}$ is an almost-surely bounded stopping time, we have that
\begin{align*}
    \qcouple(\tau_i) = V_L(\tau_i) = 0
                    = \ex{V_L\left(\mintwo{t}{\tau_{i+1}}\right)\middle| \tau_i}
                    = \ex{\qcouple\left(\mintwo{t}{\tau_{i+1}}\right)\middle| \filt{\tau_i}} - \mu i \ex{\mintwo{t}{\tau_{i+1}}\middle| \filt{\tau_i}}.
\end{align*}
Since 
\begin{equation*}
    \ex{\mintwo{t}{\tau_{i+1}}\middle|\filt{\tau_i}} \geq t\cdot \pr{\tau_{i+1} - \tau_i \geq t} \geq t \cdot \pr{\tau_{i+1} - \tau_i \geq \beta} = t \pr{L = i \middle | L \geq i},
\end{equation*}
we have 
\begin{align*}
   \pr{L\geq i} \ex{\int_{\tau_i}^{\mtb}{Q(t) \dd t}\middle| L \geq i}  &\geq \pr{L \geq i} \ex{\int_{\tau_i}^{\tau_i + \beta}{\qcouple\left(\mintwo{t}{\tau_{i+1}}\right) \dd t}\middle | L \geq i}\\
    & \geq \int_{\tau_i}^{\tau_i + \beta}{\mu i t \pr{L = i}}\\
    &= \mu \frac{\beta^2}{2} i \pr{L=i}
\end{align*}
Summing across all $i$, we obtain \eqref{eq:int_taul}.

\paragraph{\torp{Proof Sketch for \eqref{eq:lbound}, bound on $\ex{L}$.}{Proof Sketch of bound on E[L].}}
We defer the full proof of \eqref{eq:lbound} to Section~\ref{sec:lbound}, and for now give a proof sketch. 
We prove \eqref{eq:lbound} by first showing that
\begin{equation*}
    \pr{ L > j \middle | L \geq j } \geq \left(1 - \frac{j}{R}\right) \left(1 - \frac{b_1}{\sqrt{\mu \beta R}}\right), 
\end{equation*}
where $b_1 = \frac{2}{\sqrt{\pi}}$.
Next, we show that this implies that, for any $\delta \in (0,1)$ and any $j < \delta R$,
\begin{equation*}
     \pr{L > j} \geq \left( 1 - \frac{b_1}{\sqrt{\mu \beta R}}\right)^{j+1} e^{-\frac{j(j+1)}{2 R} \frac{1}{1 - \delta}}.
\end{equation*}
From here, we use the sum of tails formula $\ex{L} = \sum_{j=0}^{\infty}{\pr{L>j}}$ to show
\begin{align*}\label{eq:lbex}
    \ex{L} &\geq \left(1 - \frac{b_1}{\shs}\right) \left(\left[\sqrt{\frac{\pi}{2}}(1-\delta) - \frac{1.15(1-\delta)}{\sqrt{\mu \beta}}\right] \sqrt{R}  - \frac{1}{2} - \frac{2(1-\delta)}{\delta} e^{- R \frac{\delta^2}{1-\delta}}\right).
\end{align*}
Choosing $\delta = \frac{2}{\sqrt{R}}$ then noting that $\mu \beta \geq 100$ and $R \geq 100$ gives the result. 
\hfill\Halmos

\subsubsection{\torp{Lower Bound on Integral after $\tau_L + \beta$.}{Lower Bound on Integral after tau\_L + beta.}}

To finish our lower bound on the integral, we now show the following claim.
\begin{claim}\label{clm:LBintFIN}
    Let $L$ be the index of the first long epoch.
    Then,
    \begin{equation}
        \ex{\int_{\tlba}^{X}{Q(t)\dd t}} \geq \ibp{\left[\mu \beta L_1 \sqrt{R} - \left( k - R\right)\right]^+}{k-R}
    \end{equation}
    where $L_1$ is some absolute constant.
\end{claim}

\paragraph{Claim~\ref{clm:LBintFIN}: Proof Strategy.} First, we show that the remaining integral is bounded by
\begin{equation}\label{eq:int_aft}
    \ex{\int_{\tau_L+ \setuptime}^{X}{Q(t) \dd t}} \geq \ibp{\left[ \ex{N(\tau_L + \beta)} - k\right]^+}{k-R}.
\end{equation}
Then, we use martingales again to show that
\begin{equation}\label{eq:ntlba}
    \ex{N\left(\tau_L + \beta\right)} \geq R + \mu \beta \ex{L}.
\end{equation}
Applying \eqref{eq:lbound}, our bound on $\ex{L}$, we obtain the result.

\paragraph{\torp{Proof of \eqref{eq:int_aft}, Bound in terms of $\ex{N(\cdot)}$.}{Bound in terms of N(TA).}}
To prove \eqref{eq:int_aft}, we make a simple coupling argument. Let $\eta_k \triangleq \minpar{t\geq \tau_L + \beta: N(t) \leq k}$. Since the draining phase starts at $T_A \geq \tau_L + \beta$ and the end of the cycle $X = \minpar{t \geq T_A: N(t) \leq R}$, we know that $X \geq \eta_k$. Moreover, we know the number of busy servers $Z(t) \leq k$; it follows by Claim~\ref{clm:coupling} that we can define $\ncouple(t)$ as 
\begin{equation*}
    \ncouple(t) \triangleq N\left(\tau_L + \beta\right) + A\left(\tau_L + \beta, t\right) - \depOp{k}{\left(\tau_L + \beta, t\right)}
\end{equation*}
and have $\ncouple(t) \leq N(t)$ for any $t > \tlba.$ Even further, we can defined a coupled hitting time $\Tilde{\eta}_k \triangleq \minpar{t > \tlba: \ncouple(t) \leq k}$ which must happen before $\eta_k$. In other words,
\begin{align*}
    \int_{\tlba}^{X}{Q(t) \dd t} &\geq \int_{\tlba}^{\eta_k}{Q(t) \dd t} 
    \geq \int_{\tlba}^{\eta_k}{\left[N(t) - k\right] \dd t}
    \geq \int_{\tlba}^{\Tilde{\eta}_k}{\left[\ncouple(t) - k\right] \dd t}.
\end{align*}
This final term is just the time integral of the number of jobs in system over a M/M/1 busy period started by $\left[N\left(\tlba\right) - k\right]^+$ jobs, where jobs arrive at rate $k\lambda$ and depart at rate $k\mu$. 
Accordingly, we have
\begin{align*}
     \ex{\int_{\tlba}^{X}{Q(t) \dd t}} 
     &\geq \ex{\ibp{\left[N\left(\tlba\right) - k \right]^+}{k - R}}\\
     &\geq \ibp{\ex{N\left(\tlba\right) - k }^+}{k-R}\\
     &\geq \ibp{\ex{\left[N\left(\tlba\right) - R - \left( k - R\right)\right]^+}}{k-R},
\end{align*}
where in the last two lines we have applied Jensen's inequality.

\paragraph{\torp{Proof of \eqref{eq:ntlba}, Bound on $\ex{N(\cdot)}$.}{Upper bound on E[N(--)]}}

To bound $\ex{N(\tlba)}$, we condition on the value of $L$, then make a martingale argument.
\begin{align*}
    \ex{N(\tlba)} &= \sum_{i=0}^{R}{\ex{N(\tau_i + \beta) \indc{L = i }}}\\
    &= \sum_{i=0}^{R}{\ex{N(\tau_i + \beta) \indc{L = i, L \geq i }}}\\
    &\geq \sum_{i=0}^{R}{\pr{L\geq i}\ex{N(\tau_i + \beta) \indc{L = i} \middle | L \geq i}}.
\end{align*}
Continuing with this conditional expectation,
\begin{align}
    &\;\ex{N(\tau_i + \beta) \indc{L = i} \middle | \filt{\tau_i}}  \nonumber \\
    &\;\;=  \ex{N(\tau_i + \beta) \indc{\tau_i + \beta < \tau_{i+1}} \middle | \filt{\tau_i}}  \nonumber \\
    &\;\;= \ex{\left[N(\tau_i + \beta) - (R-(i+1))\right]\indc{\tau_i + \beta < \tau_{i+1}} \middle | \filt{\tau_i}} + (R-i - 1)\pr{\tau_i + \beta < \tau_{i+1} \middle | \filt{\tau_i}} \nonumber \\
    &\;\;= \ex{N(\mintwo{\tau_i + \beta}{\tau_{i+1}}) - (R-(i+1)) \middle | \filt{\tau_i}} +  + (R-i - 1)\pr{\tau_i + \beta < \tau_{i+1} \middle | \filt{\tau_i}} \nonumber \\
    &\;\;= 1 +\mu i \ex{\mintwo{\beta}{\tau_{i+1} - \tau_i}}   + (R-i - 1)\pr{\tau_i + \beta < \tau_{i+1}} \nonumber \\
    &\;\;\geq 1 +\mu i \beta \pr{\tau_{i+1} - \tau_i \geq \beta}   + (R-i - 1)\pr{\tau_i + \beta < \tau_{i+1}} \nonumber \\
    &\;\;= 1 +\mu i \beta\pr{L = i \middle | L \geq i}   + (R-i - 1)\pr{L = i \middle | L \geq i}.\label{eq:resaba}
\end{align}

Summing across $i$, we find
\begin{align*}
        \ex{N(\tlba)} &= \sum_{i=0}^{R}{\pr{L\geq i}}
        = \left(1+ \ex{L}\right) + \left(\mu \beta \ex{L}\right) + \left(R - \ex{L} - 1\right)
        = R + \mu \beta \ex{L}.\tag*{\Halmos}
\end{align*}

Combining Claims~\ref{clm:LBintINIT} and~\ref{clm:LBintFIN}, we obtain a lower bound on $\ex{\int_0^{X}{Q(t) \dd t}}$, proving Lemma~\ref{lem:LBint}.

%% file: ThesisFiles/LB_improv/lemma_UB_TA.tex
\paragraph{Defining a coupling.} To prove Lemma~\ref{lem:UBTA}, we first note that, during the accumulation phase, we have two bounds on the number of busy servers $Z(t)$: it must be less than the total number of jobs $N(t)$, and it must be less than $R$; the former because every busy server must be working on a job, and the latter because otherwise the accumulation phase would be over. Thus, we can define a coupled $M/M/R$ system for which the number of jobs $\ncac(t)$ in the coupled system is always at least $N(t)$ in the original system.

\paragraph{How we use the coupling.} To use this coupled process to bound $\ex{T_A}$, recall that the accumulation point $T_A$ is the first time the $(R+1)$-th server turns on. Accordingly, one can also think of this as the first time that there have been at least $R+1$ jobs in the system for $\beta$ time.  Thus, if we define a coupled accumulation point
$\tac \triangleq \minpar{t \geq \beta : \min_{s \in [t-\beta, t)}{\ncac(t)} \geq R+1}$, then we know $\tac \geq T_A$. It thus suffices to bound $\ex{\tac}$.

\paragraph{General Strategy.} 
We bound $\ex{\tac}$ using the MIST method of Lemma~\ref{lem:waldstop}. As such, we define a few stopping times, then list the preconditions/claims that we will satisfy to complete the proof of Lemma~\ref{lem:UBTA}.

\paragraph{Definition of $\upc$ and $\dnc$.} Let the initial cycle-downcrossing occur at $\dnc_0 \triangleq  0$ and iteratively define the upcrossings $\upc$ and downcrossings $\dnc$ as
$$
    \upc_{i} \triangleq \minpar{t\geq \dnc_i: \ncac(t) \geq R+1} 
    \qquad  \mbox{and} \qquad 
    \dnc_{i+1} \triangleq \minpar{t \geq \upc_{i}: \ncac(t) \geq R+1}.
$$

\paragraph{Application of Lemma~\ref{lem:waldstop}, the IST Lemma.}
Applying Lemma~\ref{lem:waldstop} using $0= \dnc_0$ as our starting point, the coupled accumulation point $\tac$ as our ending point, our test function as $Y_t =1$, and the cycle-upcrossings $\left(\dnc_i\right)$ as our intervening stopping times, we now must prove that
\begin{equation}\label{eq:acc_init}
    \ex{\upc_{i} - \dnc_i \middle | \cind \geq i } \leq \frac{1}{\mu} e^{\frac{1}{12 R}} \sqrt{1 + \frac{1}{R}} \frac{\sqrt{2 \pi}}{\sqrt{R}} \leq \frac{c_3}{\mu \sqrt{R}},
\end{equation}
\begin{equation}\label{eq:acc_succ}
    \ex{\mintwo{\tac}{\dnc_{i+1}} - \upc_{i} \middle| \cind \geq i} \leq b_1\sqrt{\frac{\beta}{\mu R}} + \frac{6}{\mu R},
\end{equation}
and
\begin{equation}\label{eq:acc_prob}
    \pr{ \cind \geq i+1 \middle | \cind \geq i} \leq 1 -  \frac{b_1}{\sqrt{2}} e^{-\frac{1}{3(\mu 2 R \beta-1)}} \frac{1}{\sqrt{\mu 2 R \beta + 2}} \leq 1 - \frac{b_1 c_4}{\sqrt{\mu R \beta}},
\end{equation}
where $b_1 \triangleq \sqrt{\frac{2}{\pi}}$, $c_3 = 1.001\sqrt{2 \pi }$, and $c_4 = 0.499$.

\paragraph{Completion of Proof, assuming \eqref{eq:acc_init}, \eqref{eq:acc_succ}, and \eqref{eq:acc_prob}.}
Applying Lemma~\ref{lem:waldstop}, one finds that
\begin{align*}
    \ex{\tac} &= \sum_{i=0}^{\infty}{\ex{\mintwo{\tac}{\dnc_{i+1}} - \dnc_i \middle | \cind \geq i} \pr{\cind \geq i}}\\
    &\leq \left[\frac{c_3}{\mu \sqrt{R}} + \frac{b_1 \sqrt{\beta}}{\sqrt{\mu R}} + \frac{6}{\mu R} \right] \frac{\sqrt{\mu R \beta}}{b_1 c_4}\\
    &= \frac{1}{\mu}\left[\frac{c_3}{b_1 c_4}\sqrt{\mu \beta} + \frac{1}{c_4}\beta + \frac{6}{b_1 c_4} \sqrt{\frac{\mu \beta}{R}} \right].
\end{align*}

\subsubsection{\torp{Proof of \eqref{eq:acc_init}: Upper bound on initial up-crossing time.}{Upper bound on initial up-crossing time.}}
To prove \eqref{eq:acc_init}, we note that, since our coupled system is an $M/M/R$, the expected time $\ex{\upc_{i} - \dnc_i \middle | \cind \geq i }$ is simply the expected passage time from state $R$ to $(R+1)$ in an $M$/$M$/$R$ ( and equivalently an $M$/$M$/$R$/$(R+1)$, an $M$/$M$/$R$ which can contain only $R+1$ jobs. Solving, one finds that 
\begin{align}
    \ex{T_{R\to(R+1)}} &\leq \ex{T_{(R+1)\to(R+1)}}\nonumber\\
    &=\frac{1}{\mu(R+1)} \frac{1}{\pi_{R+1}}\nonumber\\
    &= \frac{1}{\mu(R+1)} \frac{\sum_{i=0}^{R+1}{\frac{R^i}{i!}}}{\frac{R^{R+1}}{(R+1)!}}\nonumber\\
    &\leq \frac{1}{\mu (R+1)} e^{R}\frac{(R+1)!}{R^{R+1}}\nonumber\\
    &\leq \frac{1}{\mu (R+1)} e^{R}\frac{e^{\frac{1}{12 (R+1)}}\sqrt{2 \pi (R+1)}(R+1)^{R+1} e^{-(R+1)}}{R^{R+1}}\label{eq:acc_2}\\
    &= e^{\frac{1}{12 (R+1)}}\frac{1}{\mu} \sqrt{2 \pi} \frac{\sqrt{R+1}}{R} \left(1 + \frac{1}{R}\right)^R e^{-1}\nonumber\\
    &\leq \frac{1}{\mu} e^{\frac{1}{12 (R+1)}} \sqrt{1 + \frac{1}{R}} \frac{\sqrt{2 \pi}}{\sqrt{R}}\label{eq:acc_4}\\
    &\leq \frac{1}{\mu} 1.006 \frac{\sqrt{2 \pi}}{\sqrt{R}}\nonumber\\
    &\triangleq \frac{c_3}{\mu \sqrt{R}}\nonumber,
\end{align}
\chmade{where in \eqref{eq:acc_2} we have applied Stirling's approximation and in \eqref{eq:acc_4} we have used that $R\geq 100$.}
\hfill\Halmos

\subsubsection{\torp{Proof of \eqref{eq:acc_succ}: Bound on time between up-crossings.}{ Bound on time between up-crossings.}}
To bound the expected time $\ex{\mintwo{\tac}{\dnc_{i+1}} - \upc_i \middle | \cind \geq i}$, we first note that, if $\upc_i + \beta \leq \dnc_{i+1}$, then $\tac = \upc_i + \beta$. Likewise, if $\upc_i + \beta > \dnc_{i+1}$, then $\tac > \dnc_{i+1}$. It follows that, given that $\cind \geq i$, the time $\mintwo{\tac}{\dnc_{i+1}} = \mintwo{\beta + \upc_i}{\dnc_{i+1}}$. Thus, we have that
\begin{equation*}
    \ex{\mintwo{\tac}{\dnc_{i+1}} - \upc_i \middle | \cind \geq i} = \ex{\mintwo{\beta}{\dnc_{i+1} - \upc_i} \middle | \cind \geq i} = \int_0^{\beta}{\pr{\dnc_{i+1} - \upc_i > s\middle| \cind \geq i} \dd s}.
\end{equation*}

We continue by bounding this tail probability. To begin, note that, while $\ncouple(t)$ stays above $R+1$, the dynamics of $\ncouple$ are precisely that of a critically-loaded $M$/$M$/$1$ queue with arrival rate and departure rate equal to $k\lambda$. The tail probability we are interested in bounding is precisely the probability that a busy period (started with 1 job) in such a system lasts longer than $s$ time. Applying Claim~\ref{clm:contdown}, one finds that, for any $t \geq \frac{3}{\mu 2 R}$,
\begin{equation*}
    \pr{\dnc_{i+1} - \upc_i > s\middle| \cind \geq i} \leq b_1 \left( \frac{1}{\sqrt{\mu 2 R s}} + \frac{b_2}{(\mu 2 R s)^{3/2}}\right).
\end{equation*}
Integrating, we find that
\chmades{
\begin{align*}
    \int_0^{\beta}{\pr{\dnc_{i+1} - \upc_i > s\middle| \cind \geq i} \dd s} &\leq \frac{3}{\mu 2 R} + \frac{b_1}{\sqrt{2}} \int_{\frac{3}{\mu 2 R}}^{\beta} {\frac{1}{\sqrt{\mu 2 R s}} + \frac{b_2}{(\mu 2 R s)^{3/2}} \dd s}\\
    &\leq \frac{3}{\mu 2 R} + \frac{b_1}{\sqrt{2}} \left[\sqrt{\frac{2 \beta}{ \mu R}} + b_2 \sqrt{\frac{2}{3}}  \frac{1}{\mu R} \right]\\
    &\leq \frac{2}{\sqrt{\pi}}\sqrt{\frac{\beta}{\mu R}} + \frac{6}{\mu R}.\tag*{\Halmos}
\end{align*}
}

\subsubsection{\torp{Proof of \eqref{eq:acc_prob}: Bound on probability of another $\upc$ up-crossing.}{Bound on probability of another up-crossing.}}
To prove \eqref{eq:acc_prob}, it suffices to note that, upon conditioning on the filtration at $\upc_i$, the probability $\pr{ \cind \geq i+1 \middle | \cind \geq i}$
is simply the probability that a busy period in a critically-loaded M/M/1, with arrival and departure rate equal to $\mu R$, ends before $\beta$ time has passed. Applying Claim~\ref{clm:contdown}, one finds that this is
\begin{equation*}
    \pr{ \cind \geq i+1 \middle | \cind \geq i} \geq 1 -  \frac{b_1}{\sqrt{2}} e^{-\frac{1}{3(\mu 2 R \beta-1)}} \frac{1}{\sqrt{\mu 2 R \beta + 2}}. 
    \tag*{\Halmos}
\end{equation*}

%% file: ThesisFiles/LB_improv/lemma_UB_TB.tex
\newcommand{\bind}{n_\zeta}
\newcommand{\tbthresh}{M_L}

We now prove the upper bound on $\ex{X - T_A}$. We make use of the ``wait-busy'' idea from Section~\ref{sec:waitbusy} as well as our main tool, Lemma~\ref{lem:waldstop}. As such, we begin by defining some stopping times.

\paragraph{Definition of $\dnb{i}$ and $\unb{i}$. }
Recall that the draining phase begins at time $T_A$. Let $\tbthresh \triangleq \min\left(k-R, \max\left(\frac{\sqrt{R}}{D_1 \sqrt{\beta}}, 1 \right)\right)$ be a specially-set analysis threshold. Let the stopping time $\dnb{1} \triangleq \minpar{t \geq T_A: N(t) < R + \tbthresh}$ be the first time the number of jobs $N(t)$ drops below $R +\tbthresh$, and recursively define
\begin{align*}
    \unb{i} &\triangleq \minpar{ t\geq \dnb{i}: N(t) \geq R + \tbthresh},\\
    \dnb{i+1} &\triangleq \minpar{ t \geq \unb{i}: N(t) < R+ \tbthresh}.
\end{align*}

\paragraph{Specification Step.}
Now, we apply Lemma~\ref{lem:waldstop} using the accumulation point $T_A$ as our initial point, the cycle end $X$ as our ending point, the constant function $Y_t = 1$ as our test function, and the draining-downcrossing points $\left(\dnb{i}\right)$ as our intervening points; we use $\bind$ to count the number of intervening points. To complete the proof, we must show that the following claims:
\begin{equation}\label{eq:tb_init}
    \ex{\dnb{1} - T_A} \leq \beta +  \frac{1}{\mu}\frac{F_1 \beta \sqrt{R}}{k - R} + \frac{1}{\mu}\frac{3}{2}\ln(\beta) + \frac{1}{\mu}\ln(F_1 D_1) ,
\end{equation}
\begin{equation}\label{eq:tb_cont}
    \ex{ \mintwo{X}{\dnb{i+1}} - \dnb{i} \middle | \bind \geq i} \leq \frac{D_2}{\mu \sqrt{R}} + \frac{D_3}{\mu R} + \frac{2}{\mu \tbthresh} 
\end{equation}
\begin{equation}\label{eq:tb_prob}
    \pr{ \bind \geq i+ 1 \middle | \bind \geq i} \leq \frac{1}{\tbthresh}.
\end{equation}

\paragraph{Completion of Proof assuming \eqref{eq:tb_init}, \eqref{eq:tb_cont}, and \eqref{eq:tb_prob}.}
Before proving the claims, we now prove the lemma. It suffices to give a bound on $\ex{X - \dnb{1}}$; applying Lemma~\ref{lem:waldstop} gives
\begin{align*}
    \ex{X- \dnb{1}} &\leq \tbthresh \left[ \frac{D_2}{\mu \sqrt{R}} + \frac{D_3}{\mu R} + \frac{2}{\mu \tbthresh}\right]
    = \frac{2}{\mu} + \left[D_2 + \frac{D_3}{\sqrt{R}} \right]\max\left(\frac{1}{D_1 \sqrt{\mu \beta}}, \frac{1}{\sqrt{R}}\right). 
\end{align*}

\subsubsection{\torp{Proof of \eqref{eq:tb_init}: Upper bound on time until first downward visit.}{Upper bound on time until first downward visit.}}
To bound $\ex{\dnb{1} - T_A}$, we make a coupling argument then apply basic results on M/M/1 busy periods. Moreover, instead of proving \eqref{eq:tb_init} directly, we first show a more general claim.
\begin{claim}\label{clm:tb_induct}
    For $\tbthresh \leq j \leq N(T_A) - R$, define $\eta_j$ as the first time after $T_A$ that $N(t) \leq R+ j$. Note that this means that $\eta_{N(T_A)} = T_A$ and $\eta_{\tbthresh} = \dnb{1}$. Then we have the following bound:
    \begin{equation*}
        \ex{\eta_{\tbthresh} - \eta_{j} \middle | \filt{\eta_j}} \leq \rem{R+j}(\eta_j) + \frac{1}{\mu}\sum_{i= \tbthresh}^{j}{\frac{1}{\mintwo{i}{k-R}}}. 
    \end{equation*}
\end{claim}
\newcommand{\remrj}{\rem{R+j}\left(\eta_j \right)}
Afterwards, we complete the proof by noting that $\left[N(T_A) - k\right]^+ \leq \left[N(T_A) - R \right]$, taking expectations, applying Jensen's inequality to the minimum function and the $\ln(\cdot)$ (which is concave), using the bound on $\ex{N(T_A) - R}$ from Claim~\ref{clm:taexpect}, then letting $h = \tbthresh$.

\paragraph{Proof of Claim~\ref{clm:tb_induct}.}
We prove Claim~\ref{clm:tb_induct} by induction. In the base case, suppose that $j = \tbthresh + 1$. Note that at time $\eta_{\tbthresh + 1}$, the numbers of jobs $N(\eta_{\tbthresh + 1}) = R + \tbthresh + 1$ and the remaining time until the $(R+\tbthresh+1)$-th server turns on is $\rem{R + \tbthresh+1}\left(\eta_{\tbthresh + 1}\right)$. As such, we can simply wait until either that server turns on, in which case we can analyze the system as an M/M/1 busy period with departure rate $\mu \mintwo{R+\tbthresh+1}{k}$, or the number of jobs $N(t)$ drops below $R+\tbthresh + 1$ on its own. In other words, (using $j$ here to save space)
\begin{equation*}
    \ex{\eta_{j-1} - \eta_{j}\middle | \filt{\eta_{j}}} 
    \leq \remrj + 
    \frac{\ex{\left[N\left(\eta_j + \rem{R+j}\left(\eta_j\right)\right) - \left(R+(j-1)\right)  
        \right] \indc{\eta_{j-1} > \eta_{j} + \remrj} \middle | \filt{\eta_j}}}
    {\mu \mintwo{j}{k-R}}.
\end{equation*}
Now, we reframe the expectation as an expectation up to a stopping time. We note that, if $\eta_{j-1} > \eta_j + \remrj$, then we have that
\begin{equation*}
    N\left(\eta_j + \remrj \right) = N\left(\mintwo{\eta_j + \remrj}{\eta_{j-1}}\right).
\end{equation*}
Likewise, if $\eta_{j-1} \leq \eta_j + \remrj$, then
\begin{equation*}
     R + j - 1 =  N\left(\eta_{j-1}\right) = N\left( \mintwo{\eta_j + \remrj}{\eta_{j-1}}\right).
\end{equation*}
Using this and applying a simple coupling argument, one sees that
\begin{align*}
    &\ex{\left[N\left(\eta_j + \rem{R+j}\left(\eta_j\right)\right) - \left(R+(j-1)\right)  
        \right] \indc{\eta_{j-1} > \eta_{j} + \remrj} \middle | \filt{\eta_j}}\\
        &\;\; = \ex{N\left(\mintwo{\eta_j + \remrj}{\eta_{j-1}}\right) - (R + j - 1) \middle | \filt{\eta_j}}\\
        &\;\; \leq N\left(\eta_j\right) - (R + j - 1) = 1.
\end{align*}
Thus, we find that
\begin{equation*}
    \ex{\eta_{j-1} - \eta_{j}\middle | \filt{\eta_{j}}} 
    \leq \remrj + \frac{1}{\mu \mintwo{j}{k-R}}.
\end{equation*}

\paragraph{Inductive case.}
The inductive case proceeds in much the same way, except now, if $N(t)$ does drop below $R+j$ ``early'', then we can factor in the time that has elapsed in the value of $\rem{R+j}(\eta_j)$.
In particular, note that, since the $(R+j)$-th server would have already turned on,
\begin{equation*}
\ex{\eta_{\tbthresh} - \eta_j \middle |\filt{\eta_j}}\indc{\eta_j \geq \eta_{j+1} + \rem{R+j+1}(\eta_{j+1})}\leq \frac{1}{\mu}\sum_{i=\tbthresh}^{j}{\frac{1}{\mu \mintwo{i}{k-R}}}\indc{\eta_j \geq \eta_{j+1} + \rem{R+j+1}(j+1)}.
\end{equation*}
It follows that
\begin{equation*}
\ex{\eta_{\tbthresh} - \eta_j \middle |\filt{\eta_j}}\leq \rem{R+j}\left(\eta_j\right)\indc{\eta_j < \eta_{j+1} + \rem{R+j+1}\left(\eta_{j+1}\right)} + \frac{1}{\mu}\sum_{i=\tbthresh}^{j}{\frac{1}{\mu \mintwo{i}{k-R}}}.
\end{equation*}
Now, we note that
\begin{align*}
    \remrj \indc{\eta_j < \eta_{j+1} + \rem{R+1+j}\left(\eta_{j+1}\right)} &= \left[\remrj + \eta_j - \eta_j \right] \indc{\eta_j < \eta_{j+1} + \rem{R+1+j}\left(\eta_{j+1}\right)}\\
    &= \left[\rem{R+j}\left(\eta_{j+1}\right) + \eta_{j+1} - \eta_j \right] \indc{\eta_j < \eta_{j+1} + \rem{R+1+j}\left(\eta_{j+1}\right)}\\
   &\leq  \left[\rem{R+j+1}\left(\eta_{j+1}\right) + \eta_{j+1} - \eta_j \right] \indc{\eta_j < \eta_{j+1} + \rem{R+1+j}\left(\eta_{j+1}\right)}\\
   &= \left[\rem{R+j+1}\left(\eta_{j+1}\right) + \eta_{j+1} - \eta_j \right]^+,
\end{align*}
so that we find
\begin{equation*}
    \ex{\eta_{\tbthresh} - \eta_j \middle |\filt{\eta_j}}\leq \left[\rem{R+j+1}\left(\eta_{j+1}\right) + \eta_{j+1} - \eta_j \right]^+ + \frac{1}{\mu}\sum_{i=\tbthresh}^{j}{\frac{1}{\mu \mintwo{i}{k-R}}}.
\end{equation*}
Finally, we note that
\begin{equation*}
    \ex{\eta_{j} - \eta_{j+1} \middle | \filt{\eta_{j+1}}} \leq \ex{\mintwo{\eta_{j} - \eta_{j+1}}{\rem{R+j+1}\left(\eta_{j+1}\right)}\middle | \filt{\eta_{j+1}}} + \frac{1}{\mu \mintwo{j+1}{k-R}}.
\end{equation*}
Summing these final two expressions gives the inductive result, proving Claim~\ref{clm:tb_induct}. 

\paragraph{Using Claim~\ref{clm:tb_induct}.}
Thus, we obtain that, using $H_i$ to denote the $i$-th harmonic number,
\begin{align*}
    \ex{\dnb{1} - T_A\middle | \filt{T_A}} &\leq \beta + \frac{1}{\mu}\frac{\left[N(T_A) - k\right]^+}{k-R} + \frac{1}{\mu}\left[H_{\mintwo{N(T_A)-R}{k-R}} - H_{\tbthresh}\right]\\
    &\leq \beta + \frac{1}{\mu}\frac{\left[N(T_A) - R\right]^+}{k-R} + \frac{1}{\mu}\ln\left(\frac{\mintwo{N(T_A) - R}{k-R}}{M_L}\right). 
\end{align*}
Taking expectations and applying Jensen's inequality twice, we find
\begin{align*}
     \ex{\dnb{1} - T_A\middle | \filt{T_A}} &\leq \beta + \frac{1}{\mu} \frac{F_1 \mu \beta \sqrt{R}}{k-R} + \frac{1}{\mu}\ln\left(\frac{F_1 \mu \beta \sqrt{R}}{M_L}\right)\\
     &\leq \beta + \frac{1}{\mu} \frac{F_1 \mu \beta \sqrt{R}}{k-R} + \frac{1}{\mu}\ln\left(\frac{\min\left(F_1 \mu \beta \sqrt{R}, k - R\right)}{\min\left(\max\left(1,\frac{\sqrt{R}}{D_1 \sqrt{\beta}}\right) , k-R\right)}\right)\\
     &\leq \beta + \frac{1}{\mu} \frac{F_1 \mu \beta \sqrt{R}}{k-R} + \frac{1}{\mu}\ln\left(F_1 D_1 \beta^{3/2}\right)\\
     &= \beta + \frac{1}{\mu} \frac{F_1 \mu \beta \sqrt{R}}{k-R} + \frac{1}{\mu}\frac{3}{2}\ln\left(\beta\right) + \frac{1}{\mu} \ln\left(F_1 D_1\right). \tag*{\Halmos}
\end{align*}

\subsubsection{\torp{Proof of \eqref{eq:tb_cont}: Upper Bound on Time between Consecutive Downward Visits.}{Upper Bound on Time between Consecutive Downward Visits.}}\label{sec:downtodown}
To bound the expectation $\ex{\mintwo{\dnb{i+1}}{X} - \dnb{i} \middle | \filt{\dnb{i}}}$, we split the interval into two parts, $\left[\mintwo{\unb{i}}{X} - \dnb{i}\right]$ and $\left[\dnb{i+1} - \dnb{i}\right]$.

To bound the expectation of the first quantity, it suffices to note that, if we couple the system to an M/M/$\infty$, the coupled number of jobs $\ncouple(t)$ will reach $R + \tbthresh$ only after the original system. Using Claim~\ref{clm:mminf} to bound this passage time, we thus know that
\begin{align*}
    \ex{\mintwo{\unb{i}}{X} - \dnb{i}\middle | \filt{\dnb{i}}} &\leq \ex{\mintwo{T^{M/M/\infty}_{(R+\tbthresh - 1)\to (R+ \tbthresh)} + \dnb{i}}{X} - \dnb{i} \middle | \filt{\dnb{i}}}\\
    &\leq \ex{T^{M/M/\infty}_{(R+\tbthresh - 1)\to (R+ \tbthresh)}}\\
    &\leq \frac{D_2}{\sqrt{R}}.
\end{align*}

To bound the expectation of the second quantity, we provide two bounds. First, we again make use of the ``wait-busy'' idea; as we argued in the proof of \eqref{eq:tb_init},
\begin{equation*}
    \ex{\dnb{i+1} - \unb{i}\middle | \filt{\unb{i}}} \leq \ex{\mintwo{\dnb{i+1} - \unb{i}}{\beta}\middle | \filt{\unb{i}}} + \frac{1}{\mu \tbthresh}.
\end{equation*}
From here, we note, by coupling to an M/M/1 with arrival rate and departure rate both equal to $k \lambda$, we can bound $\ex{\mintwo{\dnb{i+1} - \unb{i}}{\beta}\middle | \unb{i} < X}$ by the expected minimum between $\beta$ and the length of a single-job busy period in that system. Applying Claim~\ref{clm:contint}, we can complete the proof, finding that
\begin{equation*}
    \ex{\mintwo{\dnb{i+1} - \unb{i}}{\beta}\middle | \filt{\unb{i}}} \leq D_1 \frac{\sqrt{\beta}}{\sqrt{\mu R}} + \frac{6}{\mu R}.
\end{equation*}

For the second bound, we simply note that, during the draining phase, the number of busy servers $Z(t) \geq R+1$. It follows from a simple coupling argument that
\begin{equation*}
     \ex{\mintwo{\dnb{i+1} - \unb{i}}{\beta}\middle | \filt{\unb{i}}} \leq \frac{1}{\mu}.
\end{equation*}
Combining the bounds pessimistically, we find that
\begin{align*}
    \ex{\mintwo{\dnb{i+1}}{X} - \dnb{i} \middle | \filt{\dnb{i}}} &\leq \frac{D_2}{\mu \sqrt{R}} + \frac{D_3}{\mu R} + \frac{1}{\mu\tbthresh} + \min \left(D_1 \frac{\sqrt{\beta}}{\sqrt{\mu R}}, \frac{1}{\mu}\right)\\
    &\leq  \frac{D_2}{\mu \sqrt{R}} + \frac{D_3}{\mu R} + \frac{2}{\mu \tbthresh}.\tag*{\Halmos}
\end{align*}

\subsubsection{\torp{Proof of \eqref{eq:tb_prob}: Upper Bound on Probability of Another Downward Visit.}{Upper Bound on Probability of Another Downward Visit.}}
To bound the probability of an additional downcrossing, we again make a coupling argument. In particular, we couple again to the system which only has $R$ servers busy, which gives an upper bound on the number of jobs in the system $N(t)$. If, in our coupled system, we reach $\ncouple(t) = R + \tbthresh$ before we reach $\ncouple(t) = R$, then another upcrossing \emph{must} have previously occurred in the original system, and thus another downcrossing must also occur. But, of course, we know classically that the probability that this happens is just $\frac{1}{\tbthresh}$; this is precisely what is asserted by \eqref{eq:tb_prob}. \hfill \Halmos

%% file: 8_claims.tex
\newcommand{\dgcouple}{\Tilde{\dgen}}
\newcommand{\startpt}{\tau}
\newcommand{\intend}{\tau + \mintwo{\ell}{\dgen}}

\newcommand{\cnine}{\left(1 - \frac{C_9}{\sqrt{\mu \beta R}}}





\section{Hitting Time Bounds}
\input{claims/hitting_time_bounds}

\section{Miscellaneous Claims}\label{sec:appendix}

\input{claims/coupling_expectations}

\subsection{Proof of Claim~\ref{clm:bpclaim}, the Busy Period Integral Bound.}

\input{claims/bpclaim}

\subsection{Proof of Claim~\ref{clm:waitbusy2}, the Wait Busy Claim.}

\waitbusy*

\input{claims/waitbusyv2}




\subsection{Proof of Claim~\ref{clm:berryess}.}
\input{claims/berry_ess}

\input{claims/lbound}

\subsection{\texorpdfstring{Proof of Claim~\ref{clm:taexpect}: Upper Bound on $\ex{\nta}$.}{Proof of Claim~\ref{clm:taexpect}: Upper Bound on E[N(T\_A)]} }\label{sec:taexpect}
\taexpect*

\input{claims/taexpect}




%% file: claims/hitting_time_bounds.tex
\subsection{Proof of Claim~\ref{clm:discretedown}, Discrete-Time Hitting Time Tail Bound.}\label{sec:discretedown}
\begin{restatable}[Discrete-Time Hitting Time Tail Bound]{claim}{discretedown}\label{clm:discretedown}
    Suppose one has an upwards-biased discrete random walk $V(t)$ where in each step 
    \begin{equation*}
        \pr{V(t+1) = V(t) + 1\middle| \filt{t}} = p = 1-q,
    \end{equation*}
    where $p \geq \frac{1}{2} \geq q$. Suppose that $V(0) = 1$ and let the hitting time $\gamma \triangleq \minpar{t \in \mathcal{N}: V(t) = 0}$ be the first timestep where the walk $V(t) = 0$. Then, for $n \geq 1$,
    \begin{equation*}
        \pr{ \gamma \geq 2m + 1} \leq \frac{1}{\sqrt{\pi}}  \frac{2q}{\sqrt{m}}\left(1 + \frac{1}{2(m+1)}\right).
    \end{equation*}
\end{restatable}

Moreover, if $p=q=\frac{1}{2}$, then
\begin{equation*}
    \pr{ \gamma \geq 2m + 1} \geq \frac{1}{\sqrt{\pi}} e^{-\frac{1}{6 m}} \frac{1}{\sqrt{m+1}}.
\end{equation*}
\subsubsection{Proof}

 First note, as in \cite{WHW}, that by a counting argument $\pr{\gamma = 2\ell + 1} = q \left(qp\right)^\ell C_\ell$, where $C_\ell\triangleq \frac{1}{\ell +1}\frac{(2\ell)!}{\ell!\ell!}$ is the $\ell$-th Catalan number; note that $\gamma$ can not be even, since the number of downward steps must exceed the number of upward steps by exactly 1.

We proceed by bounding the Catalan numbers using Stirling's approximation. For $m = 0$, then $\pr{\gamma \geq 1} = \pr{\gamma \geq 2} = p$, i.e. the probability that the first step is an upward step. For $m \geq 1$, applying Stirling's approximation and simplifying gives
\begin{equation*}
    e^{-\frac{1}{6 \ell}} \frac{1}{\sqrt{\pi \ell}(\ell+1)} q \left(4 p q\right)^\ell \leq  \pr{\gamma = 2\ell + 1} \leq \frac{1}{\sqrt{\pi \ell}(\ell + 1)} q \left(4 p q\right)^\ell.
\end{equation*}

\paragraph{Lower bound.}
Since we are interested in the lower bound only when $q = p = \frac{1}{2}$, we obtain that
\begin{align*}
    \pr{\gamma \geq 2m + 1} &\geq  \frac{1}{\sqrt{\pi}}\frac{1}{2} \sum_{\ell = m}^{\infty}{\frac{e^{-\frac{1}{6\ell}}}{\sqrt{\ell}(\ell + 1)}}\geq \frac{1}{\sqrt{\pi}} \frac{1}{2} \sum_{\ell = m}^{\infty}{ \frac{ e^{-\frac{1}{6 m}} }{\sqrt{\ell}(\ell +1)}}\geq \frac{1}{\sqrt{\pi}}\frac{1}{2} \int_{m}^{\infty}{ \frac{ e^{-\frac{1}{6 m}} }{(\ell+1)^{3/2}} \dd \ell} = \frac{1}{\sqrt{\pi}} e^{-\frac{1}{6 m}} \frac{1}{\sqrt{m+1}}.
\end{align*}
\paragraph{Upper bound.}
Noting that $4pq \leq 1$, we have likewise that
\begin{align*}
   \pr{\gamma \geq 2m + 1} &\leq \frac{1}{\sqrt{\pi}} q \sum_{\ell = m}^{\infty}{\frac{1}{\sqrt{\ell}(\ell + 1)}}  \leq \frac{1}{\sqrt{\pi}} q \frac{1}{\sqrt{m}(m+1)} + \int_{m}^{\infty}{\frac{1}{\ell^{3/2}} \dd \ell} =\frac{1}{\sqrt{\pi}} q \frac{2}{\sqrt{m}}\left(1 + \frac{1}{2(m+1)}\right).
\end{align*}

\subsection{Proof of Claim~\ref{clm:contdown}, Continuous-Time Hitting Time Tail Bound.}\label{sec:contdown}
We further extend this discrete-time bound into a continuous-time bound.
\begin{restatable}[Continuous-Time Hitting Time Tail Bound]{claim}{contdown}\label{clm:contdown}
    Suppose one has an Poisson arrival process $Y_A(t)$ of rate $k \lambda$ and a Poisson departure process $Y_D(t)$ of rate $\mu (R - j)$, for some integer $j\geq 0$. Let the continuous random walk $X_c(t) = Y_A(t) - Y_D(t)$, with $X_c(0) = 1$, and define $\gamma_c \triangleq \minpar{t>0: X_c(t) = 0}$. Let $\nu = (2R - j)\mu t$. For any $\nu \geq 3$, we have
    \begin{equation*}
        \pr{\gamma_c \geq t} \leq  \frac{b_1}{\sqrt{2}} \left(\frac{1}{\sqrt{\nu}}  + \frac{b_2}{\nu^{3/2}}\right)\\
    \end{equation*}
    where $b_1 = \sqrt{\frac{2}{\pi}}$ and $b_2 = 1 + \frac{2.5}{b_1 \sqrt{2}}$.

    Moreover, if $j = 0$, then
    \begin{equation*}
        \pr{\gamma_c \geq t} \geq \frac{b_1}{\sqrt{2}} e^{-\frac{1}{3(\nu-1)}} \frac{1}{\sqrt{\nu + 2}}.
    \end{equation*}
\end{restatable}
\newcommand{\epsj}{\indc{\text{$j$ is even}}}
\subsubsection{Proof of Upper Bound.}
To prove this claim, we first condition on the value of $Y_T = Y_A(t) + Y_D(t)$, the total number of Poisson events during the interval $[0,t]$, then relate that to the same question in a discrete-time random walk, a la Claim~\ref{clm:discretedown}. Note that $Y_T \sim \Poisson(\nu)$, and thus
\begin{align*}
    \pr{ \gamma_c \geq t} = \pr{\gamma \geq Y_T}
    &= \sum_{j=0}^{\infty}{e^{-\nu}\frac{\nu^j}{j!} \pr{ \gamma \geq j}}\\
    &= e^{-\nu} + 2p\nu e^{-\nu} +\sum_{j=3}^{\infty}{e^{-\nu}\frac{\nu^j}{j!} \pr{ \gamma \geq j + \epsj}}\\
    &=  e^{-\nu} + 2p\nu e^{-\nu} + \sum_{j=0}^{\infty}{e^{-\nu}\frac{\nu^j}{j!} \pr{ \gamma \geq 2\left(\frac{j + \epsj - 1}{2}\right)+1}}.
\end{align*}
Applying the discrete upper bound to the sum, we obtain
\begin{align*}
    &\sum_{j=3}^{\infty}{e^{-\nu}\frac{\nu^j}{j!} \pr{ \gamma \geq 2\left(\frac{j + \epsj - 1}{2}\right)+1}}\\
    &\;\;\leq b_1 \sqrt{2} q \sum_{j=3}^{\infty}{e^{-\nu}\frac{\nu^j}{j!}  \frac{1}{\sqrt{j + \epsj - 1}}\left(1 + \frac{1}{(j + \epsj + 1)}\right)}\\
    &\;\;= b_1 \sqrt{2} q \frac{1}{\nu}\sum_{j=3}^{\infty}{e^{-\nu}\frac{\nu^{(j+1)}}{(j+1)!}  \frac{1}{\sqrt{j + \epsj - 1}}\left(j+1 + \frac{j+1}{j + \epsj + 1}\right)}\\
    &\;\;\leq b_1 \sqrt{2} q \frac{1}{\nu}\sum_{j=3}^{\infty}{e^{-\nu}\frac{\nu^{(j+1)}}{(j+1)!}  \frac{j+2}{\sqrt{j + \epsj - 1}}}\\
    &\;\;\leq b_1 \sqrt{2} q \frac{1}{\nu}\sum_{j=3}^{\infty}{e^{-\nu}\frac{\nu^{(j+1)}}{(j+1)!}  \frac{j+\epsj - 1 + 3}{\sqrt{j + \epsj - 1}}}\\
    &\;\;= b_1 \sqrt{2} q \frac{1}{\nu}\sum_{j=3}^{\infty}{e^{-\nu}\frac{\nu^{(j+1)}}{(j+1)!}  \left(\sqrt{j + \epsj -1} + \frac{3}{\sqrt{j + \epsj -1}} \right)}.\intertext{From here, we note that the function $f(x) = \sqrt{x} + \frac{3}{\sqrt{x}}$ is both increasing and concave for all $x\geq 3$. After increasing the argument and applying Jensen's inequality, we find that}
    &\;\;\leq b_1 \sqrt{2} q \frac{1}{\nu}\sum_{j=3}^{\infty}{e^{-\nu}\frac{\nu^{(j+1)}}{(j+1)!}  \left(\sqrt{j + 1} + \frac{3}{\sqrt{j + 1}} \right)}\\
    &\;\; \leq b_1 \sqrt{2} q \frac{1}{\nu} \left( \sqrt{\nu} + \frac{3}{\sqrt{\nu}}\right),
\end{align*}
where in the final line we have used that the function $f(x)$ is increasing in $x$ for any $x \geq 3$, and that $\ex{Y_T\indc{Y_T \geq 4}} \geq \nu - 3 \geq 3$. Thus, we have that
\begin{align*}
    \pr{\gamma_c \geq t} &\leq \left(3 \nu\right) e^{-\nu} + 2q \sqrt{\frac{2}{\pi}} \left(\frac{1}{\sqrt{\nu}}  + \frac{1}{\nu^{3/2}}\right)
     \leq \frac{2.5}{\nu^{3/2}} + 2q \sqrt{\frac{2}{\pi}} \left(\frac{1}{\sqrt{\nu}}  + \frac{1}{\nu^{3/2}}\right).
\end{align*}

\subsubsection{Proof of Lower Bound.}
We approach the initial stages of the proof in the precisely the same way:
\begin{align*}
    \pr{ \gamma_c \geq t} &= \pr{\gamma \geq Y_T}\\
    &=  e^{-\nu} + 2p\nu e^{-\nu} + \sum_{j=3}^{\infty}{e^{-\nu}\frac{\nu^j}{j!} \pr{ \gamma \geq 2\left(\frac{j + \epsj - 1}{2}\right)+1}}\\
    &\geq \sum_{j=3}^{\infty}{e^{-\nu}\frac{\nu^j}{j!}  b_1 e^{-\frac{1}{3 (j+ \epsj - 1)}} \frac{q\sqrt{2}}{\sqrt{(j + \epsj + 1)}}}\\
    &\geq \sum_{j=3}^{\infty}{e^{-\nu}\frac{\nu^j}{j!}  b_1 e^{-\frac{1}{3 (j - 1)}} \frac{q\sqrt{2}}{\sqrt{(j + 2)}}} \\
    &\geq b_1 q \sqrt{2} e^{-\frac{1}{3(\nu-1)}} \frac{1}{\sqrt{\nu + 2}}. \quad \mbox{ (via Jensen's inequality)}\tag*{\Halmos}
\end{align*}

\input{claims/contint}

\subsection{\torp{Proof of Claim~\ref{clm:mminf}, Bound on the Expected Hitting Time in the M/M/$\infty$.}{ Proof of Claim~\ref{clm:mminf}, Bound on the Hitting Time in the M/M/inf. }}\label{sec:mminf}

\input{claims/mminf}

%% file: claims/contint.tex
\subsection{Proof of Claim~\ref{clm:contint}, Bound on Expected Length of Stopped Random Walk.}
\begin{restatable}[Bound on Expected Length of Stopped Random Walk]{claim}{contint}\label{clm:contint}
Suppose we have a critically loaded M/M/1 queue with arrival rate and departure rate both equal to $k\lambda$, with offered load $R > 100$ and setup time $\beta > 100$. Suppose also that at time $0$, a job arrives. Let $\tau$ be the length of the busy period which follows. Then, letting $b_1 \triangleq \sqrt{\frac{2}{\pi}}$,
\begin{equation*}
    \ex{\mintwo{\beta}{\tau}} \leq b_1 \frac{\sqrt{\beta}}{\sqrt{\mu R}} + \frac{6}{\mu R}.
\end{equation*}

\end{restatable}

\subsubsection{Proof.}
From Claim~\ref{clm:contdown}, the continuous-time random walk hitting time bound, we have that
\begin{equation}
    \pr{\tau \geq t} \leq  \frac{b_1}{\sqrt{2}} \left(\frac{1}{\sqrt{\nu}}  + \frac{b_2}{\nu^{3/2}}\right),
\end{equation}
where $\nu=2 \mu R t$ and we require that $\nu \geq 3$.
By integrating this bound (using a bound of $1$ wherever this bound doesn't apply), we obtain
\begin{align*}
    \ex{\mintwo{\beta}{\tau}} &= \int_0^{\beta}{\pr{\tau > t} \dd t} \leq \frac{3}{2 \mu R} + \int_{\frac{3}{2\mu R}}^{\beta}{ \frac{b_1}{\sqrt{2}} \left(\frac{1}{\sqrt{2 \mu R t}}  + \frac{b_2}{\left(2\mu R\right)^{3/2}}\right) \dd t}\\
    &\leq \frac{3}{2\mu R} + b_1 \frac{\sqrt{\beta}}{\sqrt{\mu R}} +  \frac{b_1 b_2}{4 \mu R}\left[2\sqrt{\frac{2}{3}}\right] \leq b_1 \frac{\sqrt{\beta}}{\sqrt{\mu R}} + \frac{6}{\mu R}. \tag*{\Halmos}
\end{align*}

%% file: claims/mminf.tex
\begin{restatable}[M/M/$\infty$ Passage Time Bound]{claim}{mminf}\label{clm:mminf}
    Given an M/M/$\infty$ queue, let $T_{x \to y}$ denote the random amount of time taken to go from state $x$ to state $y$. Suppose this system has an arrival rate of $k\lambda$ and a per-server departure rate of $\mu$. Let $R\triangleq k \frac{\lambda}{\mu}$. Then, for any $h$ such that $1\leq h \leq \sqrt{R}$,
    \begin{equation*}
        \ex{T_{(R+h-1)\to (R+h)}} \leq \frac{\sqrt{2\pi}}{\mu \sqrt{R}}
        \left(1 + \frac{h}{R}\right)^{h - \frac{1}{2}} e^{\frac{1}{12R}} 
        \leq D_2 \frac{\sqrt{\pi}}{\mu \sqrt{R}}.
    \end{equation*}
\end{restatable}

\subsubsection{Proof.}
The proof here is quite simple. 
First, we note that the passage time in the M/M/$\infty$ from state $(R+h-1)$ to state $(R+h)$ is exactly the passage time from those states in the M/M/$(R+h)$/$(R+h)$. 
This new system has a nice product form, so that
\chmades{
\begin{align*}
    \ex{T_{(R+h-1)\to(R+h)}} \leq \ex{T_{(R+h)\to(R+h)}}=\frac{1}{\mu(R+h)} \frac{1}{\pi_{R+h}} = \frac{1}{\mu(R+h)} \frac{\sum_{i=0}^{R+h}{\frac{R^i}{i!}}}{\frac{R^{R+h}}{(R+h)!}} \leq \frac{1}{\mu (R+h)} e^{R}\frac{(R+h)!}{R^{R+h}}.
\end{align*}
}
\chmades{
Applying Stirling's approximation,
\begin{align*}
    \frac{1}{\mu (R+h)} e^{R}\frac{(R+h)!}{R^{R+h}} \leq \frac{1}{\mu} e^{R}\frac{e^{\frac{1}{12 (R+h)}}\sqrt{2 \pi (R+h)}(R+h)^{R+h-1} e^{-(R+h)}}{R^{R+h}} \leq  e^{\frac{1}{12 R}}
    \frac{1}{\mu\sqrt{R+h}}\sqrt{2 \pi}  \left(1 + \frac{h}{R}\right)^{R+h} e^{-h}.
\end{align*}
}
\chmades{
Continuing on with this term,
\begin{align*}
     e^{\frac{1}{12 R}}
    \frac{1}{\mu\sqrt{R+h}}\sqrt{2 \pi}  \left(1 + \frac{h}{R}\right)^{R+h} e^{-h}
    \leq \frac{\sqrt{2\pi}}{\mu \sqrt{R}}
        \left(1 + \frac{h}{R}\right)^{h - \frac{1}{2}} e^{\frac{1}{12R}}
    \leq \frac{1}{\mu}\frac{\sqrt{2\pi}}{\sqrt{R}}e^{\frac{h^2}{R}} e^{\frac{1}{12 R}}
    \leq \frac{7}{\mu \sqrt{R}},
\end{align*}
}
where we have made extensive use of Stirling's approximation and the bound $(1+x) \leq e^x$.

%% file: claims/coupling_expectations.tex
\subsection{\torp{Proof of Claim~\ref{clm:couplingEx}: Bound on Expected Value After Coupling.}{Bound on Expected Value After Coupling.}}\label{sec:couplingEx}

\begin{restatable}[Bound on Expected Value after Coupling.]{claim}{couplingEx}\label{clm:couplingEx}
Let $\tau$ be some stopping time and $\dgen$ be the next down-crossing as described in Section~\ref{sec:remCoupling}. Suppose that we have a \textbf{lower} bound on the number of busy servers $Z(t)$ over some length $\ell$ interval starting at time $\tau$, i.e. the busy servers $Z(t) \geq R - j,$ for all $t \in \left[\tau, \tau + \mintwo{\ell}{\dgen} \right]$ and for some non-negative $j$. Then, bounding the first moment,
\begin{equation}\label{eq:couplingfirst}
    \ex{\left[N\left(\tau + \ell\right) - h\right] \indc{\dgen > \ell} \middle | \filt{\tau}} \leq \left[N(\tau) - h\right] + \mu j \ell,
\end{equation}
and, bounding the second moment,
\begin{equation}\label{eq:couplingsecond}
    \ex{\left[N\left(\tau + \ell\right) - h\right] \indc{\dgen \geq \ell}} \leq \left[N(\tau) - h + \mu j \ell\right]^2 + 2 \mu R \ell.
\end{equation}

\end{restatable}

\subsubsection{Proof.}
The proof is essentially an application of Doob's Optional Stopping Theorem to an appropriately selected martingale. To begin, we define a coupled process $\ncouple(t)$ with
\begin{equation*}
    \ncouple(t - \tau) \triangleq N(\tau) + A[\tau, t] - \depOp{R-j}{[\tau, t]};
\end{equation*}
by Claim~\ref{clm:coupling}, we know that $\ncouple(t - \tau) \geq N(t)$ for any $t \in \left[ \tau, \tau + \mintwo{\dgen}{\ell}\right]$, and that the coupled hitting time $\dgcouple\triangleq \minpar{t>0: \ncouple(t) \leq h}$ can not be smaller than the original hitting time $\dgen$. It follows that
\begin{equation*}
    N\left(\tau + \ell\right)\indc{\dgen > \ell} \leq \ncouple\left(\ell\right) \indc{\dgcouple > \ell}.
\end{equation*}
Thus, we bound coupled versions of \eqref{eq:couplingfirst} and \eqref{eq:couplingsecond}.

\paragraph{Construction of martingales.}
We now construct our martingales and set up the language of optional stopping.
Note that, for any process $\ncouple(t)$ with independent, stationary increments, both functions $V_1$ and $V_2$, defined as
\begin{equation*}
    V_1(t) \triangleq  \left[\ncouple(t) -h\right] - \ex{\ncouple(t) - \ncouple(0)}
\end{equation*}
and 
\begin{align*}
    V_2(t) &\triangleq \left[ \ncouple(t) - h - \ex{\ncouple(t) - \ncouple(0)} \right]^2 - \ex{ \left[ \ncouple(t) - h - \ex{\ncouple(t) - \ncouple(0)} \right]^2 }\\
    &= \left(\ncouple(t) - h - \mu j t \right)^2 - \mu \left(2R - j\right) t
\end{align*}
are martingales \cite{karatzas2012brownian}. Moreover, one has that
\begin{align*}
    \left[\ncouple(\ell) - h\right]\indc{\dgcouple > \ell} &= \left[\ncouple\left(\mintwo{\dgcouple}{\ell}\right) - h\right]\indc{\dgcouple > \ell}\\
    &= \left[\ncouple\left(\mintwo{\dgcouple}{\ell}\right)-h\right] \indc{\dgcouple > \ell} + \left[\ncouple\left(\mintwo{\dgcouple}{\ell}\right)-h\right] \indc{\ell \leq \dgcouple}\\
    &= \left[\ncouple\left( \mintwo{\dgcouple}{\ell}\right)-h\right].
\end{align*}

\paragraph{Proof of \eqref{eq:couplingfirst}.}
Combining these facts allows us to prove our desired result. Applying Doob's Optional Stopping Theorem along with our previous deductions, we obtain 
\begin{align*}
     \ex{\left[N\left(\tau + \ell\right) - h\right]\indc{\dgen > \ell}\middle | \filt{\tau}} \leq \ex{\left[\ncouple\left(\ell\right) - h\right]\indc{\dgcouple > \ell}}
     = \ex{\ncouple\left(\mintwo{\dgcouple}{\ell}\right) - h}. 
\end{align*}
Substituting in our martingale $V_1(\cdot)$ and applying the Optional Stopping Theorem,
\begin{align*}
    \ex{\ncouple\left(\mintwo{\dgcouple}{\ell}\right) - h}
     = \ex{V_1\left (\mintwo{\dgcouple}{\ell}\right)} + \mu j \ex{\mintwo{\dgcouple}{\ell}} = \ex{V_1\left(0\right)} + \mu j \ex{\mintwo{\dgcouple}{\ell}}
\end{align*}
Exchanging definitions again, we have
\begin{align*}
     \ex{V_1\left(0\right)} + \mu j \ex{\mintwo{\dgcouple}{\ell}} = \left[\ncouple(0) - h\right] + \mu j \ex{\mintwo{\dgcouple}{\ell}}
    \leq \left[\ncouple(0) - h\right] + \mu j \ell
    = \left[ N(\tau) - h \right] + \mu j \ell.
\end{align*}
\newcommand{\dmin}{\mintwo{\dgcouple}{\ell}}
\paragraph{Proof of \eqref{eq:couplingsecond}.}
To do the same for the squared martingale $V_2(t)$, we must first note, via some algebra, that
\begin{equation*}
    \left(\ncouple(t) - h\right)^2 = V_2(t) + \left(\ncouple(t) - h\right) \mu j t - \mu j^2 t^2 + \mu \left(2R - j\right)t.
\end{equation*}
Now, applying the same deductions we made previously, 
\begin{align*}
    \ex{\left[N\left(\tau + \ell\right) - h\right]\indc{\dgen > \ell}\middle | \filt{\tau}} 
    &\leq \ex{\left[\ncouple\left(\ell\right) - h\right]^2\indc{\dgcouple > \ell}}\\
    &= \ex{\left(\ncouple\left(\mintwo{\dgcouple}{\ell}\right) - h\right)^2}\\
    &= \ex{V_2\left(\mintwo{\dgcouple}{\ell}\right)} + \ex{\left(\ncouple(\dmin) - h\right) \mu j \dmin}  \\
    &\;\;\; - \mu j^2 \left(\dmin\right)^2 + \mu \left(2R - j\right)\ex{\dmin}\\
    &\leq \ex{V_2\left(\mintwo{\dgcouple}{\ell}\right)} + \ex{\left(\ncouple(\dmin) - h\right) } \mu j \ell + \mu \left(2R\right)\ell\\
    &\leq \ex{V_2\left(0\right)} + \left[\ncouple(0)- h + \mu j \ell\right] \mu j \ell + \mu \left(2R\right)\ell\\
    &= \left[\ncouple(0) - h\right]^2 + \left[\ncouple(0)- h\right]\mu j \ell + \left(\mu j \ell\right)^2 + \mu \left(2R\right)\ell\\
    &= \left[ \ncouple(0) - h + \mu j \ell\right]^2 - \left[\ncouple(0) - h\right]\mu j \ell + \mu 2R \ell\\
    &\leq \left[\ncouple(0) - h + \mu j \ell\right]^2 + 2\mu R \ell\\
    &= \left[N(\tau) - h + \mu j \ell\right]^2 + 2\mu R \ell. \tag*{\Halmos}
\end{align*}

%% file: claims/bpclaim.tex
\begin{restatable}[Busy Period Integral Bound]{claim}{bpclaim}\label{clm:bpclaim}
    Suppose that, at time $\tau$, we can guarantee that $N(\tau)\geq Z(\tau) \geq R + j$. Let $\eta_i \triangleq \minpar{t > 0: N(t) \leq R + i}$, for $i \in \left\{j,j+1, \dots, \left[N(\tau) - R\right]\right\}.$ Then,
    \begin{equation*}
        \ex{\int_{\tau}^{\eta_j}{[N(t) - R] \dd t} \middle| \filt{\tau}} \leq \left( N(\tau) - (R+j) \right)\left[\frac{3}{2\mu j} + \frac{1}{\mu} + \frac{R}{\mu j^2} \right] + \frac{\left(N(\tau) - (R+j)\right)^2}{2 \mu j} \triangleq \ibp{[N(\tau) - R]}{j}.
    \end{equation*}
\end{restatable}
\begin{proof}
    We prove this claim via an appeal to conventional M/M/1 busy period analysis. In particular, we first note that
    \begin{equation*}
        \int_{\tau}^{\eta_j}{\left[N(t) - R\right] \dd t} = \sum_{i=j+1}^{N(\tau) - R}{\int_{\eta_{i}}^{\eta_{i-1}}{\left[N(t) - R\right] \dd t}},
    \end{equation*}
    meaning we need only bound the integrals between the $\eta_i$'s. To bound that process, we define a coupled process $\ncouple(t)$ and bound the integrals over that process.

    To do so, note that, until time $\eta_j$, the number of busy servers $Z(t) \geq R + j$. By Claim~\ref{clm:coupling}, we can define, for each index $i$, the $i$-th coupled process $\ncouple_i(t)$ as
    \begin{equation*}
        \ncouple_i(t) = N(\eta_{i+1}) + A(\eta_{i+1}, t) - \depOp{R+j}{[\eta_{i+1},t)},
    \end{equation*}
    and have $\ncouple(t) \geq N(t)$ on the interval $[\eta_{i+1}, \eta_i]$. Furthermore, we can extend our integral of interest from the interval $[\eta_{i+1}, \eta_i)$ to the interval $[\eta_{i+1}, \ceti)$, where $\ceti \triangleq \minpar{t > 0: N(t) \leq R + i}$. Now, we note that
    \begin{equation*}
        \ex{\int_{\eta_{i+1}}^{\ceti}{\left[\ncouple_i(t) - R\right] \dd t}\middle | \filt{\tau}} = \ex{\int_{\eta_{i+1}}^{\ceti}{\left[\ncouple_i(t) - (R+i)\right] \dd t}\middle | \filt{\tau}} + i \ex{\eta_{i+1} - \ceti\middle | \filt{\tau}}.
    \end{equation*}
    The first term on the right is simply the expected time integral of the number of jobs in an M/M/1 queue over a busy period, with arrival rate $k \lambda$ and departure rate $\mu(R+j)$. The second term is simply the quantity $i$ multipled by the expected length of that M/M/1 busy period. Let $\rho_j = \frac{k\lambda}{\mu(R+j)}$. Then, from standard results on the M/M/1 busy period,
    \begin{align*}
        \ex{\int_{\eta_{i+1}}^{\ceti}{\left[\ncouple_i(t) - (R+i)\right] \dd t}\middle | \filt{\tau}} = \frac{1}{\mu j}\left[ \frac{1}{1-\rho_j} \right] = \frac{1}{\mu j}\left[ \frac{R}{j} + 1 \right] = \frac{1}{\mu j} + \frac{R}{\mu j^2}.
    \end{align*}
    Summing over all values of $i$, we obtain
    \begin{align*}
        &\ex{\int_{\tau}^{\eta_j}{\left[N(t) - R \right] \dd t}\middle | \filt{\tau}} \\
        &\;\;\leq  \sum_{i={j+1}}^{N(\tau) - R}{\ex{\int_{\eta_{i}}^{\Tilde{\eta}_{i-1}}{\left[\ncouple_i(t) - R\right] \dd t}\middle | \filt{\tau}}}\\
        &\;\;=\sum_{i={j+1}}^{N(\tau) - R}{\left[ \frac{1}{\mu j} + \frac{R}{\mu j^2}\right] + i \frac{1}{\mu j}}\\
        &\;\;= \left( N(\tau) - (R + j)\right) \left[ \frac{1}{\mu j} + \frac{R}{\mu j^2}\right] + \left( N(\tau) - (R + j)\right) \frac{1}{\mu} + \frac{1}{\mu j} \left[ \frac{\left(N(\tau) - (R+j)\right) \left(N(\tau) - (R+j) + 1\right)}{2}\right] \\
        &= \left( N(\tau) - (R+j) \right)\left[\frac{3}{2\mu j} + \frac{1}{\mu} + \frac{R}{\mu j^2} \right] + \frac{\left(N(\tau) - (R+j)\right)^2}{2 \mu j}.\tag*{\Halmos}
    \end{align*}
\end{proof}

%% file: claims/waitbusyv2.tex
\newcommand{\remh}{\rem{R + \numserv{h}}\left(\tau\right)}
\newcommand{\remi}{\rem{R + \numserv{i+1}}\left(\eta_{i+1}\right)}
\newcommand{\ntakpos}{\left[\nta - k\right]^+}
\newcommand{\nprime}[1]{N_{\textup{adj}}\left(#1\right)}
\newcommand{\yimp}{Y_{\textup{imp}}}


\subsubsection{\texorpdfstring{Proof of \eqref{eq:waitbusy1}.}{Proof of Eq.~\ref{eq:waitbusy1}.}}
We prove the two parts of Claim~\ref{clm:waitbusy2} separately; we first show \eqref{eq:waitbusy1} by applying coupling, martingales, and busy period analysis.
First, note that, if the down-crossing at $\tau + \dgen$ does not occur by time $\rem{R+\numserv{h}}(\tau)$, then the system must have at least $(R + \numserv{h})$ servers at its disposal afterwards.
(Note that the $\numserv{\cdot}$ function here is just to account for the case where you have more jobs than servers.)
Accordingly, we split our analysis into two parts.

\paragraph{First portion.} 
For the first portion, since the number of busy servers $Z(t)\geq R$, by coupling our system to a critically-loaded M/M/1 using Claim~\ref{clm:dcbnd}, we have
\begin{equation*}
    \ex{\int_{\tau}^{\tau + \mintwo{\dgen}{\rem{R+\numserv{h}}(\tau)}}{\left[N(t) - (R+h-1)\right] \dd t}\middle | \filt{\tau}} \leq \rem{R+\numserv{h}}(\tau).
\end{equation*}

\paragraph{Second portion.} For the second portion, since, at that point the number of busy servers $Z(t) \geq R + \numserv{h}$, we can apply a stronger bound. 
In particular, at time $\left(\tau + \remh \right)$, we can couple to an accordingly-stronger M/M/1 with the same number of jobs.
From basic busy period analysis and Claim~\ref{clm:coupling}, this tells us that, letting the adjusted number of jobs $\nprime{t} \triangleq N(\tau + t) - (R+h-1)$ and the important remaining setup time $\yimp \triangleq \remh$ as a shorthand,
\begin{equation*}
        \ex{\indc{\yimp < \dgen} \int_{\yimp}^{\dgen}{\left[\nprime{t}\right] \dd t}\middle | \filt{\tau + \yimp}} \leq \indc{\yimp< \dgen} g\left(\nprime{\yimp}^2, \nprime{\yimp}, \numserv{h}\right),
\end{equation*}
since the function $g(x^2,x,z)$ describes the integral of the number of jobs over a busy period started by $x$ jobs in an M/M/1 with arrival rate $k\lambda$ and departure rate $\mu\left(R + z\right)$.
Note that $\indc{\yimp<\dgen}\left[N(\tau + \yimp) - (R+h - 1)\right] = \left[N(\tau + \mintwo{\dgen}{\yimp}) - (R+h - 1)\right]$, since $N(\tau + \dgen) = R + h -1$ by definition.
Coupling our system to the critically-loaded M/M/1 $\ncouple(t)$ system as before, then taking expectations and applying Doob's Optional Stopping Theorem, we obtain, as desired,
$$\ex{\indc{\remh < \dgen} \int_{\remh}^{\dgen}{\nprime{t} \dd t}\middle | \filt{\tau + \remh}} \leq g\left(1 + 2\mu R \ex{\mintwo{\remh}{\dgen}}, 1,\numserv{h} \right).$$
Combining these two terms, we obtain \eqref{eq:waitbusy1}.

\subsubsection{\texorpdfstring{Proof of \eqref{eq:waitbusy2}.}{Proof of Eq.~\ref{eq:waitbusy2}.}}
We now apply \eqref{eq:waitbusy1} to prove \eqref{eq:waitbusy2}.

\paragraph{Decomposition in terms of $\eta_i$'s.} We first fix the filtration/state at the accumulation time $T_A$, then define the hitting times
$\eta_i \triangleq \minpar{t > T_A: N(t) \leq R + i}$ for a set number of jobs $i \in \left[\mb, N(T_A)\right]$; we accordingly omit the filtration at time $T_A$ in our expectations.
Note that the down-crossing $\dnb{1} = \eta_{\mb}$, by this definition.
From there, it's clear that
$$\int_{T_A}^{\dnb{1}}{[N(t) - R] \dd t} = \sum_{i=\mb}^{\nta - R - 1}{\int_{\eta_{i+1}}^{\eta_i}{\left[N(t) - R\right] \dd t}}.$$
For each of these terms, we can separate $[N(t) - R] =[N(t) - (R+i)] + i$ and apply \eqref{eq:waitbusy1} to find
\begin{align}
    \ex{\int_{\eta_{i+1}}{\eta_i}{[N(t) - R]}} &\leq \ex{\rem{R+ \numserv{h}}(\eta_{i+1})}  + i\ex{\mintwo{\remi}{\eta_{i} - \eta_{i+1}}}\label{eq:wbh1}\\
    &\;\;\;+ g\left(\left[1 + 2\mu R \ex{\mintwo{\remi}{(\eta_{i} - \eta_{i+1})}}\right], 1, \mu \numserv{i+1}\right)\label{eq:wbh2}\\
    &\;\;\;+ i\cdot\frac{1}{\mu \numserv{i+1}}.\label{eq:wbh3}
\end{align}
We analyze each of these terms separately.

\paragraph{Bound on \eqref{eq:wbh1}, the remaining setup time portion.} From here, it suffices to note that the sum 
\begin{equation}\label{eq:waitbusySumInt}
    \sum_{i=1}^{N(T_A) - R}{i \mintwo{\remh}{\eta_{i} - \eta_{i+1}}} + \sum_{i=1}^{N(T_A) - R}{\ex{\remi}} \leq \beta \left[N(T_A) - R\right];
\end{equation} actually, the statement is true without expectations.
To see this, we make an interchange of summation argument.
First note that, if the $(R+i)$-th server becomes busy, then, by the monotonicity of server states, all servers of index smaller than $(R+i)$ must also be busy; in other words,  the remaining setup time $\remi = 0$ for all $i < s$.
To use this, we let $s$ be the largest index for which $\rem{R+\numserv{s}}(\eta_s) < (\eta_{s-1} - \eta_{s})$.
Note also that $\eta_i$ is the first time after $T_A$ that the number of jobs $N(t) \leq R+i$ (and so must have been continuously decreasing from time $T_A$); it follows that the remaining setup time $\rem{R+\numserv{i}}(\eta_{i}) \leq \left[\beta - (\eta_i - T_A)\right]^+$.
Breaking things down further,
\begin{equation*}
    \sum_{i=s}^{N(T_A) - R}{i \mintwo{\remi}{(\eta_{i} - \eta_{i+1})}} = \sum_{i=s}^{N(T_A) - R}{i \left(\eta_{i} - \eta_{i+1}\right)} + s \rem{R+\numserv{s}}(\eta_s).
\end{equation*}
From here, by an interchange of summation argument,
$$\sum_{i=s}^{N(T_A) - R}{i \left(\eta_{i} - \eta_{i+1}\right)} = \sum_{i=s}^{N(T_A) - R}{\sum_{j=1}^{i}{\left(\eta_{i} - \eta_{i+1}\right)}} = \sum_{j=s}^{N(T_A) - R}{\sum_{i=j}^{N(T_A) - R} {\left(\eta_{i} - \eta_{i+1}\right)}} = \sum_{j=1}^{N(T_A) - R}{\eta_s - T_A},$$
where, by definition, the (relative to $T_A)$ hitting time $(\eta_s - T_A) \leq \beta - \rem{R+\numserv{s}}(\eta_s)$; using this,\eqref{eq:waitbusySumInt} follows.

\paragraph{Bound on \eqref{eq:wbh2}, the busy period integral portion.}
Applying similar reason to the sum of the first terms in $g$, and using the independence of $g$ in its first and second arguments (i.e. that $g(x,y,z) = f_1(x,z) + f_2(y,z)$ for two functions $f_1$ and $f_2$ linear in their first argument),
$$\sum_{i=\mb}^{N(T_A)-R}{\eqref{eq:wbh2}} \leq g\left(2 \mu R \beta, 0, \mb\right) + g\left(\ntakpos, \ntakpos, k(1-\rho)\right)  + \sum_{i=\mb}^{\mintwo{\nta -R - 1}{k(1-\rho)}}{g(1,1,i+1)}.$$

This last term above can be bounded by replacing it with an integral, which gives
\begin{align*}
    \sum_{i=\mb}^{\mintwo{\nta -R}{k(1-\rho)}}{g(1,1,i)} &\leq \int_{\mb}^{\mintwo{\nta-R}{k(1-\rho)}}{\frac{2}{\mu i} + \frac{R}{\mu i^2} \dd i}\\
    &\;\;\; \leq \frac{2}{\mu}\ln\left(\frac{\mintwo{\nta-R}{k(1-\rho)}}{\mb}\right) + \frac{R}{\mu} \left[ \frac{1}{\mb}\right]\\ 
    &\;\;\;\leq \frac{2}{\mu}\ln\left(\frac{\nta - R}{\mb}\right) + \frac{R}{\mu \mb}.
\end{align*}

\paragraph{Bound on \eqref{eq:wbh3}, the busy period length portion. } Using the definition of the function $g$, we also find that
$\sum_{i=\mb}^{\nta-R}{\frac{i}{\mu \numserv{i}}} \leq g\left(\left(\left[\nta-k\right]^+\right)^2 + \ntakpos, 0, k(1-\rho) \right) + \frac{1}{\mu}\mintwo{k(1-\rho), N(T_A) - R}.$

\paragraph{Combining the terms to bound \eqref{eq:waitbusy2}, the integral from $T_A$ to $\dnb{1}$.} Combining terms, noting that $N(T_A) - k \geq \nta - R \geq 0$, and applying Jensen's inequality to the $\ln{(\cdot)}$ term, we find that
\begin{align*}
    \eqref{eq:waitbusy2} &\leq  
    \beta \ex{N(T_A) - R} 
    + g\left(2 \mu R \beta, 0, \mb\right)
    + \ex{g\left(\ntakpos, \ntakpos, k(1-\rho)\right)}\\
    &\;\;\;+\ex{\frac{2}{\mu}\ln\left(\frac{\nta - R}{\mb}\right)} + \frac{1}{\mu}\frac{R}{\mb}\\
    &\;\;\;+ \ex{g\left(\left(\left[\nta-k\right]^+\right)^2 + \ntakpos, 0, k(1-\rho) \right)} + \frac{1}{\mu}\ex{\mintwo{k(1-\rho)}{ N(T_A) - R}}\\
&\leq \left[\beta + \frac{1}{\mu}\right] \left[\ex{\nta - R} + \frac{R}{\mb}\right] + \frac{2}{\mu}\ln\left(\frac{\ex{\nta - R}}{\mb}\right)\\
    &\;\;\;+  g\left(\ex{\left[\nta - R\right]^2} + 2\ex{\nta - R}, \ex{\nta - R}, k(1-\rho)\right). \tag*{\Halmos}
\end{align*}

%% file: claims/berry_ess.tex
\begin{restatable}[Berry-Esseen bound for the Skellam distribution]{claim}{berryess}\label{clm:berryess}
Given two independent random variables $Y_1\sim \Poisson(\mu_1)$ and $Y_2 \sim \Poisson(\mu_2)$, as well as a constant $C$ with $\mu_1 > \mu_2 + C$, one has
    \begin{equation*}
        \pr{Y_1 - Y_2 \geq C} \geq 1 - \Phi\left(-\left[\frac{\mu_1 - \mu_2 - C}{\mu_1 + \mu_2} \right]\right) - \frac{1}{3\sqrt{\mu_1 + \mu_2}}.
    \end{equation*}
\end{restatable}

This follows directly from the Poisson Berry-Esseen bound of \cite{normalBnd}, applied twice; first approximating $Y_1$ then approximating $Y_2$.
\hfill \Halmos

%% file: claims/lbound.tex
\subsection{\torp{Proof of \eqref{eq:lbound}: Lower Bound on $\ex{L}$, Expected Value of First Long Epoch Index.}{Lower Bound on Expected Value of First Long Epoch Index.}}\label{sec:lbound}

We prove this result by first showing that
\begin{equation}\label{eq:lbcons}
    \pr{ L > j \middle | L \geq j } \geq \left(1 - \frac{j}{R}\right) \left(1 - \frac{b_1}{\sqrt{\mu \beta R}}\right), 
\end{equation}
where $b_1 = \frac{2}{\sqrt{\pi}}$.
Next, we show that this implies that for any $\delta \in (0,1)$ and any $j < \delta R$,
\begin{equation}\label{eq:lbtail}
     \pr{L > j} \geq \left( 1 - \frac{b_1}{\sqrt{\mu \beta R}}\right)^{j+1} e^{-\frac{j(j+1)}{2 R} \frac{1}{1 - \delta}}.
\end{equation}
From here, we use the sum of tails formula $\ex{L} = \sum_{j=0}^{\infty}{\pr{L>j}}$ to show
\begin{equation*}\label{eq:lbex}
    \ex{L} \geq \left(1 - \frac{b_1}{\shs}\right) \left(\left[\sqrt{\frac{\pi}{2}}(1-\delta) - \frac{1.15(1-\delta)}{\sqrt{\mu \beta}}\right] \sqrt{R}  - \frac{1}{2} - \frac{2(1-\delta)}{\delta} e^{- R \frac{\delta^2}{1-\delta}}\right).
\end{equation*}
Choosing $\delta = \frac{2}{\sqrt{R}}$ then noting that $\mu \beta \geq 100$ and $R \geq 100$ gives the result.

\subsubsection{\torp{Proof of \eqref{eq:lbcons}: Lower Bound on Probability that Current Epoch is Short.}{Lower Bound on Probability that Current Epoch is Short.}}\label{sec:lbcons}
Recall that an epoch $j$ is \emph{long} if $\tau_{j+1} - \tau_j > \setuptime$, that $L$ is the index of the first long epoch, and that, if $L \geq j$, then we learn that $L \geq j$ precisely at time $\tau_j$, i.e. when epoch $j$ begins. Moreover, since the system is Markovian, the behavior of the system from $\tau_j$ onwards is completely independent of what happened previously. Thus,
\begin{equation*}
    \pr{ L > j \middle | L \geq j} = \pr{ L > j\middle | \filt{\tau_j}, L \geq j} = \pr{ \tau_{j+1} - \tau_j \leq \beta \middle | \filt{\tau_j}, L \geq j} = \pr{\tau_{j+1} - \tau_j \leq \beta}.
\end{equation*}
From here, we note that the random time $\tau_{j+1} - \tau_j$ is a stopping time; a hitting time, to be exact. Moreover, since the number of servers $Z(t)$ can not increase before time $\tau_j + \beta$ and can not decrease until $\tau_{j+1}$, we have that the coupled process $\ncouple(t)$ defined as
\begin{equation*}
    \ncouple(t-\tau_j) \triangleq 1 + A(\tau_j, t) - \depOp{R-j}{(\tau_j, t)}
\end{equation*}
is in correspondence with $N(t)$; in particular,  
\begin{equation*}
    N(t) = \ncouple(t -\tau_j) + R - j - 1
\end{equation*}
for any time $t \in \left[ \tau_j, \mintwo{\tau_j + \setuptime}{\tau_{j+1}}\right]$. If we define the coupled hitting time $\gamma_c \triangleq \minpar{t>0: \ncouple(t) \leq 0}$, then we also have that the hitting time $\gamma_c = \tau_{j+1} - \tau_j$, whenever the event $\left\{\tau_{j+1} - \tau_j \leq \setuptime \right \}$ occurs. From here, we can apply Claim~\ref{clm:contdown} to find that
\begin{equation*}
    \pr{\gamma_c \leq \beta} \geq  \left(1 - \frac{j}{R}\right) \left(1 - \frac{b_1}{\sqrt{\mu \beta R}}\right). \tag*{\Halmos}
\end{equation*}
\subsubsection{Proof of \texorpdfstring{\eqref{eq:lbtail}}{eq:lbtail}.}
Having shown the above bound on the conditional extension of the tail, we note that, for $j \leq \delta R$,
\chmades{
\begin{align*}
    \pr{L \geq j+1} = \prod_{i=0}^{j}{\pr{L \geq i + 1\middle | L \geq i}}
    \geq \prod_{i=0}^{j}{ \left(1 - \frac{i}{R}\right) \left(1 - \frac{b_1}{\sqrt{\mu \beta R}}\right)}.
\end{align*}
Manipulating the product, we complete the proof by noting
\begin{align*}
    \prod_{i=0}^{j}{e^{-\frac{i}{R-i}}} =  e^{-\sum_{i=0}^{j}{\frac{i}{R-i}}}
    = e^{-\sum_{i=0}^{j}{\frac{i}{R}}\frac{R}{R-j}}
     \geq e^{-\frac{j(j+1)}{2 R} \frac{1}{1 - \delta}}.\tag*{\Halmos}
\end{align*}
}


\subsubsection{\torp{Proof of \eqref{eq:lbex}: Final Bound on $\ex{L}$ using Gaussian Integral.}{Final Bound on E[L] using Gaussian Integral.}}
We now complete the proof. Let $a \triangleq \frac{1}{2 R}\frac{1}{1-\delta}$ and $\psi \triangleq -\ln\left( 1 - \frac{b_1}{\sqrt{\mu \beta R}}\right)$ as a shorthand. Then we can rewrite \eqref{eq:lbtail} as 
\begin{equation*}
    \pr{ L \geq j+1} \geq e^{- a j^2 - \left(\psi + a\right) j - \psi}.
\end{equation*}
Now, using the sum-of-tails formula for expectations, we find that
\chmades{
\begin{align*}
    \ex{L} = \sum_{j=0}^{R-1}{\pr{L \geq j+1 }}
    \geq \sum_{j=0}^{\delta R - 1}{\pr{L \geq j+1 }}
    \geq \sum_{j=0}^{\delta R-1}{e^{- a j^2 - \left(\psi + a\right) j - \psi}}
    \geq \int_{0}^{\delta R}{e^{- a j^2 - \left(\psi + a\right) j - \psi} \dd j}.
\end{align*}
Manipulating the argument of the exponential,
\begin{align*}
    \int_{0}^{\delta R}{e^{- a j^2 - \left(\psi + a\right) j - \psi} \dd j} = \int_{0}^{\delta R}{e^{- a \left(j + \frac{1}{2}\left(\frac{\psi}{a} + 1\right)\right)^2 + \frac{a}{4} \left(\frac{\psi}{a}   + 1\right)^2 - \psi} \dd j} 
    = e^{\frac{a}{4} \left(\frac{\psi}{a}   + 1\right)^2 - \psi} \int_{0}^{\delta R}{e^{- a \left(j + \frac{1}{2}\left(\frac{\psi}{a} + 1\right)\right)^2} \dd j}.
\end{align*}
}
Evaluating the integral further, we find that
\begin{align*}
    \int_0^{\delta R}{e^{- a \left(j + \frac{1}{2}\left(\frac{\psi}{a} + 1\right)\right)^2} \dd j} = \int_{\frac{1}{2}\left(\frac{\psi}{a} + 1\right)}^{\delta R + \frac{1}{2}\left(\frac{\psi}{a} + 1\right)}{e^{- a j^2} \dd j}
    = \int_{0}^{\infty}{e^{- a j^2} \dd j} 
    -  \int_0^{\frac{1}{2}\left(\frac{\psi}{a} + 1\right)}{e^{- a j^2} \dd j} 
    -  \int_{\delta R + \frac{1}{2}\left(\frac{\psi}{a} + 1\right)}^{\infty}{e^{- a j^2} \dd j}.
\end{align*}
We now bound each of these integrals in turn. First, we know classically that $$\int_0^{\infty}{e^{-aj^2}\dd j} = \frac{1}{2}\sqrt{\frac{\pi}{a}} = \sqrt{\frac{\pi}{2}}\cdot\sqrt{1-\delta}\sqrt{R}\geq \sqrt{\frac{\pi}{2}}\cdot(1-\delta)\sqrt{R}$$ Next, we note that, since the integrand is $\leq 1$,
\begin{align*}
    \int_0^{\frac{1}{2}\left(\frac{\psi}{a} + 1\right)}{e^{- a j^2} \dd j} \leq \frac{1}{2}\left(\frac{\psi}{a} + 1\right)
    &=\frac{1}{2}\left(2R(1-\delta)\ln\left(\frac{1}{ 1 - \frac{b_1}{\sqrt{\mu \beta R}}}\right)\right) + \frac{1}{2}\\
    &\leq  R(1-\delta) \frac{b_1}{\shs}\frac{1}{1 - \frac{b_1}{\shs}} + \frac{1}{2}\\
    &\leq \left(\frac{1}{\sqrt{\setuptime}}\right)(1-\delta)\cdot\frac{100\cdot b_1}{100 - b_1}\sqrt{R} + \frac{1}{2}\\
    &\leq \frac{1.15(1-\delta)}{\sqrt{\mu \beta}} \sqrt{R} + \frac{1}{2}.
\end{align*}
Finally, we have that,
\begin{align*}
    \int_{\delta R + \frac{1}{2}\left(\frac{\psi}{a} + 1\right)}^{\infty}{e^{- a j^2} \dd j} \leq  \int_{\delta R}^{\infty}{e^{- a j^2} \dd j} \leq \int_{\delta R }^{\infty}{e^{-a \delta R j} \dd j} = \frac{1}{a \delta R}e^{- a \delta^2 R^2} = \frac{2(1-\delta)}{\delta} e^{- R \frac{\delta^2}{1-\delta}}.
\end{align*}

To complete the proof, we note that $e^{-\psi} = \left(1 - \frac{b_1}{\shs}\right)$, thus
\begin{align*}
    \ex{L} &\geq e^{\frac{a}{4} \left(\frac{\psi}{a}   + 1\right)^2 - \psi} \int_{0}^{\delta R}{e^{- a \left(j + \frac{1}{2}\left(\frac{\psi}{a} + 1\right)\right)^2} \dd j}\\
    &\geq e^{- \psi} \int_{0}^{\delta R}{e^{- a \left(j + \frac{1}{2}\left(\frac{\psi}{a} + 1\right)\right)^2} \dd j}\\
    &\geq \left(1 - \frac{b_1}{\shs}\right) \left[\sqrt{\frac{\pi}{2}}(1-\delta) - \frac{1.15(1-\delta)}{\sqrt{\mu \beta}}\right] \sqrt{R}  - \frac{1}{2} - \frac{2(1-\delta)}{\delta} e^{- R \frac{\delta^2}{1-\delta}}\\
    &= \left(1 - \frac{b_1}{\shs}\right) \left[\sqrt{\frac{\pi}{2}}(1-\delta) - \frac{1.15(1-\delta)}{\sqrt{\mu \beta}}\right] \sqrt{R}  - \frac{1}{2} - \frac{2(1-\delta)}{\delta} e^{- R \frac{\delta^2}{1-\delta}}\\
    &= \left(1 - \frac{b_1}{\shs}\right) \left(\left[\sqrt{\frac{\pi}{2}}(1-\delta) - \frac{1.15(1-\delta)}{\sqrt{\mu \beta}}\right] \sqrt{R}  - \frac{1}{2} - \frac{2(1-\delta)}{\delta} e^{- R \frac{\delta^2}{1-\delta}}\right).
\end{align*}
From here, we could choose $\delta$ to maximize our lower bound further, 
but simply using $\delta = \frac{2}{\sqrt{R}}$ yields:
\begin{align*}
    \ex{L} &\geq \left(1 - \frac{b_1}{\shs}\right)\left[\left(1 - \frac{2}{\sqrt{R}}\right) \left(\sqrt{\frac{\pi}{2}} - \frac{1.15}{\sqrt{\mu \beta}} - e^{-4}\right) - \frac{1}{2 \sqrt{R}}\right] \sqrt{R}
    \geq \frac{2}{3}\sqrt{\frac{\pi}{2}} \sqrt{R}.\tag*{\Halmos}
\end{align*}

%% file: claims/taexpect.tex
\newcommand{\ma}{\threshone}
\newcommand{\uij}{u_i^{(j)}}
\newcommand{\dij}{d_i^{(j)}}

\subsubsection{Proof.}
The beginning of the proof will be the same for both of these inequalities. Using the up-crossing and down-crossing decomposition of Section~\ref{sec:UBintTA}, we know that time $T_A$ occurs either during a \emph{rise} or during a \emph{fall}. Since the number of jobs $N(t) \leq R + \ma$ during a rise, 
\begin{equation*}
    \ntar\indc{\textrm{$T_A$ during a rise}} \leq \ma \indc{\textrm{$T_A$ during a rise}}.
\end{equation*}

If $T_A$ occurs during a fall, we need a more nuanced bound. Writing out the event $\brc{\textrm{$T_A$ during a fall}}$ in terms of disjoint events, we find
\begin{equation*}
    \brc{\textrm{$T_A$ during a fall}} = \bigcup_{j=0}^{R}{\bigcup_{i=1}^{\infty}{\brc{\uij \leq T_A < \dij}}},
\end{equation*}
so that, for $c\in\brc{1,2}$,
\begin{align*}
    \ex{\ntar^c \indc{\textrm{$T_A$ during a fall}}} &= \sum_{j=0}^{R-1}{\sum_{i=1}^{\infty}{\ex{\ntar^c \indc{\uij \leq T_A < \dij}}}}\\
    &= \sum_{j=0}^{R-1}{\sum_{i=1}^{\infty}{\ex{\ntar^c \indc{T_A < \dij}\middle | \filt{\uij}, n_u^{(j)}\geq i}\pr{n_u^{j} \geq i} }}.
\end{align*}

To bound this conditional expectation, we apply Claim~\ref{clm:couplingEx}. Notice that $N(\uij) - R = \threshone$, the $(R+1)$-th server starts up at time $T_A = \uij + \rem{R+1}(\uij)$ if $T_A < \dij$, the time $\dij$ is a hitting time, and that $Z(t)\geq R-j$ until time $\tau_{j+1} \geq \dij$. It follows that
\begin{equation*}
    \ex{\ntar \indc{T_A < \dij}\middle | \filt{\uij}, n_u^{(j)}\geq i} \leq \threshone + \mu j \rem{R+1}(\uij) \leq \threshone + \mu j \beta,
\end{equation*}
and that
\begin{align*}
    \ex{\ntar^2 \indc{T_A < \dij }\middle | \filt{\uij}, n_u^{(j)}\geq i}
    &\leq \left(\threshone + \mu j \rem{R+1}\left(\uij\right) \right)^2+ \mu 2 R \rem{R+1}\left(\uij\right)\\
    &\leq \left(\threshone + \mu j \beta \right)^2+ \mu 2 R \beta\\
    &= C_3^2 \mu \beta R + 2 \threshone \mu \beta j + (\mu \beta)^2 j^2 + 2 \mu R \beta.
\end{align*}

It now suffices to bound $\sum_{j}{\sum_{i}{j^c \pr{n_u^{(j)} \geq i}}}$, where $c\in\brc{0,1,2}$.
We do this via the same method used in Section~\ref{sec:UBintTA}:
\begin{align*}
    \sum_{j=1}^{R}{\sum_{i=1}^{\infty}{j^c\pr{n_u^{(j)} \geq i}}} &\leq \sum_{j=1}^{R}{\sum_{i=1}^{\infty}{j^c\pr{n_e \geq j}\prise{j} (1- p_2)^{i-1}}}\\
    & = \frac{1}{p_2} \sum_{j=1}^{R}{j^c \prise{j} \pr{n_e \geq j}}\\
    & \leq  \frac{1}{C_4 p_2} \sum_{j=1}^{R}{j^c C_4\prise{j} \prod_{\ell=0}^{j-1}{\left(1 - C_4 \prise{\ell}\right)}}.
\end{align*}
This is simply the expectation of a time-varying geometric random variable $G$, with $\pr{G = j\middle | G  \geq j} = C_4\prise{j}$. It follows that if one lower-bounds $\prise{j}$, then an upper bound on the desired expectation is obtained. Applying Claim~\ref{clm:prisebnd}, we note that we are essentially bounding $G$ using $Y\sim \text{Geometric}\left(\frac{0.99 C_4 A_5}{\sqrt{R}}\right)$ and saying $Y+A_5 \sqrt{R}$ stochastically-dominates $G$. It follows that
\begin{equation*}
    \ex{G} \leq A_5\sqrt{R} + \frac{1}{0.99 C_4 A_5}\sqrt{R}
\end{equation*}
and that, for any $b$,
 \begin{align*}
     \ex{(G+b)^2} &\leq \ex{(Y+ A_5 \sqrt{R}+b)^2} = \ex{Y^2} + 2 (A_5\sqrt{R} +b) \ex{Y} + (A_5 \sqrt{R} + b)^2\\
     &= 2\ex{Y}^2 - \ex{Y} + 2 \left(A_5\sqrt{R} +b\right) \ex{Y} + \left(A_5 \sqrt{R} + b\right)^2\\
     &\leq \left(\ex{Y}  + A_5 \sqrt{R} + b\right)^2 + \ex{Y}^2.
 \end{align*}
 Defining $B_5 \triangleq \frac{C_3}{C_4 p_2}$, $B_6 \triangleq \frac{1}{C_4 p_2} \left(\frac{1}{0.99 C_4 A_5} + A_5\right)$, and $B_7 \triangleq \frac{1}{2 C_4 p_2}\left[\frac{1}{(0.99 C_4 A_5)^2} + 2\right]$, it follows that
{\small \begin{align*}
      \ex{\ntar^2 \indc{\textrm{$T_A$ during a fall}}} & \leq \frac{1}{C_4 p_2} \sum_{j=1}^{R}{C_4\prise{j} \left[\left(\threshone + \mu j \beta \right)^2+ \mu 2 R \beta\right]}\prod_{\ell=0}^{j-1}{\left(1 - C_4 \prise{\ell}\right)}\\
      &= \frac{1}{C_4 p_2}\ex{\left(\threshone + \mu G \beta \right)^2+ \mu 2 R \beta}\\
      &\leq \frac{1}{C_4 p_2}\ex{\left(\threshone + \mu \beta Y + \mu \beta A_5 \sqrt{R} \right)^2+ \mu 2 R \beta}\\
      &=\frac{1}{C_4 p_2}\left[\left(\threshone + \mu \beta \frac{1}{0.99 C_4 A_5}\sqrt{R} + \mu \beta A_5 \sqrt{R} \right)^2 + \mu \beta \frac{1}{(0.99 C_4 A_5)^2}R + 2\mu \beta R\right]\\
      &\leq\frac{1}{C_4^2 p_2^2}\left[\left(\threshone + \mu \beta \frac{1}{0.99 C_4 A_5}\sqrt{R} + \mu \beta A_5 \sqrt{R} \right)^2\right] + \frac{1}{C_4 p_2}\left[\frac{1}{(0.99 C_4 A_5)^2} + 2\right] \mu \beta R\\
      &=\left(B_5 \shs + B_6 \mu \beta \sqrt{R}\right)^2 + 2 B_7 \mu \beta R.
 \end{align*} }
 and that
 \begin{align*}
      \ex{\ntar \indc{\textrm{$T_A$ during a fall}}}
      &\;\;\leq \frac{1}{C_4 p_2} \sum_{j=1}^{R}{C_4\prise{j} \left[\threshone + \mu j \beta\right]}\prod_{\ell=0}^{j-1}{\left(1 - C_4 \prise{\ell}\right)}\\
      &\;\;= \frac{1}{C_4 p_2}\ex{\threshone + \mu \beta G}\\
      &\;\; \leq \frac{1}{C_4 p_2}\left[\threshone + \mu \beta\left(A_5 \sqrt{R} + \frac{1}{0.99 C_4 A_5}\sqrt{R}\right)\right]\\
      &\;\; = \left(B_5 \shs + B_6 \mu \beta \sqrt{R}\right).
 \end{align*}

Defining $F_1 \triangleq B_6$ and $F_2 \triangleq \frac{\left(B_5 + C_3\right)}{B_6}$, it follows that
\begin{equation*}
    \ex{N(T_A) - R} \leq \left((B_5 + C_3) \shs + B_6 \mu \beta \sqrt{R}\right) = F_1 \mu \beta \sqrt{R} \left(1 + \frac{F_2}{\sqrt{\mu \beta}}\right)
\end{equation*}
and that
\begin{equation*}
    \ex{\ntar^2 }\leq \left(B_5 \shs + B_6 \mu \beta \sqrt{R}\right)^2 + 2 B_7 \mu \beta R \leq F_1 (\mu \beta)^2 R \left(1 + \frac{F_2}{\sqrt{\mu \beta}}\right)^2 + 2\mu\beta R. \hfill \Halmos
\end{equation*}

%% file: 9_extension.tex
In this section, we study a natural policy that aims to mitigate the impact of setup times on waiting time.
Our result reveals that this policy is inadequate unless it keeps significantly more servers on compared to the original policy described in Section~\ref{sec:model}; see Figure~\ref{fig:m_comp} for an illustration.
This highlights the need for further studies on designing policies for turning servers on and off in queueing systems with setup times to achieve better tradeoffs between energy efficiency and performance. 

The policies we study in this section are referred to as $m$-policies; we now describe how these policies operate.
Recall that servers are indexed by integers $1,2,\dots,k$, and that we, without loss of generality, re-index our collection of servers to ensure that each job departure is always from the highest indexed server among the busy servers.
Under the original policy (described in Section~\ref{sec:model}), server $i$ is turned off when the number of jobs in the system drops from $i$ to $i-1$.
In contrast, an $m$-policy turns off servers less aggressively.
Specifically, server $i$ is turned off when the number of jobs in the system drops from $i-m$ to $i-m-1$, where $0\le m\le k$.
Here $m$ is referred to as the buffer size.
The special case $m=0$ corresponds to the original policy.
The $m$-policies with $m>0$ reduce job waiting times, as keeping $m$ extra servers on ensures that the next $m$ job arrivals do not need to wait for server setup.

The $m$-policies are more complicated to analyze compared to the original policy.
One difficulty comes from the fact, that under an $m$-policy, servers which are \emph{on} are not necessarily \emph{busy}.
To capture this distinction, we introduce a new variable $\on(t)$ which tracks the number of servers that are \textit{on} at time $t$.
This quantity generally is not equal to the number of servers that are \textit{busy} at time $t$, which we continue to call $Z(t)$.
However, despite this difficulty, we are able to analyze $m$-policies using our MIST technique, with a more complicated coupling construction.

\subsection{\texorpdfstring{Main Result: A Lower Bound on $m$-Policies.}{Main Result: A Lower Bound on m-Policies.}}
Our main result is that, unless the buffer size $m$ exceeds $\Theta\left(\sqrt{R}\right)$, the average waiting time under an $m$-policy will not be much smaller than the original $(m=0)$ policy (see Theorem~\ref{thm:mpol}).

\begin{theorem}[Lower Bound on $m$-Policies]\label{thm:mpol}
When using an $m$-policy to control setup in an M/M/k/Setup-Deterministic system with an offered load $R \triangleq k \rho \geq 100$, a setup time $\setuptime \geq 100 \frac{1}{\mu}$, and an $m$-value with $m \leq \sqrt{R}$,
the expected number of jobs in queue in steady state is lower-bounded by
\begin{align*}
    \ex{Q(\infty)} 
    &\geq F_1\left(\mu \beta \sqrt{R} + \frac{1}{1-\rho}\right),
\end{align*}
for some constant $F_1$ independent of system parameters.
\end{theorem}

\begin{figure}
    \centering
    \includegraphics[scale=0.6]{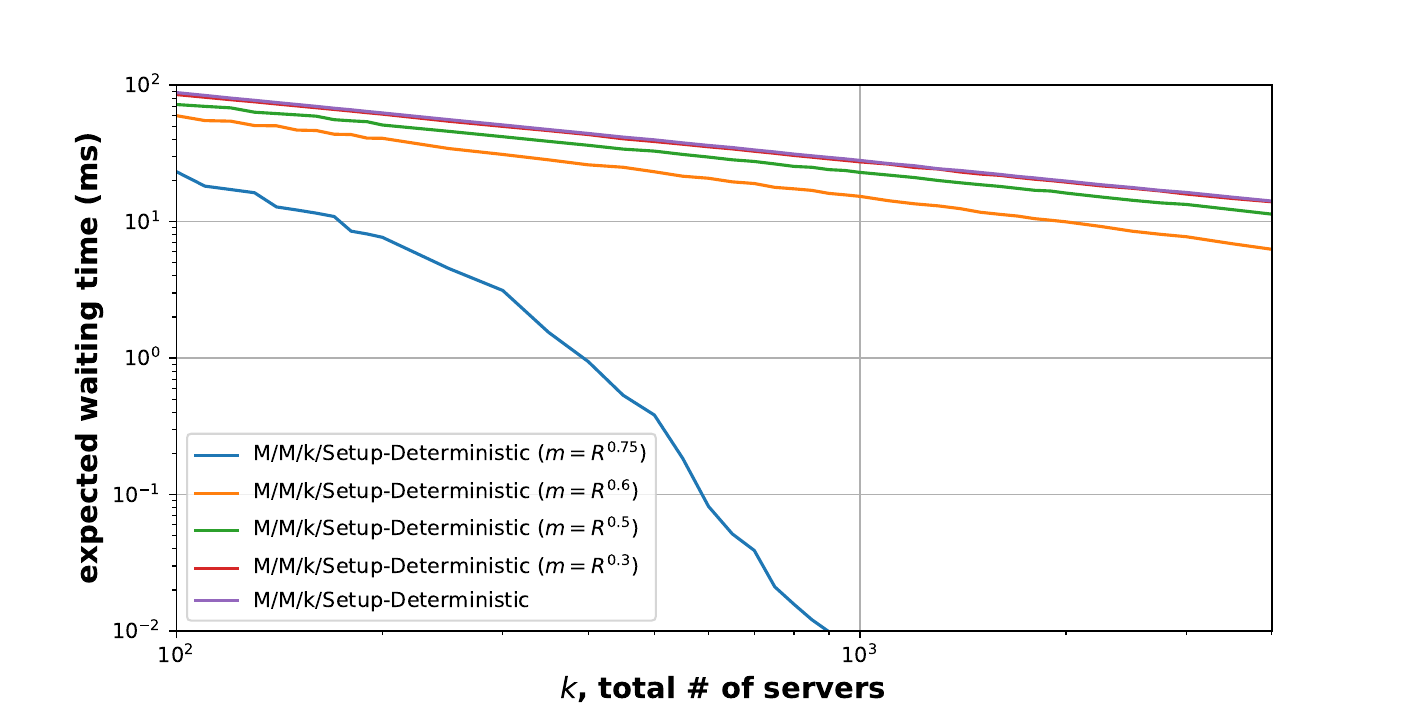}
    \caption{Simulation results showing the average waiting time for the M/M/k/Setup-Deterministic under various $m$-policies. All results obtained with mean service time $\frac{1}{\mu}=1$ ms, mean setup time $\beta = 1000$ ms, and load kept at a constant $\rho = 0.5$. Note that a significant performance difference can only be found once the quantity $m$ is growing faster than $R^{0.5}$.}
    \label{fig:m_comp}
\end{figure}

\paragraph{Proof Overview.}
In the remainder of Section~\ref{sec:extension}, we provide a proof of the stated lower bound, Theorem~\ref{thm:mpol}. 
Our argument here is a variation on the argument made in Section~\ref{sec:lbound}. 
We begin by defining our renewal points to be those moments when the number of \textit{on} servers $\on(t)$ transitions from $(R+1)$ to $R$; note that this is slightly distinct from our previous definition since, in general, the number of busy servers $Z(t) \neq \on(t)$.
Taking $\cyclem$ to be the next renewal point and applying the renewal reward theorem, it suffices to lower bound the renewal integral $\ex{\int_0^{\cyclem}{Q(t) \dd t}}$ and upper bound the cycle length $\ex{\cyclem}$.
In what follows, we show the following two lemmas, which, via the above, suffice to prove the lower bound of Theorem~\ref{thm:mpol}.
\begin{lemma}[Lower Bound on Renewal Integral]\label{lem:extint}
The renewal integral can be bounded by
\begin{equation*}
    \ex{\int_0^X{\left[ N(t) - R \right] \dd t}} \geq F_2 \mu \beta^2 \sqrt{R} + \ibp{\left[F_2\mu \beta \sqrt{R} - (k-R)\right]}{k-R},
\end{equation*}
where $\ibp{n}{j} \triangleq \frac{n}{\mu j}\left[\frac{n+1}{2} + \frac{R}{j} + 1 \right]$ and $F_2 \triangleq 0.23$.
\end{lemma}

\begin{lemma}[Upper Bound on Cycle Length]\label{lem:extcyc}
The length of a renewal cycle can be bounded by
\begin{equation*}
    \ex{X} \leq F_3 \beta + F_4 \frac{\mu \beta \sqrt{R}}{\mu \min\left(k - R, \sqrt{R}\right)},
\end{equation*}
where $F_3 \triangleq 2.6$ and $F_4 \triangleq 7.2$.
\end{lemma}

\paragraph{Completion of Proof, assuming Lemmas~\ref{lem:extint} and~\ref{lem:extcyc}.}
This result follows directly from the casework used in the proof of Theorem~\ref{thm:tightness}.

\input{ExtensionFiles/ext_LB_int}

\subsection{Proof of Lemma~\ref{lem:extcyc}, the Cycle Length Upper Bound.}

\input{ExtensionFiles/ext_UB_X}






%% file: ExtensionFiles/ext_LB_int.tex
\newcommand{\maxmr}{\max{\left(\sqrt{R}, m\right)}}
\newcommand{\plong}{p_{\text{long}}^{(j)}}
\newcommand{\prej}{p_{\text{re}}^{(j)}}
\newcommand{\pstj}{p_{\text{start}}^{(j)}}

\subsection{Proof of Lemma~\ref{lem:extint}, the Integral Lower Bound.}

\paragraph{High-Level Strategy.} As mentioned previously, to prove this lower bound, we use a variation on the strategy used in Section~\ref{sec:lbound}.
In that proof, we first define a random quantity $L$ which represents how many servers end up turning off at the beginning of a renewal cycle, then derive a bound on the desired integral in terms of its expected value $\ex{L}$, then directly derive a bound on $\ex{L}$.
In this proof, we define an almost equivalent random variable $\lext$ and proceed with more or less the same steps; however, the more-complex behavior of the $m$-policy requires a more delicate analysis in order for us to find $\ex{\lext}$.

\paragraph{Definition of $\poss$.}
Before defining $\lext$, we define a related stopping time $\poss$ and prove a small but critical fact about it. 
Let $\poss \triangleq \min\left\{t\geq0:\text{$Q(s) \geq 1$ for all $s \in [t-\beta, t]$}\right\}$ be the first moment that the number of jobs in queue $Q(t)$ has been non-zero for a full setup time. 
At the beginning of a renewal cycle, since servers are setting up if and only if jobs are in the queue, it follows that the time $\poss$ is the first time a server could possibly complete setup.
Since $T_A$ is the first time a specific server sets up, we have $\poss \leq T_A$.

\paragraph{Definition of $\lext$.}
With the stopping time $\poss$ defined, we now define the random variable $\lext$ and describe its role going forward.
Let $\lext \triangleq \max_{s \in [0, \poss]}{\left[R - \on(t)\right]}$ be the maximum number of servers turned off in the period before time $\poss$.
Now, via the same martingale arguments given in the proof of \eqref{eq:int_taul}, we have that
\begin{equation}\label{eq:int_poss}
    \ex{\int_{0}^{\poss}{Q(t) \dd t}} \geq  \frac{1}{2} \beta^2 \mu \ex{\lext}.
\end{equation}
Likewise, in analogy to \eqref{eq:int_aft}, by making a coupling argument we also have that
\begin{equation}\label{eq:ext_int_aft}
    \ex{\int_{\poss}^{X}{Q(t) \dd t }} \geq \ibp{\left[ \ex{N(\poss)} - k\right]^+}{k-R},
\end{equation}
where one should recall that $\ibp{n}{j} \triangleq \frac{n}{\mu j}\left[\frac{n+1}{2} + \frac{R}{j} + 1 \right]$.
After some martingale analysis \eqref{eq:ext_int_aft} becomes
\begin{equation}
    \ex{\int_{\poss}^{X}{Q(t) \dd t }} \geq \ibp{\left[ \beta \ex{\lext} - (k-R)\right]^+}{k-R}.
\end{equation}
Thus, outside of some algebra, to prove Lemma~\ref{lem:extint} it suffices for us to give a bound on the expectation $\ex{\lext}$.
We give this lower bound in the following claim, which we will prove immediately.
\begin{claim}[Bound on $\lext$]\label{clm:ext_lext}
    For any $R \geq 2\max\left(\sqrt{R}, m\right)$ and a constant $F_2 \triangleq 0.23$, we have
    \begin{equation*}
        \ex{\lext} \geq F_2 \min{\left(\frac{R}{m}, \sqrt{R}\right)}.
    \end{equation*}
\end{claim}

\subsection{\texorpdfstring{Proof of Claim~\ref{clm:ext_lext}: Lower Bound on $\ex{\lext}$.}{Lower Bound on E[L\_ex]}}
To bound $\ex{\lext}$, we upper bound the stopping probability $\pr{\lext = j \middle | \lext \geq j}$, showing that, for any $j \leq \max{\left(\sqrt{R}, m\right)},$
\begin{equation}\label{eq:lext_stop}
    \pr{\lext = j \middle | \lext \geq j} \leq 4 \frac{\max{\left(m, \sqrt{R}\right)}}{R}\left(1 + \frac{b_1}{\sqrt{\mu \beta}}\right);
\end{equation}
we defer the proof of this until the next section, Section~\ref{sec:lext_stop}.
From here, we can bound $\ex{\lext}$ using the usual sum-of-tails method, obtaining
\begin{align*}
    \ex{\lext} &= \sum_{j=0}^{R-m}{\pr{\lext \geq j+1}}\\
    &\geq \sum_{j=0}^{\maxmr}{\left( 1 - 4 \frac{\max{\left(m, \sqrt{R}\right)}}{R}\left(1 + \frac{b_1}{\sqrt{\mu \beta}}\right)\right)^j}\\
    &\geq \left[\min{\left(\frac{R}{m}, \sqrt{R}\right)}\right]\frac{1}{4\left(1 + \frac{b_1}{\sqrt{\mu \beta}}\right)}\left[1 - \left( 1 - 4 \frac{\max{\left(m, \sqrt{R}\right)}}{R}\left(1 + \frac{b_1}{\sqrt{\mu \beta}}\right)\right)^{\maxmr}\right]\\
    &\geq \left[\min{\left(\frac{R}{m}, \sqrt{R}\right)}\right]\frac{1}{4\left(1 + \frac{b_1}{\sqrt{\mu \beta}}\right)}\left[1 - e^{-4}\right]\\
    &\geq F_2 \min{\left(\frac{R}{m}, \sqrt{R}\right)},
\end{align*}
where the constant $F_2 \triangleq 0.23 \leq \frac{1}{4\left(1 + \frac{b_1}{\sqrt{\mu \beta}}\right)}\left[1 - e^{-4}\right].$
To complete our proof, it suffices to show \eqref{eq:lext_stop}.

\subsubsection{\torpdf{Proof of \eqref{eq:lext_stop}: Upper Bound on $\pr{\lext = j \middle | \lext \geq j}$.}{Upper Bound on Pr(Lex = j | Lex >= j}}\label{sec:lext_stop}
To show the inequality \eqref{eq:lext_stop}, we begin by giving an exact analysis of the probability $\pr{\lext = j \middle | \lext \geq j}$ by analyzing the sample paths of $N(t)$ at the beginning of a renewal cycle, then afterwards apply some straightforward Markov chain analysis.
In particular, for two probabilities $\pstj$ and $\plong$ and a parameter $\gamma$ (to be defined later), we first show that
\begin{equation}\label{eq:lext_stop_p}
    \pr{\lext = j \middle | \lext \geq j} \leq \min\left(\frac{\plong}{\gamma}, \pstj \right).
\end{equation}
Afterwards, we bound these quantities, showing that, for some $x \in (0,1)$ and for all $j \leq \max\left(\sqrt{R}, m\right)$,
\begin{equation}\label{eq:lext_stop_x}
    \min\left(\frac{\plong}{\gamma}, \pstj \right) \leq 2 \frac{\max{\left(m, \sqrt{R}\right)}}{R}\left(1 + \frac{b_1}{\sqrt{\mu \beta}}\right) \min\left(\frac{1}{x},\frac{1}{1-x}\right).
\end{equation}
Since $\min\left(\frac{1}{x},\frac{1}{1-x}\right) \leq 2$, we obtain that
\begin{equation*}
    \pr{\lext = j \middle | \lext \geq j} \leq 4 \left(1 + \frac{b_1}{\sqrt{\mu \beta}}\right) \frac{\max{\left(m, \sqrt{R}\right)}}{R}, 
\end{equation*}
as desired.
Thus, it suffices to show \eqref{eq:lext_stop_p} and \eqref{eq:lext_stop_x}.

\subsubsection{\torpdf{Proof of \eqref{eq:lext_stop_p}.}{Proof of Conditional Stopping Probability.}} 
We begin by exactly analyzing 
$\pr{\lext = j \middle | \lext \geq j}$.
We make a recursive argument, based on a sequence of events (and their implicitly-defined stopping times). Note that, if we were to write this explicitly, then we would need to invoke our main lemma, Lemma~\ref{lem:waldstop}.

\paragraph{First set of events.} Assume that one begins at time $\tau_j$, the first moment where one knows that $\lext \geq j$; note that, since at time $\tau_j$ we have just turned off the $(R-j + 1)$-th server, we know the number of jobs $N(t) = R - j + 1 - m$. From here, one can wait until one of two events happen: either 1) the number of jobs $N(t)$ decreases to $(R- j - m)$, at which point $\lext > j$; or 2) the number of jobs $N(t)$ increases to $(R -j + 1)$, at which point the event $\left[\lext = j\right]$ still has a chance to occur.
We call the probability of the latter event $\pstj$; we think of this as a ``starting'' probability.

\paragraph{Second set of events.} At the moment when $N(t)$ reaches $(R-j + 1)$, the number of jobs in the queue $Q(t) = 1$. As such, we can again partition our sample paths based on whether: 1) $Q(t)$ stays non-zero for a full setup time, at which point $\left[\lext = j\right]$; or 2) $Q(t)$ drops to $0$ within a setup time.
We call the probability of the first event $\plong$.

\paragraph{Third set of events.} In the case where the queue empties before a setup time has passed, the number of jobs $N(t) = (R-j)$ and we again can condition on the next event to happen: either 1) the number of jobs $N(t)$ decreases to $(R- j - m)$, at which point $\lext > j$; or 2) the number of jobs $N(t)$ increases to $(R -j + 1)$, at which point the event $\left[\lext = j\right]$ still has a chance to occur.
We call the latter probability $\prej$, since it can be thought of as a ``restart'' probability.
Note that, in the case where the number of jobs $N(t)$ rises, we have reached an identical state to that which initiates the second set of events; thus we can induct.

\paragraph{Result in terms of probabilities.} From this analysis, it follows that
\begin{align*}
    \pr{\lext = j \middle | \lext \geq j } &= \pstj\left(\plong + (1-\plong)\prej\left(\plong + (1-\plong)\prej\left( \dots \right) \right) \right) \\
   & = \pstj \plong \frac{1}{1 - \left(1-\plong\right) \prej}.
\end{align*}
By noting that the behavior of $N(t)$ in the first and third scenarios is sample-path-equivalent to the behavior of an M/M/$(R-j)$, and that the behavior of $N(t)$ in the second scenario is sample-path-equivalent to that of an (overloaded) M/M/1, we can analyze these probabilities.
In Section~\ref{sec:ext_prob_bnds}, we show via basic Markov chain analysis that $\prej =  1 - \gamma \pstj$, where $\gamma \triangleq \frac{(R - j)!}{(R-j-m)!}\frac{1}{R^m} \leq 1$.
It follows that both
\begin{align*}
\pstj \plong \frac{1}{1 - \left(1-\plong\right) \prej} = \pstj\frac{ \plong}{\plong + \left(1-\plong\right) \gamma \pstj} \leq \pstj
\end{align*}
and
\begin{align*}
\pstj \plong \frac{1}{1 - \left(1-\plong\right) \prej} = \plong\frac{\pstj}{\gamma \pstj + \left(1-\gamma \pstj\right) \plong} \leq \frac{\plong}{\gamma},
\end{align*}
and thus that
\begin{equation*}
     \pr{\lext = j \middle | \lext \geq j} \leq \min \left( \frac{\plong}{\gamma}, \pstj\right).
\end{equation*}

\subsubsection{\torpdf{Proof of \eqref{eq:lext_stop_x}: Upper Bound on $\pstj$ and $\plong$.}{Upper Bound on pstj and plong.}}
To complete our analysis of $\pr{\lext = j \middle | \lext \geq j}$, it suffices to bound the probabilities $\plong$ and $\pstj$, and the constant $\gamma$.
To upper bound the probability $\plong$, we note that, from the analysis in Section~\ref{sec:lbcons},
\begin{align*}
    \plong \leq 1 - \left(1 - \frac{j}{R}\right)\left(1 - \frac{b_1}{\sqrt{\mu\beta R}}\right)
    &= \frac{b_1}{\sqrt{\mu\beta R}} + \frac{j}{R}\left(1 - \frac{b_1}{\sqrt{\mu\beta R}}\right)\\
    &\leq \frac{\left(j + b_1\frac{\sqrt{R}}{\sqrt{\mu\beta}}\right)}{R}\\
    &\leq \left( 1 + \frac{b_1}{\sqrt{\mu \beta}}\right)\frac{\sqrt{R}}{R}\\
    &\leq \left( 1 + \frac{b_1}{\sqrt{\mu \beta}}\right)\frac{2 \maxmr}{R}.
\end{align*}
To bound the constant $\gamma$, we note that, defining $x \triangleq \left(1 - \frac{j+m}{R}\right)^m$,
\begin{equation*}
    \gamma \triangleq \frac{(R-j)!}{(R-j-m)!}\frac{1}{R^m} \geq \left(\frac{R-j-m}{R}\right)^m = \left(1 - \frac{j+m}{R}\right)^m \triangleq x.
\end{equation*}
To upper bound the probability $\pstj$, we use the analysis from Section~\ref{sec:ext_prob_bnds}, which shows that
\begin{align*}
    \pstj \leq \frac{j+m}{R} \frac{1}{1 - \left(1 - \frac{j+m}{R}\right)^{m+1}}
    &\leq \frac{j+m}{R} \frac{1}{1 - \left(1 - \frac{j+m}{R}\right)^{m}}\\
    &= \frac{j+m}{R} \frac{1}{1 - x}\\
    &\leq \frac{2\maxmr}{R} \frac{1}{1-x}\\
    &\leq \left(1 + \frac{b_1}{\sqrt{\mu\beta R}}\right) \frac{2\maxmr}{R} \frac{1}{1-x}.
\end{align*}
Thus, from \eqref{eq:lext_stop_p}, we have that
\begin{align*}
    \pr{\lext = j \middle | \lext \geq j } &\leq \min\left( \frac{\plong}{\gamma}, \pstj\right)\\
    & \leq \left(1 + \frac{b_1}{\sqrt{\mu\beta R}}\right) \frac{2\maxmr}{R} \min\left( \frac{1}{x}, \frac{1}{1-x}\right)\\
    & \leq 4 \left(1 + \frac{b_1}{\sqrt{\mu \beta}}\right) \frac{\max{\left(m, \sqrt{R}\right)}}{R}, 
\end{align*}
as desired.
Thus, we complete our proof of Claim~\ref{clm:ext_lext}, our lower bound on $\ex{\lext}$.

%% file: ExtensionFiles/ext_UB_X.tex
\paragraph{Preliminaries: Cycle decomposition.}
We now prove Lemma~\ref{lem:extcyc}, an upper bound on the expected cycle length $\ex{\cyclem}$.
As before, we begin by breaking a renewal cycle into three parts; we later analyze each part separately.
To do so, we define a threshold $H$ and two stopping times; the latter stopping time will depend on $H$.
To begin, noting that $R$ is the minimum number of servers needed to stabilize the system, we define the threshold $H\triangleq \min\left(k - R, \sqrt{R}\right)$ as the minimum between $(k-R)$, the number of ``extra'' servers, and $\sqrt{R}$, a quantity which represents the ``typical variation'' in the system.
Next, we define the accumulation time $T_A \triangleq \minpar{t\geq 0: \on(t) \geq R+1}$ as the first moment that the $(R+1)$-th server turns on.
Finally, we define the draining time $T_B \triangleq \minpar{t\geq T_A: N(t) \leq R+H}$ as the first moment after time $T_A$ that the number of jobs $N(t)$ dips below the threshold $(R+H)$.

\paragraph{Preliminaries: Definition of phases.}
With these two stopping times defined, our renewal cycle takes on a similar decomposition as before. 
From time $0$ to time $T_A$, the number of servers that are \textit{on} is $\leq (R+1)$, meaning that the departure rate of jobs must be less than the arrival rate of jobs. 
Accordingly, jobs accumulate and thus we call this the \textit{accumulation} phase.
From time $T_A$ to time $T_B$, the number of servers \textit{on} is $\geq (R+1)$, meaning that, on the whole, the accumulated jobs are \textit{draining} from the system; thus we call this the \textit{draining} phase.
From time $T_B$ to the end of the cycle at time $\cyclem$, the behavior of the system roughly resembles that of a balanced random walk; we thus call this period the \textit{balanced} phase.

\paragraph{Proof strategy.} 
Using this decomposition, we prove the following bounds on the expected length of each phase:
for the accumulation phase $[0,T_A)$
\begin{equation}\label{eq:ext_TA}
    \ex{T_A} \leq 1.08 \beta,
\end{equation}
for the draining phase $[T_A, T_B)$
\begin{equation}\label{eq:ext_TB}
    \ex{T_B - T_A} \leq \beta + E_1\frac{\mu \beta \sqrt{R}}{\mu H},
\end{equation}
and for the balanced phase $[ T_B, X)$
\begin{equation}\label{eq:ext_X}
    \ex{X - T_B} \leq E_3 \frac{1}{\mu} + E_{10} \beta.
\end{equation}
Putting these together, we confirm that
\begin{equation*}
    \ex{X} \leq F_3 \beta + F_4 \frac{\mu \beta \sqrt{R}}{\mu (k - R)}.
\end{equation*}
In what follows, we briefly discuss the proof of each of these inequalities in turn.

\subsubsection{\texorpdfstring{Proof of \eqref{eq:ext_TA}: Bound on $\ex{T_A}$.}{Bound on ETA}}
For the bound on $\ex{T_A}$, our previous proof still applies: since $\on(t) \leq R$ until time $T_A$, we can couple the system to an M/M/$R$ queue and bound the length of an analogous version of $T_A$ in that system.
To review the argument, we can alternatively define $T_A = \minpar{t\geq0: \min_{s \in [t - \beta, t]}{N(t)} \geq R+1}$ as the first moment that the number of jobs $N(t)$ stays above $(R+1)$ for a full setup time $\beta$.
If we let $\ncouple(t)$ denote the number of jobs in a coupled M/M/$R$ queue starting with $\ncouple(0) = (R - m)$, then we can accordingly define $\Tilde{T_A}$ as the first moment that the coupled number of jobs $\ncouple(t)$ stays above $(R+1)$ for a full setup time $\beta$.
Since the departure rate of the original system is upper-bounded by the departure rate of this coupled M/M/$R$ system, the original system always contains more jobs, i.e. $N(t) \geq \ncouple(t)$.
Consequently $T_A \leq \Tilde{T_A}$ and thus it suffices to bound $\ex{\Tilde{T_A}}$; from Section~\ref{sec:UBTA},
\begin{equation*}
    \ex{T_A} \leq \ex{\Tilde{T_A}} \leq 1.08 \beta.
    \tag*{\Halmos}
\end{equation*}

\subsubsection{\torpdf{Proof of \eqref{eq:ext_TB}: Bound on $\ex{T_B - T_A}$.}{Bound on E[TA - TB].}}
For the bound on $\ex{T_B - T_A}$, our strategy again is to derive a bound on the conditional expectation of $\left[T_B - T_A\right]$, where we condition on the number of jobs in the system at time $T_A$, and afterwards give a bound on both the first and second moments of this number of jobs $N(T_A)$.
Actually, in this case, our conditional bound from Section~\ref{sec:UBintTB} still applies;
by making worst-case assumptions about the system's setup state at time $T_A$ and noting that, if the number of jobs stays large for a full setup time then many servers must be \textit{on}, we can use a combination of random walk and busy period analysis to obtain
\begin{equation}\label{eq:ext_tb_cond}
    \ex{T_B - T_A\middle | \filt{T_A}} \leq \beta + \frac{N(T_A) - R}{\mu H} + \frac{1}{\mu}\ln{\left(\frac{N(T_A) - R}{H}\right)};
\end{equation}
see Section~\ref{sec:UBintTB} for details.
To complete our bound of the draining phase, we must bound the expected value $\ex{N(T_A)}$; we defer the proof of this bound to Section~\ref{sec:ext_nta}. 
There, we find that
\begin{equation}\label{eq:ext_nta}
    \ex{N(T_A) - R} \leq E_1 \mu \beta \sqrt{R},.
\end{equation}
where $E_1 \triangleq 5.21.$
Thus, combining \eqref{eq:ext_nta} with \eqref{eq:ext_tb_cond} (and noting that $\frac{x + \ln(x)}{x} \leq 1 + \frac{1}{e}$), we find that
\begin{align*}
    \ex{T_B - T_A} &\leq  \beta + \ex{\frac{N(T_A) - R}{\mu H}} + \frac{1}{\mu}\ex{\ln{\left(\frac{N(T_A) - R}{H}\right)}}\\
    &\leq \beta + \frac{\ex{N(T_A) - R}}{\mu H} + \frac{1}{\mu}\ln{\left(\frac{\ex{N(T_A) - R}}{H}\right)}\\
    &\leq\beta + \left(1 + \frac{1}{e}\right) \frac{\ex{N(T_A) - R}}{\mu H}\\
    &\leq \beta + 1.37 E_1 \frac{\mu \beta \sqrt{R}}{\mu H}.
\end{align*}

\subsubsection{\torpdf{Proof of \eqref{eq:ext_X}: Bound on $\ex{X - T_B}$.}{Upper bound on E[X - TB]}}
\paragraph{Main idea: Threshold-based division.} 
To bound the final portion of the renewal cycle, $\ex{\cyclem - T_B}$, we divide time based on the threshold $H$ and make a coupling argument.
More specifically, we make an argument based on whether the number of jobs $N(t)$ is \textit{below} or \textit{above} the threshold $(R+H)$.
First, we note that, when the number of jobs $N(t)$ is \textit{below} $(R+H)$, the departure rate of jobs is lower-bounded by $\mu \min\left(R, N(t)\right)$. 
Said differently, while below $(R+H)$, the number of jobs $N(t)$ is upper-bounded by the number of jobs in an appropriately-coupled M/M/$R$/$(R+H)$ queue.
On the other hand, when the number of jobs $N(t)$ becomes \textit{larger} than $(R+H)$, we can use the ``wait-busy'' period idea of the previous phase to bound the amount of time it takes for number of jobs $N(t)$ to come back down to $(R+H)$; we call this an excursion.

\paragraph{The below-threshold case.} To analyze the below-threshold portion of time, we make a coupling argument based on our above departure rate bound.  
If we let $T^{\text{M/M/$R$/$(R+H)$}}_{(R+H)\to(R-m)}$ to be the time it takes for an M/M/$R$/$(R+H)$ to go from $(R+H)$ jobs in the system to $(R-m)$ jobs in the system, then, from a straightforward coupling argument, the time spent below $(R+H)$ can be upper-bounded by $\ex{T^{\text{M/M/$R$/$(R+H)$}}_{(R+H)\to(R-m)}}$.
As such, it suffices to bound the expectation of this first passage time, which, for $m \leq \sqrt{R}$, was done in Claim~\ref{clm:mminf}.
This gives
\begin{align*}
    \ex{\text{time spent with $N(t) \leq R+H$}} \leq \ex{T^{\text{M/M/$R$/$(R+H)$}}_{(R+H)\to(R-m)}}
    &\leq 7 \frac{m + H}{\mu\sqrt{R}}\\
    &\leq 7 \frac{\sqrt{R} + \min\left(\sqrt{R},k-R\right)}{\sqrt{R}} \frac{1}{\mu}\\
    &\leq 14 \frac{1}{\mu}.
\end{align*}

\paragraph{The above-threshold case.} To analyze the above-threshold portion of time, we bound the number of excursions and use the ``wait-busy'' idea to bound the length of an excursion.
To bound the number of excursions, we note (by our below-threshold coupling argument) that the probability of having an additional excursion is upper-bounded by the probability that, in an M/M/$R$ queue started with $(R+H)$ jobs in the system, the queue reaches $(R+H+1)$ jobs in system before it reaches $(R-m)$ jobs in system; we call this probability $p_e$.
Accordingly, the number of excursions is stochastically dominated by a $\text{Geo($1-p_e$)} -1$ random variable, and thus the expected number of excursions is bounded by $\frac{1}{1-p_e}$; we give an explicit bound on $p_e$ later, showing that
\begin{equation}\label{eq:ext_pe}
    \frac{1}{1-p_e} \leq \sqrt{R}\left(1 + e^{\frac{m(m+1)}{R-m}}\right),
\end{equation}
which, for $m \leq \sqrt{R}$ and $R \geq 100$, can be bounded by
\begin{equation*}
    \frac{1}{1-p_e} \leq \sqrt{R} \left(1 + e^{\frac{\sqrt{R} + 1}{\sqrt{R} - 1}}\right) \leq 4.4 \sqrt{R}.
\end{equation*}

Next, by using the ``wait-busy'' idea of the previous phase, previous arguments tell us that, for $b_1 = \sqrt{\frac{2}{\pi}}$,
\begin{equation}
    \ex{\text{length of an excursion}} \leq \frac{1}{\mu}b_1\sqrt{\frac{\mu\beta}{R}} + \frac{1}{\mu H}.
\end{equation}
Thus, we obtain that for $\mu \beta \geq 100$,
\begin{align*}
    \ex{\text{time spent with $N(t) > R+H$}} &\leq \frac{1}{1 - p_e}\cdot \left[ \frac{1}{\mu}b_1\sqrt{\frac{\mu\beta}{R}} + \frac{1}{\mu H} \right]\\
    &\leq \beta \frac{4.4 b_1}{\sqrt{\mu \beta}}  + \frac{4.4 \sqrt{R}}{\mu H}\\
    &\leq 0.36\beta + \frac{4.4 \sqrt{R}}{\mu H}.
\end{align*}
We arrive at the following bound on $\ex{\cyclem - (T_A+T_B)}$:
\begin{align*}
    \ex{\cyclem - T_B} &\leq \ex{\text{time spent with $N(t) \leq R+H$}} + \ex{\text{time spent with $N(t) > R+H$}}\\
    &\leq  14\frac{1}{\mu} + 0.36\beta + \frac{4.4 \sqrt{R}}{\mu H}\\
    &\leq 0.5 \beta + \frac{4.4 \sqrt{R}}{\mu H}.
\end{align*}

\subsubsection{Proof of Lemma~\ref{lem:extcyc}: Closing.}
As mentioned, with \eqref{eq:ext_TA},\eqref{eq:ext_TB}, and \eqref{eq:ext_X} proven, we complete our bound on $\ex{X}$ by simple addition:
\begin{align*}
    \ex{X} &= \ex{T_A} + \ex{T_B-T_A} + \ex{X-T_B}\\
    &\leq  1.08\beta + \beta + 1.37E_1 \frac{\mu \beta\sqrt{R}}{\mu{H}} + 0.5 \beta + \frac{4.4 \sqrt{R}}{\mu H} \\
    &\leq 2.6 \beta + 7.2\frac{\mu \beta\sqrt{R}}{\mu H}\\
    &= F_3 \beta + F_4 \frac{\mu \beta\sqrt{R}}{\mu H},
\end{align*}
where we have taken $F_3 \triangleq 2.6$ and $F_4 \triangleq 7.2$.

%% file: ExtensionFiles/ext_nta.tex
\newcommand{\epochnum}{n_e}
\newcommand{\psuccj}{p_{\text{succ}}^{(j)}}
\newcommand{\numexcj}{n_{\text{exc}}^{(j)}}
\newcommand{\plongj}{p_\text{long}^{(j)}}
\newcommand{\oldthreshold}{C_3 \sqrt{\mu \beta R}}
\newcommand{\pth}{p_{\text{th}}}
\newcommand{\pf}{p_f}
Here, we discuss our bound on the expected value of the number of jobs $\ex{N(T_A) - R}$ for $m$-policies. In particular, we show that with a constant $E_1 \triangleq 5.21$,
\begin{equation*}
    \ex{N(T_A) - R} \leq E_1 \mu \beta \sqrt{R}.
\end{equation*}

\paragraph{Definition of epoch.} To do so, we again make use of the notion of \textit{epochs}. 
Recall that time $\tau_j \triangleq \minpar{t\geq 0: \on(t) \leq R-j}$ is the first moment that only $(R-j)$ servers are on.
We call the period $\big [\tau_j, \mintwo{\tau_{j+1}}{T_A}\big )$ the $j$-th epoch, and we say epoch $j$ occurs if the time $\tau_j$ happens before time $T_A$.
As a shorthand, we use $\epochnum$ to denote the number of epochs that occur; note that the condition $\left[\text{epoch $j$ occurs}\right]$ is equivalent to the condition $\left[n_e \geq j\right]$.

\paragraph{Definition of a trial.}
To analyze the time $T_A$ more closely, we define a series of trials. 
The first trial in epoch $j$ (that is, the $(1,j)$-th trial) begins at the first moment that $N(t) \geq R+\oldthreshold$ during epoch $j$.
We say this trial \textit{succeeds} if the number of jobs $N(t)$ does not drop below $(R+1)$ for a full setup time; otherwise, we say that this trial \textit{fails}.
If, in the same epoch, the number of jobs $N(t)$ again reaches $(R+1)$, then the next trial begins and may either succeed or fail; this continues until either the next epoch begins or a trial succeeds.
We refer to the $i$-th trial to occur in epoch $j$ as trial $(i,j)$.

We now note that
\begin{equation}\label{eq:ext_trial}
    \ex{N(T_A) - R} \leq  \oldthreshold + \sum_{i,j}{\pr{\text{trial $(i,j)$ occurs}} \ex{\left[N(T_A) - R\right] \indc{\text{trial $(i,j)$ succeeds}}\middle | \text{trial $(i,j)$ occurs}}}.
\end{equation}
To analyze this conditional expectation, we note that, if trial $(i,j)$ occurs at time $\phi_{i,j}$ and this trial is successful, then time $T_A = \left(\phi_{i,j} + \beta \right)$. 
On the other hand, let $\psi{i,j} \triangleq \minpar{t>0: N(\phi_{i,j} + t) < R + 1}$ be the next moment the number of jobs $N(t)$ falls below $(R+1)$; we make two points.
First, the trial succeeds if and only if $\psi_{i,j} \geq \beta$.
Second, if the trial fails, then $\left[N(\psi_{i,j})- R\right] = 0$.
As such, 
\begin{equation*}
    \left[N(T_A) - R\right] \indc{\text{trial $(i,j)$ succeeds}} =  \left[N\left(\phi_{i,j} + \mintwo{\psi_{i,j}}{\beta}\right) - R\right].
\end{equation*}
Noting that $V(t) \triangleq N(\phi_{i,j} + t) - R - \mu j t $ is a super-martingale (since we are in epoch $j$), it follows from Doob's Optional Stopping Theorem that
\begin{align*}
    V(0) = \oldthreshold 
    &\geq \ex{V\left(\mintwo{\psi_{i,j}}{\beta} \right)} \\
    &= \ex{\left[N\left(\phi_{i,j} + \mintwo{\psi_{i,j}}{\beta}\right) - R\right]\middle | \filt{\phi_{i,j}}} - \mu j \ex{\mintwo{\psi_{i,j}}{\beta}\middle | \filt{\phi_{i,j}}},
\end{align*}
and thus that
$$\ex{\left[N(T_A) - R\right] \indc{\text{trial $(i,j)$ succeeds}}\middle | \text{trial $(i,j)$ occurs}}
\leq \oldthreshold + \mu j \ex{\mintwo{\psi_{i,j}}{\beta}\middle | \text{trial $(i,j)$ occurs}}.$$

\paragraph{Reduction to trial quantities.}
Applying this to \eqref{eq:ext_trial},
\begin{align*}
    \ex{N(T_A) - R} &= \sum_{i,j}{\pr{\text{trial $(i,j)$ occurs}} \ex{\left[N(T_A) - R\right] \indc{\text{trial $(i,j)$ succeeds}}\middle | \text{trial $(i,j)$ occurs}}}\\
    &\leq \sum_{i,j}{\pr{\text{trial $(i,j)$ occurs}} \left[\oldthreshold + \mu j \ex{\mintwo{\psi_{i,j}}{\beta}\middle | \text{trial $(i,j)$ occurs}}\right]}\\
    &\leq \oldthreshold\ex{\text{total number of trials}} + \sum_j{ \mu j \ex{\text{total time spent on trials during epoch $j$}}}.
\end{align*}

To address the first term, we note that for each trial which occurs, according to Claim~\ref{clm:dcprob}, with probability at least $\left[1- 2\Phi\left(-\frac{3}{\sqrt{2}}\right) - \frac{1}{100}\right]\geq 0.95 \triangleq C_4$ the first phase ends.
Hence $\ex{\text{total number of trials}}\leq \frac{1}{0.95}$.

To address the second term, we note that 
\begin{align*}
    &\ex{\text{total time spent on trials during epoch $j$}} \\
    &\;\;= \pr{\text{trial $(1,j)$ occurs}}\ex{\text{time spent on trials during epoch $j$} \middle | \text{trial $(1,j)$ occurs}}\\
    &\;\;\leq \pr{\text{trial $(1,j)$ occurs}}\ex{\text{time remaining in first phase} \middle | \text{trial $(1,j)$ occurs}}\\
    &\;\;\leq\pr{\text{trial $(1,j)$ occurs}} \ex{\Tilde{T_A}}\\
    &\;\;\leq\pr{\text{trial $(1,j)$ occurs}} 1.08 \beta,
\end{align*}
where the final inequality is via coupling the system to an M/M/$R$ queue in the same way as in Section~\ref{sec:UBTA}.
Defining $G_1\triangleq 3.3 > \frac{1}{0.95} C_3$ and $G_2 \triangleq 1.08$, we find that
\begin{equation*}
    \ex{N(T_A) - R} \leq G_1 \sqrt{\mu \beta R} + G_2 \mu \beta \sum_{j}{j \pr{\text{trial $(1,j)$ occurs}}}.
\end{equation*}

\paragraph{Development using epochs.} From here, we note that $ \prise{j} \triangleq \pr{\text{trial $(1,j)$ occurs}\middle | n_e \geq j}$ and thus it suffices to bound $\pr{n_e \geq j}$.
To provide an upper bound on the probability that epoch $j$, it suffices to show a lower bound on the conditional probability $\pr{n_e = j \middle | n_e \geq j}$.
In Section~\ref{sec:prisebnd}, we have already shown 
\begin{equation*}
    \pr{n_e = j \middle | n_e \geq j} \geq (1- C_4 \prise{j});
\end{equation*}
applying this bound, one finds that
\begin{equation*}
    \pr{n_e \geq j} \leq \prod_{i=0}^{j-1}{\left(1 - C_4 \prise{i}\right)}.
\end{equation*}
Thus, we obtain
\begin{align*}
    \sum_j{j \pr{\text{trial $(1,j)$ occurs}}} &\leq \frac{1}{C_4}\sum_j{C_4\prise{j} j\prod_{i=0}^{j-1}{\left(1 - \prise{i}\right)}}\\
    &=\frac{1}{C_4} \sum_{j}{\prod_{i=0}^{j}{\left(1- C_4 \prise{i}\right)}}.
\end{align*}
From here, it suffices to lower bound $\prise{i}$.

\subsubsection{\torpdf{Lower bound on $\prise{j}$.}{Lower bound on prise(j).}}\label{sec:ext_prise_bnd_st}
To bound the probability that the number of jobs $N(t)$ reaches the threshold $R+\oldthreshold$ at some point during epoch $j$, we focus on one particular manner in which that event might happen.
To define this manner, we introduce the notion of an \textit{excursion}.
We do so recursively: the first excursion in epoch $j$ occurs the moment the number of jobs reaches $(R-j + 1)$, and ends when either 1) the number of jobs $N(t)$ dips back down to $(R-j)$ or 2) a full setup time has passed.
If a full setup time passes during an excursion, we call it \textit{long}; otherwise, we call it \textit{short}.
Successive excursions occur whenever an up-crossing of $(R-j)$ occurs, but \text{only} if every previous excursion in the epoch has been \textit{short}.
Thus, we can lower bound the probability of reaching the desired threshold by requiring that it occurs before the epoch's \textit{first long excursion}, i.e.
\begin{align*}
    \prise{j} \geq \pstj\left( \pth + \pf \prej \pth + \left(\pf\prej\right)^2 \pth + \dots\right)= \pth \pstj\frac{1}{1 - \pf \prej}
\end{align*}
where we have used the fact that the system state is the same at the beginning of every excursion, and where we have defined $\pth$ as the probability that the number of jobs $N(t)$ crosses the threshold during excursion $(1,j)$, defined $\prej$ as the probability that, after a short excursion, another excursion occurs, defined $\pstj$ as the probability that an excursion occurs, and defined $\pf$ as the probability that the excursion is both short and ends without $N(t)$ having reached the desired threshold.
We arrive at the following claim, proven in Section~\ref{sec:ext_pth_clm}:
\begin{claim}\label{clm:ext_pth_clm}
    Recall the definitions of $\pth, \pf,\prej,$ and $\pstj$. 
    Then, for $j\geq \sqrt{R}$, we have the following bound
    \begin{equation*}
        \pth\pstj\frac{1}{1 - \pf \prej} \geq E_6\frac{j}{R}\geq E_6 \frac{1}{\sqrt{R}},
    \end{equation*}
    where $E_6 = 0.32$.
\end{claim}

Accordingly, we find that
\begin{align*}
    \ex{N(T_A)-R} &\leq G_1 \sqrt{\mu \beta R} + G_2\mu \beta \frac{1}{C_4}\sum_{j=0}^{R-m}{\prod_{i=0}^{j}{\left(1- C_4 \prise{i}\right)}}\\
    &\leq G_1 \sqrt{\mu \beta R} + G_2\mu \beta \frac{1}{C_4}\left[\sqrt{R} \sum_{j=\sqrt{R}}^{R-m}{\prod_{i=0}^{j}{\left(1- C_4 \prise{i}\right)}}\right]\\
    &\leq  G_1 \sqrt{\mu \beta R} + G_2\mu \beta \frac{1}{C_4}\left[\sqrt{R} + \sum_{j = \sqrt{R}}{\left(1- C_4 E_6 \frac{1}{\sqrt{R}}\right)^{\left[j - \sqrt{R}\right]^+}}\right]\\
    &= G_1 \sqrt{\mu \beta R} +\mu \beta \sqrt{R} \frac{G_2}{C_4}\left[1 + \frac{1}{C_4 E_6}\right]\\
    &= \mu \beta \sqrt{R} \left( \frac{G_2}{C_4}\left[1 + \frac{1}{C_4 E_6}\right] + \frac{G_1}{\sqrt{\mu \beta}} \right)\\
    &\leq E_1 \mu \beta \sqrt{R}
\end{align*}
as desired, where we have taken $E_1 \triangleq 5.21 > \left( \frac{G_2}{C_4}\left[1 + \frac{1}{C_4 E_6}\right] + \frac{G_1}{\sqrt{100}} \right)$.
Thus, we have shown \eqref{eq:ext_nta}, an upper bound on the quantity $\ex{N(T_A)-R}$ for $m$-policies.

%% file: ExtensionFiles/ext_probabilities.tex
\newcommand{\prbu}[1]{p_{#1}^{\text{up}}}

We now analyze a general hitting probability problem in the M/M/$\infty$ queue.
In particular, we consider the case where the arrival rate is $k\lambda = \mu R$, and define an upper and lower boundary of $b$ and $a$ respectively.
We then find the probability that, starting from state $i$, we hit the upper boundary state $b$ before we hit the lower boundary point $a$; we call this probability $\prbu{i}$.

\subsubsection{\torpdf{Solving for $\prbu{i}$.}{Solving for pup(i)}} We solve for these probabilities using the usual finite difference method. We first note our boundary conditions: the probabilities at the boundaries are $\prbu{a} = 0$ and $\prbu{b} = 1$.
From here, we note that, by a conditioning argument,
\begin{equation*}
    \prbu{i} = \frac{R}{R+i} \prbu{i+1} + \frac{i}{R + i}\prbu{i-1}.
\end{equation*}
After some algebra, we obtain that, defining $\Delta_i = \prbu{i+1} - \prbu{i}$,
\begin{equation*}
    \frac{\prbu{i+1} - \prbu{i}}{\prbu{i} - \prbu{i-1}} = \frac{i}{R} = \frac{\Delta_i}{\Delta_{i-1}}.
\end{equation*}
This shows that
\begin{equation*}
    \Delta_i = \frac{i}{R} \cdot \frac{i-1}{R} \cdot \dots \cdot \frac{a+1}{R} \cdot \Delta_a = \frac{(i)!}{(a)! R^{i-a}} \Delta_a.
\end{equation*}
Using our upper boundary condition, we have
\begin{align*}
    1 &= \prbu{b} = (\prbu{b} - \prbu{b-1}) + (\prbu{b-1} + \prbu{b-2}) + \dots + (\prbu{a+2} - \prbu{a+1}) + (\prbu{a+1}) 
    = \sum_{i = a}^{b-1}{\Delta_i} 
    = \Delta_a \sum_{i=a}^{b-1}{\frac{(i)!}{(a)! R^{i-a}}}.
\end{align*}
Thus we have that
\begin{equation}
    \Delta_a = \prbu{a+1} = \frac{1}{\sum_{i=a}^{b-1}{\frac{(i)!}{(a)! R^{i-a}}}} =  \frac{1}{\sum_{i=0}^{b-1 - a}{\frac{(i+a)!}{(a)! R^{i}}}}.
\end{equation}
Notably, we also have that 
\begin{equation}
    \prbu{b-1} = 1 - \Delta_{b-1} = 1 - \frac{(b-1)!}{(a)! R^{b-1-a}} \Delta_a.
\end{equation}

\subsubsection{\torpdf{Bounds on $\Delta_a$.}{Bounds on Delta\_a.}}
We now give two upper bounds and a lower bound on $\Delta_a$ by bounding the sum in its denominator, using $a = R - j - m$ and $b = R - j +1$.
For the first upper bound on the denominator,
\begin{align*}
    \Delta_a ^{-1} = \sum_{i=0}^{b-1 - a}{\frac{(i+a)!}{(a)! R^{i}}} = \sum_{i=0}^{m}{\frac{(R - j - m + i)!}{(R - j - m)! R^{i}}} =   \prod_{\ell=1}^{i}{\frac{R-j -m +\ell}{R}}.
\end{align*}
From here, we note that
\chmades{
\begin{align*}
    \prod_{\ell=1}^{i}{\frac{R-j -m +\ell}{R}}
    \leq \prod_{\ell=1}^{i}{e^{-\frac{(j+m)}{R} + \frac{\ell}{R}}}
    = e^{-\frac{(j+m)i}{R} + \frac{i(i+1)}{2R}}
    = e^{-\frac{i}{R}\left[j + m - \frac{i+1}{2} \right]}
    \leq e^{-\frac{i}{R}\left[j + \frac{1}{2}m - \frac{1}{2}\right]}.
\end{align*}
}
\chmades{
Thus, we find that 
\begin{equation}\label{eq:pstj_up}
    \prbu{a+1} \geq \left[\sum_{i=0}^{m}{e^{-\frac{i}{R}\left[j + \frac{1}{2}m - \frac{1}{2}\right]}}\right]^{-1} \geq \left[\sum_{i=0}^{\infty}{e^{-\frac{i}{R}\left[j + \frac{1}{2}m - \frac{1}{2}\right]}}\right]^{-1} = 1 - e^{-\frac{j + \frac{1}{2}(m-1)}{R}}.
\end{equation}
}

For the second upper bound on the denominator, we begin from an intermediate step and note
\chmades{
\begin{align*}
     \left[\Delta_a\right]^{-1} = \sum_{i=0}^{m}{\prod_{\ell=1}^{i}{\frac{R-j -m +\ell}{R}}}
    \leq \sum_{i=0}^{m}{\left({\frac{R-j}{R}}\right)^i}
    = \frac{R}{j}\left(1 - \left(1 - \frac{j}{R}\right)^{m+1}\right)
    \leq \frac{R}{j}.
\end{align*}
}
This in turn gives
\begin{equation}
    \prbu{a+1} \geq \frac{j}{R}.
\end{equation}

Likewise, for a lower bound on the denominator,
\begin{align*}
    \sum_{i=0}^{m}{\prod_{\ell=1}^{i}{\frac{R-j -m +\ell}{R}}} & \geq \sum_{i=0}^{m}{\left(1 - \frac{(j+m)}{R}\right)^i}
    = \frac{R}{j+m} \left(1 - \left(1- \frac{(j+m)}{R}\right)^{m+1}\right).
\end{align*}
Thus, we find that
\begin{equation}\label{eq:pstj_low}
    \prbu{a+1} \leq \frac{j+m}{R}\frac{1}{\left(1 - \left(1- \frac{(j+m)}{R}\right)^{m+1}\right)}.
\end{equation}

\subsubsection{\torpdf{Bounds on $\prbu{b-1}$.}{Bounds on pup(b-1).}}
Likewise, since $\prbu{b-1} = 1 - \delta_{b-1} = 1 - \frac{(R-j)!}{(R-j-m)! R^m}\Delta_a$, we also have that
\begin{align*}
    \left[1 -\prbu{b-1}\right]^{-1} &= \frac{R^m (R-j-m)!}{(R-j)!}\sum_{i=0}^{m}{\frac{(R - j - m + i)!}{(R - j - m)! R^{i}}}\\
    &= \sum_{i=0}^{m}{\frac{(R - j - m + i)! R^{m-i}}{(R-j)!}}\\
    &= \sum_{i=0}^{m}{\frac{(R - j - i)! R^{i}}{(R-j)!}}\\
    &\leq \sum_{i=0}^{m}{\left(\frac{R}{R-j-m}\right)^i}\\
    &= \frac{R - j - m}{j + m} \left[\left(1+ \frac{j+m}{R-j-m}\right)^{m+1} - 1 \right]\\
    &\leq \frac{R - j - m}{j + m} \left[e^{\left(\frac{(j+m)(m+1)}{R-j-m}\right)} - 1 \right].
\end{align*}
This tells us that $1- \prbu{b-1} \geq \frac{j+m}{R-j-m} \left[e^{\left(\frac{(j+m)(m+1)}{R-j-m}\right)} - 1 \right]^{-1}$, or, equivalently, 
\begin{equation*}
    \prbu{b-1} \leq 1 -  \frac{j+m}{R-j-m} \left[e^{\left(\frac{(j+m)(m+1)}{R-j-m}\right)} - 1 \right]^{-1}.
\end{equation*}
We likewise have that
\begin{align*}
     [1 -\prbu{b-1}]^{-1} &\leq \left(m+1\right)\left(\frac{R}{R-j-m}\right)^m
     \leq \left(m+1\right) e^{\frac{(j+m)m}{R-j-m}}
     \leq \left(m+1\right) e^{\frac{(j+m)(m+1)}{R-j-m}}.
\end{align*}
Combining these two results, we obtain that
\begin{equation*}
     \prbu{b-1} \leq 1 - \max\left(\frac{1}{m+1}, \frac{j+m}{R}\right) e^{-\frac{(j+m)(m+1)}{R-j-m}}.
\end{equation*}

Continuing on for the lower bound, we find that
\begin{align*}
    [1 -\prbu{b-1}]^{-1} &\geq \sum_{i=0}^{m}{\left(\frac{R}{R-j}\right)^i}
    = \frac{R-j}{j}\left[\left(1 + \frac{j}{R-j}\right)^{m+1} - 1 \right],
\end{align*}
so that,
\begin{equation*}
    \prbu{b-1} \geq 1 - \frac{j}{R-j} \left[\left(1 + \frac{j}{R-j}\right)^{m+1} - 1 \right].
\end{equation*}

%% file: ExtensionFiles/ext_pe_clm.tex
\begin{claim}\label{clm:ext_pe_clm}
    Let $p_e$ be the probability that, in an M/M/$R$ queue started with $(R+H)$ jobs in the system, the number of jobs reaches $(R+H+1)$ before it reaches $(R-m)$. Then,
    \begin{equation*}
        \frac{1}{1-p_e} \leq H + \sqrt{R}e^{\frac{m(m+1)}{R-m}} \leq \sqrt{R}\left(1 + e^{\frac{m(m+1)}{R-m}}\right).
    \end{equation*}
\end{claim}

We prove this claim by analyzing particular state transitions.
Let $T_{i\to j}$ be the first passage time from state $i$ to state $j$, and when evaluating the event $\left[T_{i\to a} < T_{i \to b}\right]$, consider the first passage times to be evaluated on the same sample path.
Then $p_e = \pr{T_{(R+H) \to (R-m)} < T_{(R+H) \to (R+H+1)}}$.
By conditioning, we have that
\begin{align*}
    1 - p_e &=  \pr{T_{(R+H) \to (R-m)} < T_{(R+H) \to (R+H+1)}}\\
    &= \pr{T_{(R+H) \to R} < T_{(R+H) \to (R+H+1)}}\left[ \pr{T_{R \to (R-m)} < T_{R \to (R+H+1)}}  \right].
\end{align*}
From here, note that
\begin{align*}
    &\pr{T_{R \to (R-m)} < T_{R \to (R+H+1)}} \\
    &\;\;= \pr{T_{R \to (R-m)} < T_{R \to (R+1)}} + \pr{T_{(R+1) \to (R-m)} < T_{(R+1) \to (R+H+1)}}\\
    &\;\;= \pr{T_{R \to (R-m)} < T_{R \to (R+1)}} + \pr{T_{(R+1) \to R} < T_{(R+1) \to (R+H+1)}}\pr{T_{R\to(R-m)} < T_{R\to (R+H+1)} }.
\end{align*}
Note also that, from basic results on simple random walks, we have that $\pr{T_{(R+H) \to R} < T_{(R+H) \to (R+H+1)}} = \frac{1}{H+1}$ and $\pr{T_{(R+1) \to R} < T_{(R+1) \to (R+H+1)}} = \frac{H}{H+1}$. 
Let $p_d \triangleq \pr{T_{R \to (R-m)} < T_{R \to (R+1)}}$.
Then we have 
\begin{align*}
    1 - p_e &= \frac{1}{H+1}\left[p_d + (1-p_d)\frac{H}{H+1}\left(p_d + \dots\right) \right]\\
    &= \frac{p_d}{H+1} \frac{1}{1 - (1-p_d)\frac{H}{H+1}}\\
    &= \frac{p_d}{H+1 - H + p_d H}\\
    &= \frac{p_d}{1 + p_d H}.
\end{align*}
This tells us that
\begin{equation*}
    \frac{1}{1-p_e} = H + \frac{1}{p_d}.
\end{equation*}

\subsubsection{\torpdf{Analyzing $p_d$.}{Analyzing pd.}} Since the M/M/$R$ queue acts equivalently to the M/M/$\infty$ queue between the states $R$ and $(R-m)$, to analyze the probability $p_d \triangleq  \pr{T_{R \to (R-m)} < T_{R \to (R+1)}}$, it suffices to analyze the appropriate hitting probability in the M/M/$\infty$.
Applying the result from Section~\ref{sec:ext_prob_bnds}, noting that the probability in question is $1 - \prbu{b-1}$ in the case where $j=0$, we find that
\begin{equation*}
    p_d \geq \max\left(\frac{1}{m+1}, \frac{m}{R} \right) e^{-\frac{m(m+1)}{R-m}} \geq \frac{1}{\sqrt{R}}  e^{-\frac{m(m+1)}{R-m}}.
\end{equation*}

\paragraph{Final bound on $(1-p_e)^{-1}$}.
Combining these and noting that $H\triangleq \min\left(\sqrt{R}, k-R\right) \leq \sqrt{R}$, we find
\begin{align*}
    \frac{1}{1-p_e} &= H + \frac{1}{p_d}
    \leq \sqrt{R} + \sqrt{R} e^{\frac{m(m+1)}{R-m}}
    = \sqrt{R}\left(1 + e^{\frac{m(m+1)}{R-m}}\right).
\end{align*}

%% file: ExtensionFiles/ext_pth_clm.tex
We now prove Claim~\ref{clm:ext_pth_clm}; see Section~\ref{sec:ext_prise_bnd_st} for details.

\paragraph{First observations.}
To begin, we first recall that $p_f$ is the probability that a given excursion is both short (that is, lasts for less than a setup time) and does not reach the appropriate threshold. 
From a union bound, it follows that $p_f \geq 1 - (\pth + \plong)$, where $\plong$ is the probability that the excursion lasts for (at least) a setup time.
We also note that the ``retrial'' probability $\prej = 1 - \gamma \pstj \geq 1 - \pstj$.
This follows from a straightforward symmetry argument: since the associated Markov chain is biased towards \textit{gaining} jobs, it follows that the probability that the number of jobs increases by one before it decreases by $m$ is higher than the probability the number of jobs {decreases} by one before it increases by $m$.
Combining these, we find that
\begin{align*}
    \pth \pstj\frac{1}{1 - \pf \prej} &\geq \frac{\pth \pstj}{1 - (1-\plong -\pth)(1 - \pstj)}\\
    &= \frac{\pth \pstj}{1 - (1-\pstj) + (\plong + \pth)(1-\pstj)}\\
    &\geq \frac{\pth \pstj}{\pstj + \plong + \pth}.
\end{align*}

\paragraph{Bounding the probabilities.}
From here, we note that, in Claim~\ref{clm:prisebnd}, we showed that $\pth \geq 0.99 \frac{j}{R}$ for any $j \geq \sqrt{R}$. 
Likewise, applying the results from Section~\ref{sec:ext_prob_bnds}, we find that $\pstj \geq \frac{j}{R}$.
To bound $\plong,$ we use the arguments from Section~\ref{sec:lbound}, which show that
\begin{align*}
    \plong &\leq 1 - \left(1- \frac{j}{R}\right) \left(1 - \frac{b_1}{\sqrt{\mu \beta R}}\right)
    \leq \frac{j}{R} + \frac{b_1}{\sqrt{\mu \beta R}}
    \leq \frac{j}{R}\left( 1 + \frac{b_1}{\sqrt{\mu \beta}}\right).
\end{align*}

\paragraph{Completion of the proof of Claim~\ref{clm:ext_pth_clm}.}
Combining these, we find that
\begin{align*}
    \frac{\pth \pstj}{1 - \pf \prej} &\geq \frac{0.99 \frac{j}{R} \cdot \frac{j}{R}}{\frac{j}{R} + 0.99 \frac{j}{R} + \frac{j}{R} \left(1 + \frac{b_1}{\sqrt{\mu \beta}}\right)}\ =\frac{j}{R} \left(\frac{0.99}{2.99 + \frac{b_1}{\sqrt{\mu \beta}}}\right)
    \geq \frac{j}{R} (0.32).
\end{align*}